\documentclass[12pt,
a4paper,
oneside,
openany,
english,
toc=bib, 
toc=listofnumbered,
numbers=noendperiod]{scrartcl}
\usepackage[german,english]{babel}
\usepackage[utf8]{inputenc}
\usepackage[T1]{fontenc}
\usepackage{graphicx}
\usepackage{caption}
\usepackage{subcaption}
\usepackage[rightcaption]{sidecap}
\usepackage{lmodern, dsfont}
\usepackage{fancyhdr}
\usepackage{setspace}
\usepackage{color}
\usepackage{amsmath}
\usepackage{amssymb} 
\usepackage{amsthm}
\usepackage{listings}
\usepackage{textgreek} 
\usepackage[table]{xcolor}
\usepackage{setspace}
\usepackage{url}
\usepackage[pdftex,
            pdfauthor={Marius Krumm},
            pdftitle={Master Thesis},
            pdfsubject={Quantum Foundations, Quantum Information Theory, Quantum Thermodynamics},
            pdfkeywords={Quantum, Generalized Probabilistic Theories, GPT, Foundations, Information Theory, Entropy, Thermodynamic, Postulates},
            pdfproducer={Latex},
            pdfcreator={pdflatex}]{hyperref}
\usepackage{hyperref}
\usepackage{bbm}
\usepackage{braket}
\usepackage{pdfpages} 

\numberwithin{figure}{section}
\numberwithin{equation}{section}
\numberwithin{table}{section}

\definecolor{mygreen}{rgb}{0,0.6,0}
\definecolor{mygray}{rgb}{0.5,0.5,0.5}
\definecolor{mymauve}{rgb}{0.58,0,0.82}

\graphicspath{{graphs/}}

\KOMAoptions{BCOR=10mm}

\let\oldsection\section
\def\section{\cleardoublepage\oldsection}

\newtheorem{prop}{Proposition}[section]
\newtheorem{theorem}[prop]{Theorem}

\newtheorem{defi}[prop]{Definition}
\newtheorem{exam}[prop]{Example}
\newtheorem{cor}[prop]{Corollary}
\newtheorem{lem}[prop]{Lemma}
\newenvironment{postulate}[2][Postulate]{\begin{trivlist}
\item[\hskip \labelsep {\bfseries #1}\hskip \labelsep {\bfseries #2}]}{\end{trivlist}}
\newtheorem*{note}{Comment}
\newtheorem*{ass}{Assumption}

\newcommand{\be}{\begin{equation}}
\newcommand{\ee}{\end{equation}}

\DeclareFontFamily{U}{mathx}{\hyphenchar\font45}
\DeclareFontShape{U}{mathx}{m}{n}{
      <5> <6> <7> <8> <9> <10>
      <10.95> <12> <14.4> <17.28> <20.74> <24.88>
      mathx10
      }{}
\DeclareSymbolFont{mathx}{U}{mathx}{m}{n}
\DeclareFontSubstitution{U}{mathx}{m}{n}
\DeclareMathAccent{\widecheck}{0}{mathx}{"71}
\DeclareMathAccent{\wideparen}{0}{mathx}{"75}

\begin{document}
\selectlanguage{english}


\thispagestyle{empty}
\begin{center}
  \renewcommand{\baselinestretch}{2.00}
  \Large\sffamily
  Department of Physics and Astronomy\\
  \large
  University of Heidelberg
  \par\vfill\normalfont
  Master thesis\\
  in Physics\\
  submitted by\\
  Marius Krumm\\
  born in Worms\\
  2015\\
\end{center}
\newpage
\thispagestyle{empty} \
\newpage

\thispagestyle{empty}
\begin{center}
  \renewcommand{\baselinestretch}{2.00}
  \large\bfseries\sffamily
	Thermodynamics and the Structure of Quantum Theory\\ as a Generalized Probabilistic Theory
  \par
  \vfill
  \large\normalfont
  This Master thesis has been carried out by Marius Krumm\\
  at the\\
  Institute for Theoretical Physics\\
  under the supervision of\\
  Dr. rer. nat. Markus Müller\\
\end{center}\par
\vspace{5\baselineskip}

\renewcommand{\baselinestretch}{1.00}\normalsize

\newpage
\thispagestyle{empty} \
\newpage


\thispagestyle{empty}
\begin{center}
  \begin{minipage}[c][0.48\textheight][b]{0.9\textwidth}
    \small
    \textbf{
     Thermodynamik und die Struktur der Quantentheorie als eine verallgemeinerte probabilistische Theorie:
    }\par
    \vspace{\baselineskip}
    Diese Arbeit untersucht den Zusammenhang zwischen Quantentheorie, Thermodynamik und Informationstheorie. Es werden Theorien betrachtet, welche eine der Quantentheorie ähnliche Struktur besitzen und im ``Generalized Probabilistic Theories'' genannten Framework beschrieben werden. Einem Vorschlag von J. Barrett \cite{Conference} folgend wird ein Gedankenexperiment von von Neumann\cite{Neumann} adaptiert um eine natürliche thermodynamische Entropie-Definition zu erhalten. Einige mathematische Eigenschaften dieser Entropie werden physikalischen Konsequenzen des Gedankenexperiments gegenüber gestellt. Die Gültigkeit des zweiten Hauptsatzes der Thermodynamik wird untersucht. In diesem Kontext werden auch Observablen und projektive Messungen verallgemeinert, um einen Entropie-Zuwachs in projektiven Messungen von Ensembles zu beweisen. Informationstheoretisch motivierte Definitionen der Entropie, welche in \cite{WehnerEntropy}\cite{BarnumEntropy} eingeführt wurden, werden mit der thermodynamisch motivierten Definition der Entropie verglichen. Die Bedingungen für die Wohldefiniertheit der Entropie werden genauer analysiert. Es werden einige weitere Eigenschaften der behandelten Theorien (z.B. Frage nach Interferenz höherer Ordnung, Pfisters Zustandsunterscheidungsprinzip \cite{Pfister}) und deren Zusammenhang mit der Entropie untersucht.
  \end{minipage}\par
  \vfill
  \begin{minipage}[c][0.48\textheight][b]{0.9\textwidth}
    \small
    \textbf{
     Thermodynamics and the Structure of Quantum Theory as a Generalized Probabilistic Theory:
    }\par
    \vspace{\baselineskip}
    This thesis investigates the connection between quantum theory, thermodynamics and information theory. Theories with structure similar to that of quantum theory are considered, mathematically described by the framework of ``Generalized Probabilistic Theories''. For these theories, a thought experiment by von Neumann \cite{Neumann} is adapted to obtain a natural thermodynamic entropy definition, following a proposal by J. Barrett \cite{Conference}. Mathematical properties of this entropy are compared to physical consequences of the thought experiment. The validity of the second law of thermodynamics is investigated. In that context, observables and projective measurements are generalized to prove an entropy increase for projective measurements of ensembles. Information-theoretically motivated definitions of the entropy introduced in \cite{WehnerEntropy}\cite{BarnumEntropy} are compared to the entropy from the thermodynamic thought experiment. The conditions for the thermodynamic entropy to be well-defined are considered in greater detail. Several further properties of the theories under consideration (e.g. whether there is higher order interference, Pfister's state discrimination principle \cite{Pfister}) and their relation to entropy are investigated.
  \end{minipage}
\end{center}

\thispagestyle{empty}
\quad
\newpage

\thispagestyle{empty}
\tableofcontents
\clearpage


\setcounter{page}{1}
\pagenumbering{arabic}

\section{Introduction}
While quantum theory exists for roughly 100 years, it still remains mysterious. Many people have pondered about quantum theory, asking for the true reality, hidden determinism,... without coming to a real conclusion. Many interpretations with identical predictions have appeared \cite{Interpretations}\cite{Interpretations2}, therefore it is not possible to find the ``right'' interpretation by experiment. Despite these fundamental conceptional issues, quantum theory is extremely successful in experiment and technology. As summarized by Mermin's famous sentence ``shut up and calculate''\cite{ShutUpAndCalculate}, a large part of the scientific community has turned away from the foundations of quantum physics, and prefers to apply quantum physics to concrete physical systems. An important reason for this decision is that the mathematical formalism provided by quantum theory can be used without understanding its origin. Another important reason is, that many attempts to think about quantum theory remain very vague or appear helpless. Explanations in the style of collapsing electron clouds, pilot waves, infinitely many realities,... often seem to miss the point, overcomplicating quantum theory without really solving its conceptional problems. Sometimes explanations motivated by classical intuition are even in contradiction to results from standard quantum theory, especially Bell's Theorem \cite{Bell}.\\ 
\\
Thus for a very long time, many people have lost interest in asking complicated questions about the foundations of quantum physics that seemingly cannot be answered anyway.\\
This has changed when a new approach reached the field of quantum foundations \cite{Hardy}: The rise of quantum information theory showed that it is fruitful to take an operational/information-theoretic approach to think about physics. This approach is inspired by both relativity and quantum field theory in a very general sense: Relativity gives the observer a fundamental role in reality, as indicated by the famous statement ``everything is relative''. Quantum field theory and particle physics care very much about producing results like cross sections and correlation functions that can be measured in experiment; especially hypothetical particles and fields without any interaction are excluded because their existence makes no difference  (it is this way in which such particles are not real).\\
The operational/information-theoretic approach also assumes the point of view ``real is what can be performed or measured''. Preparations, transformations and measurements are basic notions and are combined into a strict mathematical framework known as \textit{Generalized Probabilistic Theories} (GPTs). Historically, parts of the tools and the formalism as well as the idea to reaxiomatize quantum theory come from quantum logic\cite{Mackey}\cite{Logic}. The most important success of the GPT framework is that it allows to replace the vague attempts to derive quantum theory by postulates that are both mathematically precise and motivated by physical or information theoretic ideas. The idea to motivate postulates by considerations about computation and information\cite{Masanes}\cite{Barrett} comes from the close connection of the framework to quantum information theory.\\
While such questions are unusual for physicists, they provide many advantages: Information and state processing can be experimentally demonstrated in a laboratory, thus the corresponding postulates often can be tested. Furthermore a connection between information theory and physics is provided by the notion of entropy. The black hole information paradox and Landauer's principle indicate that this link might be very deep, as summarized by Landauer's statement ``\textit{information is physical}'' or Wheeler's ``\textit{it from bit}'' \cite{itfrombit}. Another hint can be found in quantum teleportation: If there was no need for classical communication, then the teleportation would happen faster than the speed of light (see e.g. \cite{BestBookEver}). This implies that information is a fundamental part of reality, not just some data irrelevant for the physical processes. Thus it might be that the unsatisfying attempts to explain quantum theory might be caused by asking the wrong questions, neglecting a possible fundamental importance of information and its relevance for observers, measurements and transformations.\\
\\
This thesis explores the connection between quantum theory, thermodynamics and quantum information theory. Two postulates, one of them motivated by physics and the other one by information theory, are used to provide strong structural properties for theories. For these theories, which include quantum and classical physics and many more (e.g. quaternionic quantum theory), an old thought experiment by von Neumann \cite{Neumann} is used to derive a von Neumann-like thermodynamic entropy, following a suggestion by J. Barrett\cite{Conference}. Many properties of this entropy are proven. The postulates, supplemented by two other postulates, already have been used to derive quantum theory (in finite dimensions)\cite{Postulates} and thus suggest a deep connection between quantum theory and thermodynamics. Another point of view is that these postulates summarize structural properties of quantum theory relevant for thermodynamics into two well-defined mathematical statements.\\ 
\\
The thesis is structured as follows:
At first, some basic definitions and results from convex geometry are explained in Chapter \ref{Section:ConvexGeometry}. These are necessary to introduce the mathematical framework called ``\textit{Generalized Probabilistic Theories (GPT)}'' in Chapter \ref{Section:GPT}. This framework allows to derive quantum theory by using exact mathematically well-defined postulates. Afterwards in Chapter \ref{Section:Postulates}, the main postulates of this thesis are introduced and motivated. Then in Chapter \ref{Section:Neumann}, the thought experiment by von Neumann which derives the von Neumann entropy is presented. We apply this thought experiment to our GPTs to find a corresponding notion of entropy. In Chapter \ref{Section:Petz}, a generalized version of the thought experiment is presented, which is more general but less elegant. From this experiment follows an important property of the entropy, whose validity is checked for several theories. In Chapter \ref{Section:ProjectiveMeasurements} we generalize the projective measurements known from quantum theory to our GPTs. They can be used to describe the semi-permeable membranes used in the thought experiments. Furthermore, it is possible to prove the second law of thermodynamics for these measurements, which is done in Chapter \ref{Section:2ndLaw}. Also the second law is checked in mixing procedures. As most of our considerations so far have been from a thermodynamic point of view, in Chapter \ref{Section:InformationTheory} we analyze the entropy from an information-theoretic/operational point of view. Here, measurement and decomposition entropies \cite{BarnumEntropy}\cite{WehnerEntropy} will be introduced and compared to the thermodynamic entropy. Furthermore, the same is done for the R\'{e}nyi entropies. We also investigate the question whether there is third order interference and relate it to the entropies. Then in Chapter \ref{Section:Spectrality}, an example for a state space is constructed to show, that our first postulate alone is not enough for a well-defined entropy. In Chapter \ref{Section:StateDiscrimination}, we will show that a principle called \textit{state discrimination principle} (introduced by Corsin Pfister) holds in all theories considered by us. At last, an outlook will be given in Chapter \ref{Section:Outlook}. 
\section{Convex geometry}
\label{Section:ConvexGeometry}
Our first goal is to introduce the framework called \textit{Generalized Probabilistic Theories} (GPT). This framework allows to discuss many probabilistic theories, including quantum theory and classical probability theory. The basic framework is very natural, relying only on very weak assumptions. These assumptions lead to convex sets and convex-linear maps. Therefore, basic notions and results from convex geometry are necessary to understand the GPT-framework.\\
As convex geometry usually is not part of the physics curriculum, we will provide a short introduction here. It is mainly based on \cite{Pfister} and \cite{Webster}. In case the reader needs more examples and applications, one should take a look at \cite{Pfister}.\\
\\ 
Convex sets contain all straight connection lines between points taken from these sets:
\vspace{1mm}
\begin{defi}
	Let $V$ be a real vector space. A subset $C \subset V$ is called \textbf{convex} if for all $v,w \in C$ and $p\in [0,1]$ also $pv+(1-p)w \in C$. 
\end{defi}
\vspace{1mm}
This definition directly extends to more states:
\vspace{1mm}
\begin{prop}
	Let $V$ be a real vector space, $C$ a convex subset. For any $p_1,...,p_n \ge 0$, $\sum_j p_j =1$, $v_1,...,v_n \in C$ we find $\sum_j p_j v_j \in C$.
\end{prop}

\begin{proof}
	We prove by induction. For $n=1$ , there is nothing to show. So assume now the statement to be true for $n$. Wlog we assume $p_j > 0$ $\forall j$.\\
	We rewrite:
	\begin{equation}
		\sum_{j=1}^{n+1} p_j v_j =\sum_{k=1}^n p_k\left[ \sum_{j=1}^n \frac{p_j}{\sum_{a=1}^n p_a} v_j \right]+ p_{n+1} v_{n+1}
	\end{equation}
	By the induction hypothesis, we find:
	\begin{equation}	
	  \sum_{j=1}^n \frac{p_j}{\sum_{a=1}^n p_a} v_j \in C
	\end{equation}
	  By the definition of convex thus $\sum_{j=1}^{n+1} p_j v_j \in C$.
\end{proof}
\vspace{1mm}
However, not all points of convex sets are found in the interior of straight lines. The counter-examples will play an important role and thus deserve a name:
\vspace{1mm}
\begin{defi}
 	A point $x$ of a convex set $C$ is called an \textbf{extreme point} of $C$ if for all $p \in (0,1)$ and $v,w\in C$ with $x = pv+(1-p)w$ we find $v=w=x$. The set of extreme points is called $\rm{ext}(C)$.
\end{defi}
\vspace{1mm}
Examples for convex sets are cubes and balls. The extreme points of a cube are its corners, while all surface points of a ball are extreme points. More examples can be found in Figure \ref{Fig:ConvexSet}.

\begin{center}
\includegraphics[width= 0.7\textwidth]{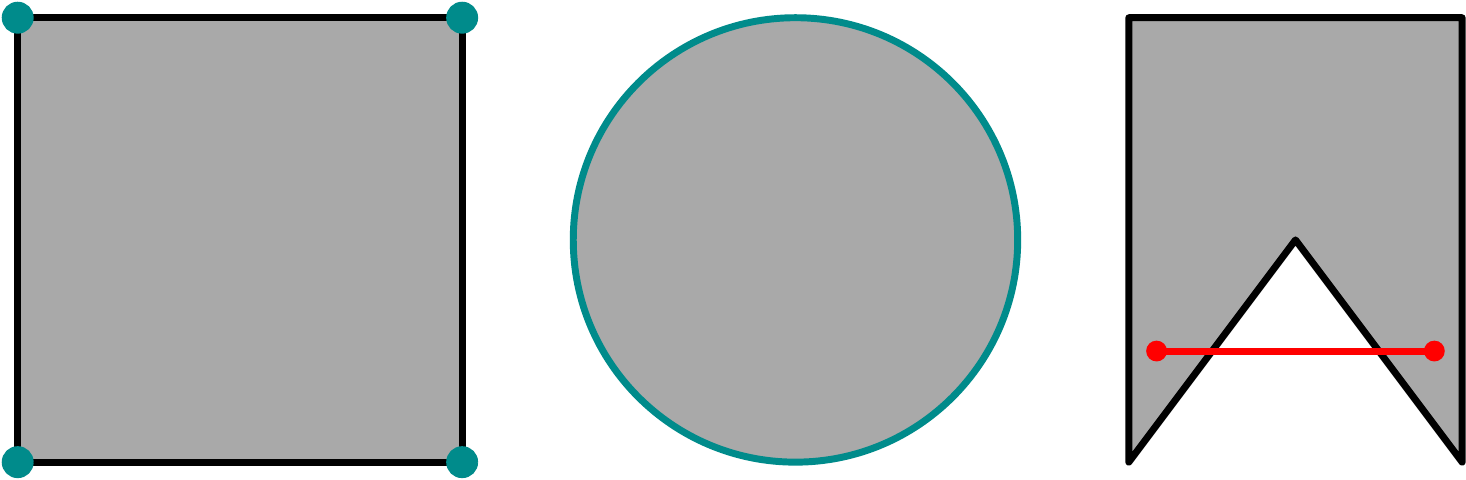}
\captionof{figure}{\textit{\small The square and the circle are examples for convex sets. For a square, only the corners are extreme points, while for the circle, all boundary points are extreme points. The third set is not convex: The red line connecting two points of the set is not fully contained in the set.}}
\label{Fig:ConvexSet}
\end{center}

\vspace{1mm}
Next we define faces. Faces are ``maximal'' planar surface parts, e.g. the sides of a cube (see also Figure \ref{Fig:Faces}):
\vspace{1mm}
\begin{defi}
	A nonempty convex subset $F$ of a convex set $C$ is called a \textbf{face} of $C$, if for all $v,w \in C$ and $p\in (0,1)$ with $pv+(1-p)w \in F$ we find $v\in F$, $w\in F$.  
\end{defi}

\vspace{1mm}

\begin{center}
\includegraphics[width= 0.5\textwidth]{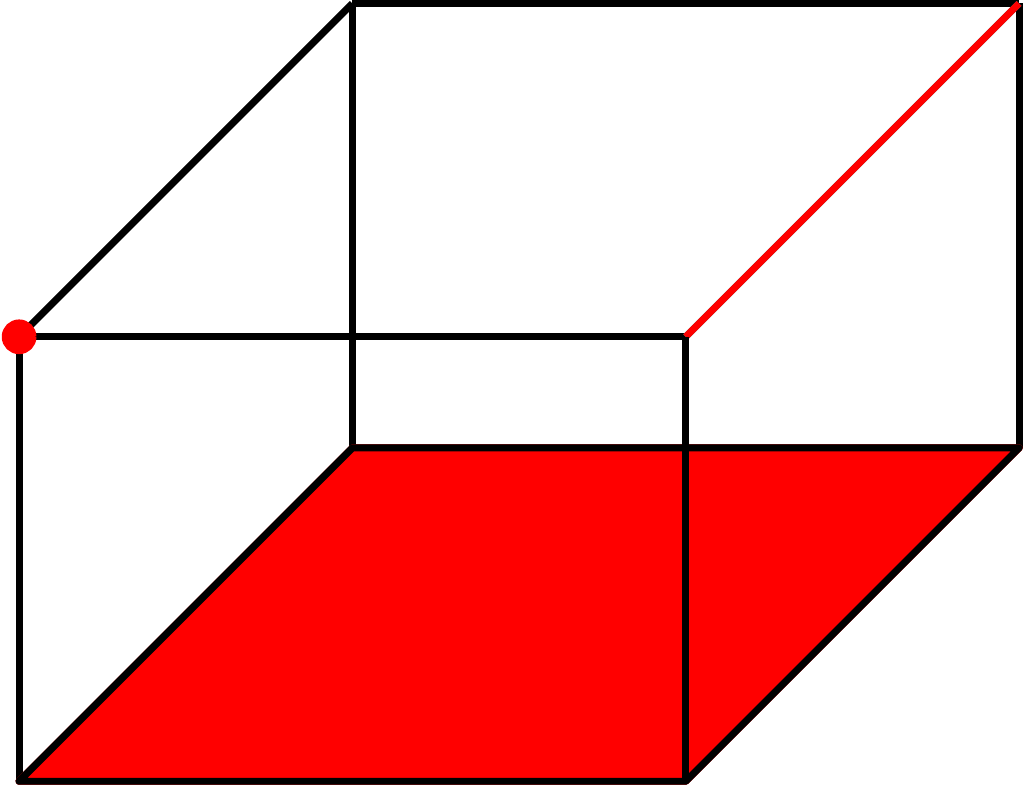}
\captionof{figure}{\textit{\small The faces of a cuboid are its corners, its edges, its rectangles and the cuboid itself.}}
\label{Fig:Faces}
\end{center}

\begin{lem}
	Let $x$ be an extreme point of a convex set $C$. Then $\{x\}$ is a face of $C$. 
\end{lem}

\begin{proof}
	$\{x\}$ is not empty, $x=px+(1-p)x$ i.e. $\{x\}$ is convex. $x=pv+(1-p)w$ for $p\in (0,1)$ implies $x=v=w$ because $x$ is extreme.
\end{proof}

\vspace{1mm}

\begin{lem}
	Let $F$ be a face of a convex set $C$. Let $v_1,...,v_n \in C$ and $p_1,...,p_n > 0$ with $\sum_{j=1}^n p_j=1$ be such that $\sum_{j=1}^n p_j v_j \in F$. Then for all $j$, $v_j\in F$
\end{lem}

\begin{proof}
	We prove by induction. For $n=2$, the statement follows immediately from the definition. So assume now that the statement is true for $n$ states, $n\ge 2$.\\
	We rewrite: 
	\begin{equation}
		\sum_{j=1}^{n+1} p_j v_j =\sum_{k=1}^n p_k\left[ \sum_{j=1}^n \frac{p_j}{\sum_{a=1}^n p_a} v_j \right]+ p_{n+1} v_{n+1}
	\end{equation}
	By definition of a face, 
	\begin{equation}		
		\sum_{j=1}^n \frac{p_j}{\sum_{a=1}^n p_a} v_j \in F
	\end{equation}	
	 and $v_{n+1} \in F$. By the induction hypothesis $v_j \in F$ for $1\le j \le n$.
\end{proof}

\vspace{1mm}

\begin{cor}
	Let $x$ be an extreme point of a convex set $C$ and $v_1,...,v_n \in C$ and $p_1,...,p_n > 0$ with $\sum_{j=1}^n p_j=1$ be such that  $\sum_{j=1}^{n} p_j v_j = x$. Then $v_j = x \ \forall j$.
\end{cor}

\vspace{1mm}

\begin{defi}
	Let $M$ be a subset of a real vector space $V$. The \textbf{convex hull} of $M$, $\rm{conv}(M)$, and the \textbf{affine hull} of $M$, $\rm{aff}(M)$,  are defined as 
	\begin{align}
		\rm{conv}(M) := \left\{\sum_{j=1}^n p_j w_j \ | \ n\in \mathbb N,v_j \in M, \ p_j \ge 0 \text{ with }\sum_{j=1}^n p_j=1\right\}\\
		\rm{aff}(M) := \left\{\sum_{j=1}^n p_j w_j \ | \ n\in \mathbb N,v_j \in M, \ p_j \in \mathbb R \text{ with } \sum_{j=1}^n p_j=1\right\}
	\end{align}
	Terms of the form $\sum_{j=1}^n p_j w_j$ with $p_j \ge 0 \text{ and }\sum_{j=1}^n p_j=1$ are called \textbf{convex (linear) combinations} . They are called \textbf{affine (linear) combinations} if $p_j \in \mathbb R$ with $\sum_{j=1}^n p_j=1$.
\end{defi}

\vspace{1mm}

\begin{prop}
	If $F$ is a face of a convex set $C$, then $F=\rm{aff}(F) \cap C$.
\end{prop}

\begin{proof}
	Trivially, $F\subset \rm{aff}(F)$ and $F \subset C$, thus $F\subset \rm{aff}(F) \cap C$.
	Much harder to show is $\rm{aff}(F) \cap C \subset F$:\\
	So let $v :=\sum_j p_j w_j \in \rm{aff}(F) \cap C$. We relabel such that for $j\le m$, $p_j \ge 0$ and for $j > m$, $p_j < 0$. We assume $m < n$, otherwise $v \in F$ as $v$ is given by a convex combination in $F$. Thus:
	\begin{equation}
		v +\sum_{j=m+1}^n |p_j|w_j = \sum_{j=1}^m |p_j| w_j
	\end{equation}
	Because of $\sum_{j=1}^n p_j = 1$, there is at least one $p_j > 0$. Therefore:
	\begin{equation} \label{eq:FaceMultipleStatesProof}
		\frac{1}{\sum_{j=1}^m |p_j|}v +\sum_{j=m+1}^n \frac{|p_j|}{\sum_{k=1}^m |p_k|}w_j = \sum_{j=1}^m \frac{|p_j|}{\sum_{k=1}^m |p_k|} w_j
	\end{equation}
	The expression on the right-hand side is a convex combination of states in F. Thus by convexity of $F$:
	\begin{equation}	 
	 \sum_{j=1}^m \frac{|p_j|}{\sum_{k=1}^m |p_k|} w_j \in F
	\end{equation}	 
	 Using 
	\begin{align}
	\frac{1}{\sum_{j=1}^m |p_j|}+& \sum_{j=m+1}^n \frac{|p_j|}{\sum_{k=1}^m |p_k|} = \frac{1+\sum_{k=m+1}^n |p_k|}{\sum_{k=1}^m |p_k|}=\frac{\sum_{k=1}^n p_k+\sum_{k=m+1}^n |p_k|}{\sum_{k=1}^m |p_k|} =1
	\end{align} we see that also the left-hand side of Equation (\ref{eq:FaceMultipleStatesProof}) is a convex combination of elements in $C$. As $F$ is a face and $v\in C$, $v \in F$.
\end{proof}

\vspace{1mm}

\begin{defi}
	For a convex set $C$, let $M \subset C$. The \textbf{face generated by} $M$ is defined as the \textbf{minimal face} containing $M$:
	\begin{equation}
		F := \bigcap_{H \subset C \text{ face, } M \subset H} H
	\end{equation}
\end{defi}

\vspace{1mm}

\begin{prop}
	For a convex set $C$, let $M \subset C$. The face $F$ generated by $M$ is indeed a face. If $G$ is another face of $C$ containing $M$, then $F\subset G$.
\end{prop}

\begin{proof}
	The last statement is clear by definition of $F$ as intersection of all faces containing $M$.\\
	So it remains to show, that $F$ is indeed a face:\\
	$F$ is not empty because $C$ itself is a face which contains $M$.
	For $\{v_k\} \subset F$, $\{p_k\}$ a probability distribution, $\{v_k\}$ is also found in all faces containing $M$ by definition of $F$. As all faces are convex, all these faces also contain $\sum_k p_k v_k$. By definition of $F$ as an intersection, also $\sum_k p_k v_k \in F$. Thus $F$ is convex.\\
	Now let $w = pv_1+(1-p)v_2 \in F$ with $v_1,v_2 \in C$ and $0<p<1$. Every face containing $M$ also contains $w$ and thus also $v_1,v_2$ because they are faces. By definition of $F$ as intersection of all these faces, also $v_1,v_2\in F$.
\end{proof}

\vspace{1mm}
The importance of the extreme points is that the extreme points generate (compact) convex sets, as shown by the famous Krein-Milman theorem (see e.g. \cite{Werner} Theorem VIII.4.4):
\vspace{1mm}

\begin{theorem}[Krein-Milman]
	Let $V$ be a locally convex topological vector space (Hausdorff) and $C$ a compact convex subset of $V$. Then \begin{equation} C = \overline{\rm{conv}(\rm{ext}(C))}\end{equation}
\end{theorem}

\vspace{1mm}
In finite dimension, there is a simpler version by Minkowski (see e.g. \cite{Webster} Theorem 2.6.16, \cite{Pfister}):
\vspace{1mm}

\begin{theorem}
	Let $V$ be a real finite-dimensional vector space and $C$ a compact convex subset of $V$. Then \begin{equation} C = \rm{conv}(\rm{ext}(C)) \end{equation}
\end{theorem}

\vspace{1mm}
Note that the finite-dimensional version is much simpler, all the topological difficulties are gone. Later on, we will restrict ourselves to finite dimension in order to not obscure the physics by topological technicalities. 

Now we consider maps that preserve the convexity structure.
\vspace{1mm}

\begin{defi}
	A map $f:V\rightarrow W$ between finite-dimensional vector spaces is called \textbf{convex-linear} if for all $p\in [0,1]$, $x,y \in V$ we have $f(px+(1-p)y) = pf(x)+(1-p)f(y)$.\\
	It is called \textbf{affine-linear} if $f(px+(1-p)y) = pf(x)+(1-p)f(y)$ for all $p\in \mathbb R$, $x,y \in V$.
\end{defi}

\vspace{1mm}

\begin{prop}
	Every convex-linear map is affine-linear.
\end{prop}

\begin{proof}
	We have to check $f(px+(1-p)y) = pf(x)+(1-p)f(y)$ for $x,y \in V$, $p \in \mathbb R$. For $p \in [0,1]$ this is clear by convexity. Now assume $p > 1$:\\
	Then we have to show $f(x) = \frac{1}{p} f(px+(1-p)y) + \frac{1-p}{-p}f(y)$. As $p > 1$ we find $1-p < 0$. Especially, $0 < \frac{1}{p}<1$ and $0< \frac{1-p}{-p} < 1$. By convexity
	\begin{equation}
		\frac{1}{p} f(px+(1-p)y) + \frac{1-p}{-p}f(y) = f\left(x +\frac{1-p}{p}y + \frac{1-p}{-p}y\right) = f(x)
	\end{equation}
	as we had to show. The case $p < 0$ is equivalent to $1-p > 1$ and thus can be proved like the case before (by exchanging the roles of $x$ and $y$).
\end{proof}

\vspace{1mm}

\begin{prop}
	Let $f: V \rightarrow W$ be a convex- or affine-linear map. Then for any $x_1,...,x_n \in V$, $p_1,...,p_n \in \mathbb R$ with $\sum_{j=1}^n p_j = 1$ we have: 
	\begin{equation}
		f\left(\sum_{j=1}^n p_j x_j\right) = \sum_{j=1}^n p_j f(x_j)
	\end{equation}
\end{prop}

\begin{proof}
	We can assume that $f$ is affine-linear. We use a proof by induction. For $n=2$, the statement is clear by affine-linearity. Now assume that the statement is true for $n \ge 2$:\\
	Let the $p_j$ be labelled such that $\sum_{j=1}^n p_j \ne 0$.We rewrite: 
	\begin{align}
	f\left(\sum_{j=1}^{n+1} p_j v_j\right) &=f\left(\sum_{k=1}^n p_k\left[ \sum_{j=1}^n \frac{p_j}{\sum_{a=1}^n p_a} v_j \right]+ p_{n+1} v_{n+1}\right)\\
		&= \sum_{k=1}^n p_k f\left( \sum_{j=1}^n \frac{p_j}{\sum_{a=1}^n p_a} v_j \right)+ p_{n+1} f(v_{n+1}) \\
		&= \sum_{k=1}^n p_k \sum_{j=1}^n \frac{p_j}{\sum_{a=1}^n p_a} f(v_j)+ p_{n+1} f(v_{n+1})
	\end{align}
	where we first used affine-linearity and then the induction hypothesis.
\end{proof}

\vspace{1mm}

\begin{prop}
	A map $f: V \rightarrow W$ between finite-dimensional real vector spaces is affine- (or convex-) linear exactly if it is of the form
	\begin{equation}
		f(\cdot) = L(\cdot)+y
	\end{equation} for some linear map $L:V\rightarrow W$ and some $y\in W$.
	\label{Proposition:LinearConvexLinear}
\end{prop}

\begin{proof}
Here, we only provide a sloppy proof sketch:
	\begin{align}
		\frac{\partial f(x)}{\partial x_j} &= \lim_{h \rightarrow 0} \frac{1}{h} \big[ f(x+he_j)-f(x) \big] = \lim_{h \rightarrow 0} \frac{2}{h} \Big[ \frac{1}{2}f(x+he_j)+1\cdot f(0)-\frac{1}{2}f(x)- 1\cdot f(0)\Big]\\
			&= \lim_{h \rightarrow 0} \frac{2}{h} \Big[f\left( \frac{1}{2} (x+he_j)+1\cdot 0 - \frac{1}{2} \cdot x \right)-f(0)\Big] = \lim_{h \rightarrow 0} \frac{2}{h} \left[f\left( \frac{1}{2} \cdot h e_j\right)-f(0)\right]
	\end{align}
	Here, $e_j = (\delta_{kj})_k$ is the vector filled with zeroes except for the $j$-th component which is a $1$.\\
	Thus the partial derivatives are constant. Therefore, ``$f=$ linear map $+$ constant''.\\
	See also Theorem 1.5.2 from \cite{Webster} for a more detailed proof.
\end{proof}

\vspace{1mm}

\begin{defi}
	Let $V$ be a real vector space. A non-empty subset $K$ of $V$ is a \textbf{cone} if the following conditions are satisfied \cite{Cones}:
	\begin{enumerate}
		\item $K+K\subset K$
		\item $pK \subset K \ \forall p \ge 0$
		\item $K \cap (-K) = \{0\}$
	\end{enumerate}
\end{defi}

\vspace{1mm}
A typical example for a cone is an infinitely long ice-cream cone. Another typical example, which looks like an infinitely long pyramid turned upside-down, is shown in Figure \ref{Fig:Cone}.

\begin{center}
\includegraphics[width= 0.3\textwidth]{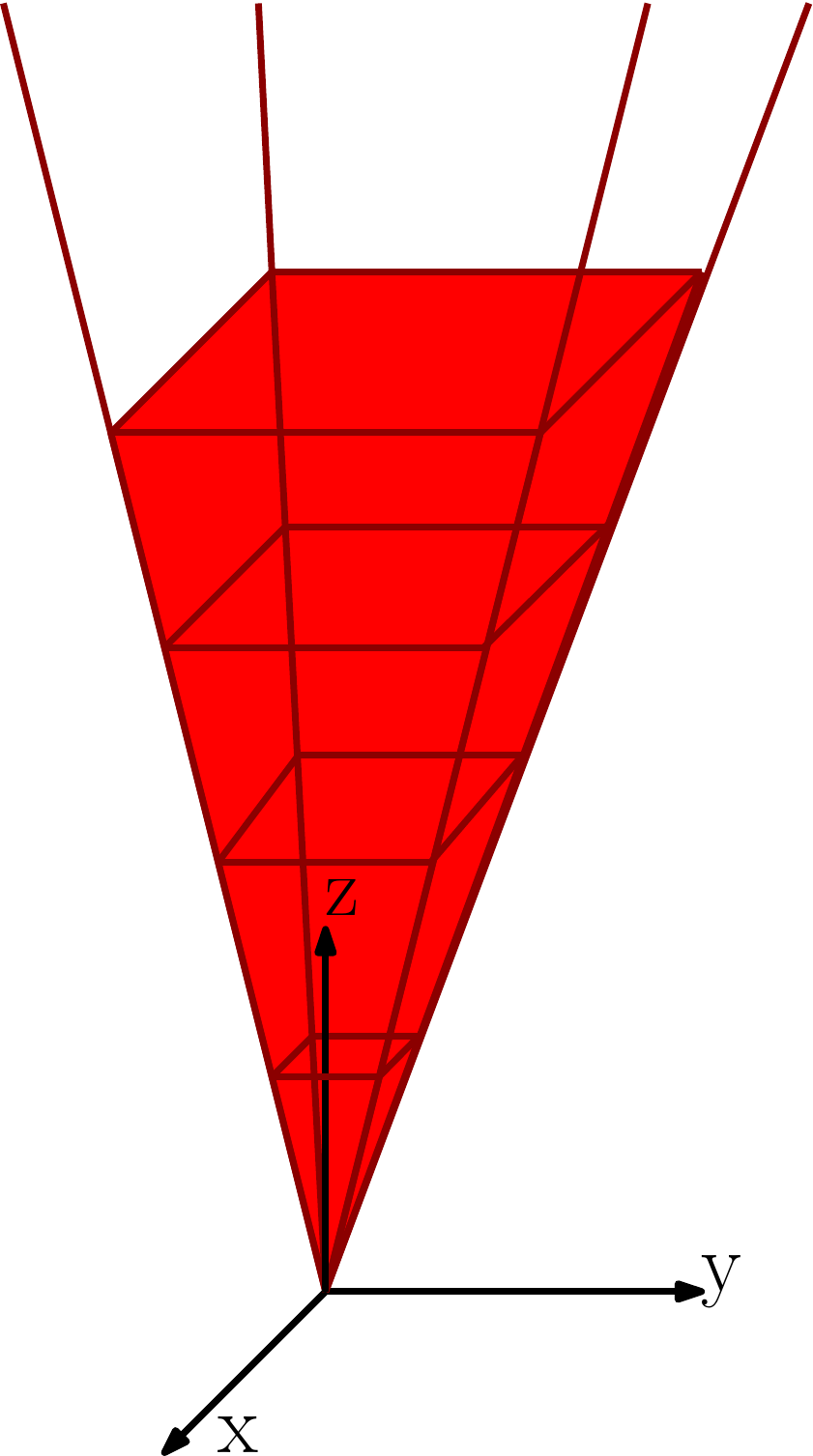}
\captionof{figure}{\textit{\small A typical cone. Note that cones are infinitely long, as indicated in the figure.}}
\label{Fig:Cone}
\end{center}  

\begin{prop}
	Let $K$ be a cone. Then $\rm{span}(K) = K-K$.
\end{prop}

\begin{proof}
	$K-K \subset \rm{span}(K)$ is clear by the definition of $\rm{span}(K)$.\\
	Let $\sum_{j=1}^n p_j v_j \in \rm{span}(K)$ with $v_j \in K$, $p_j \in \mathbb R$. By relabelling, we assume $p_j \ge 0$ for $j\le m \le n$ and $p_j < 0$ for $j>m$. Thus $\sum_{j=1}^n p_j v_j = \sum_{j=1}^m |p_j| v_j - \sum_{j=m+1}^n |p_j| v_j$. Thus if we can show that terms of the form $\sum_{j=m}^n q_j w_j$ are in $K$ for $q_j \ge 0$ and $w_j \in K$, then we find $\sum_{j=1}^n p_j v_j \in K-K$ and $\rm{span}(K) \subset K-K$ in total.\\
	By the second property of cones, $q_j w_j \in K$. Thus by the first property of cones, $\sum_{j=m}^n q_j w_j \in K$.
\end{proof}

\vspace{1mm}

\begin{defi}
	A cone $K$ of a real vector space $V$ is called \textbf{generating} if $\rm{span}(K) = K-K= V$.
\end{defi}

\vspace{1mm}

\begin{note}
	Sometimes, the definition of cones varies in the literature. For example in \cite{Ududec}, the condition $K \cap (-K) = \{0\}$ is not necessary for a set to be called cone. There, cones that satisfy $K \cap (-K) = \{0\}$ are called \textbf{pointed}. But in \cite{Ududec}, all cones are required to be generating.
\end{note}

\vspace{1mm}

\begin{defi}
	Given a cone $K \subset V$, an \textbf{order unit} $u_K$ is an element of $V^* = \{ f:V\rightarrow \mathbb R \text{ linear} \}$, which is \textbf{strictly positive} on all non-zero elements of the cone, i.e. 
	\begin{equation}
		u_K (v) > 0 \ \ \forall v \in K\backslash \{0\}.
	\end{equation} 
\end{defi}
\section{Framework: Generalized Probabilistic Theories}
\label{Section:GPT}
\subsection{The state space}
In this chapter, we provide an introduction to the framework called \textit{Generalized Probabilistic Theories (GPTs)}. This framework includes a wide range of physical theories, including classical and quantum theory. It is very general, starting from the idea that theories should specify measurement probabilities, adding only weak and natural assumptions. Other assumptions and postulates can be added while constructing a specific theory. Thus this framework allows us to use mathematically well-defined postulates instead of vague motivations to single out quantum theory. Many other sources also give introductions to this framework, but often there are slightly different points of view or approaches \cite{Pfister}\cite{Ududec}\cite{Barrett}\cite{Fuchs}\cite{PhysicalRequirements}\cite{InformationalDerivation}, e.g. concerning whether measurements or states are introduced first in the theory. We will use \cite{Pfister}, \cite{Hardy} and \cite{Barrett} as an orientation.\\
\\
The basic notions of GPTs are \textit{states} and \textit{measurements}. The state $w$ completely describes a physical system in the sense that the state determines the probabilities of all the measurement outcomes for all measurements. A meaningful representation of a state would be to just list the probabilities of all the possible measurements ($o_k$ is the outcome, $m_j$ the measurement):
\begin{equation}
	w = \begin{pmatrix}
		\vdots\\
		p(o_1|m_j)\\
		p(o_2|m_j)\\
		\vdots
		\end{pmatrix}
\end{equation}
However, even for the simple example of a spin-$\frac 1 2$ system, there are infinitely many axis and thus infinitely many possible measurements. However, the probabilities for spin-up-results for measurements along the $x$-, $y$- and $z$-axis already determine the whole quantum state. Thus the example also shows, that in many cases, knowing the probabilities of some measurements already completely determines the outcomes of the other measurements. Such a set of measurements is called \textit{fiducial}.\\
\\
Next we consider the following \textit{mixing}-procedure:
Assume we have $n$ preparation devices, and each of them can prepare a state $w_j$, $j \in \{1,...,n\}$. Furthermore assume that there is random number generator with $n$ outcomes, given by the probability distribution $(p_1,...,p_n)$. Now the preparation devices and the random number generator are put into a black box with a single button on the outside. If you push that button, the random number generator is activated. If you get outcome $j$, device $j$ is activated and system $w_j$ is produced and sent to the outside. An example for this device is shown in Figure \ref{Fig:MixingDevice}.
\begin{center}
	\includegraphics[width=0.9\textwidth]{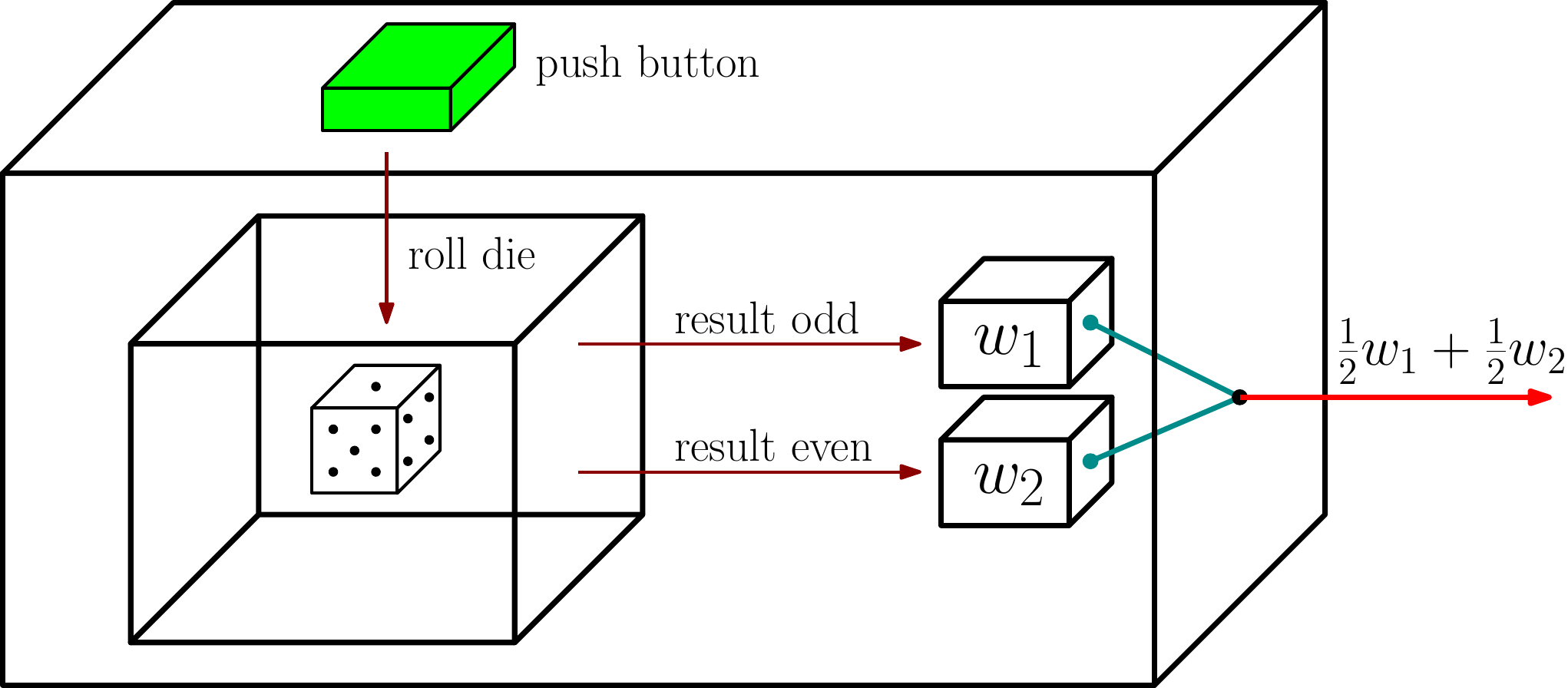}
	\captionof{figure}{\textit{\small An example for the mixing procedure described in the text: Pushing the green button activates a random number generator (here, a die) and depending on the outcome (here: odd or even), one of several states is prepared.}}
	\label{Fig:MixingDevice}
\end{center}
As everything happens within the black box, we never learn the result of the random number generator. Thus we only know that in $p_j$ of the cases, the system $w_j$ is obtained. We wish to describe the states the box outputs by something that says ``with probability $p_j$ you get the results expected for $w_j$'', i.e. a statistical mixture. If we consider the representation with fiducial probability vectors,
\begin{equation}
	w_j = \begin{pmatrix} \vdots \\ P^{(j)}(o_1|m_k)\\ P^{(j)}(o_2|m_k) \\ \vdots \end{pmatrix}
\end{equation}
we now show that it is meaningful to assume that the new state can be written as
\begin{equation}
	w = \sum_{j=1}^n p_j w_j = \begin{pmatrix} \vdots \\ \sum_{j=1}^n p_j \cdot P^{(j)}(o_1|m_k) \\ \sum_{j=1}^n p_j \cdot P^{(j)}(o_2|m_k) \\ \vdots \end{pmatrix}
\end{equation}
With probability $p_j$, the state is $w_j$. In case the state is $w_j$, for measurement $m_k$ the outcome $o_i$ occurs with probability $P^{(j)}(o_i|m_k)$. Thus the total probability for the outcome $o_i$ of measurement $m_k$ is given by $\sum_{j=1}^n p_j \cdot P^{(j)}(o_i|m_k)$. So the list of probabilities of the state $w$ should be of the form $\sum_{j=1}^n p_j \cdot P^{(j)}(o_i|m_k)$, i.e. exactly of the form $w=\sum_{j=1}^n p_j w_j$ as suggested above. This result suggests that the set of states should be embedded into a real vector space, and that statistical mixtures are described by convex linear combinations.\\
The black-box-random preparation device is an operational abstraction for a source or preparation device whose rules are not known. For example, for a random photon emitted by a star, we do not know which energy or polarisation was chosen by the star, and it might be a different one for each emitted photon.\\
\\
A general assumption for the sake of simplicity is that the vector space is finite-dimensional, i.e. that a finite list of fiducial probabilities is sufficient in case of a representation via fiducial probabilities. In quantum theory, this assumption restricts us to finite dimensional quantum theory. This assumption allows us to separate the mathematical technicalities introduced by functional analysis or topology from the ``real physics''. For example, in finite dimension there is norm equivalence, i.e. the choice of the norm is less important. Furthermore we do not have to deal with integration measures and divergent sums. Already in finite dimension the proofs often are really hard because of the generality of the framework. The general idea is therefore to characterize the finite dimensional case first. Afterwards one can try to generalize the results to infinite dimension. Furthermore the ``true'' physical theory should be capable of describing finite-dimensional systems as well. Such finite-dimensional systems often arise in computation, for example the qubit ion chains often used for quantum computation or the finite memory of regular computers. Thus if a theory fails to describe such systems, it must be wrong. Furthermore, insights from quantum gravity, especially the holographic principle, suggest that the fundamental basis of nature could be discrete and might even be finite-dimensional (for a non-technical introduction to quantum gravity, see e.g. \cite{Gravity}).
\\
Norm equivalence and the representation by lists of fiducial probabilities suggest that the state space should be bounded.\\
\\
Now assume that there is an element $w$ of the vector space, such that there is a sequence of states $w_n$ with $\lim_{n \to \infty} w_n = w$, i.e. $w$ can be approached arbitrarily well. As no preparation procedure is perfect and as all measurement devices have a finite reliability, there is no practical difference between perfect preparability and arbitrarily good preparability. Thus we also assume that $w$ is a state. This means that the set of states should be closed. In finite dimensions together with the boundedness, this means that the set of states should be compact. 
\\

Furthermore, it also makes sense to define ``subnormalized'' states. As an example, we consider a projective measurement $\mathcal P$ in quantum theory performed on a system described by a density matrix $\rho$. We call the projectors $P_j$. There is a probability of $\text{Tr}(P_j \rho P_j)$ that the outcome $j$ occurs, and afterwards that system is described by the density matrix $\frac{P_j \rho P_j}{\text{Tr}(P_j \rho P_j)}$. Instead, we can say that the state of the system is $P_j \rho P_j$, with the following interpretation: With probability $\text{Tr}(P_j \rho P_j)$, the system is in the state $\frac{P_j \rho P_j}{\text{Tr}(P_j \rho P_j)}$ after the measurement, i.e. the notation $P_j \rho P_j$ summarizes both the probability and the state in case of outcome j. Using that notation, after averaging or forgetting the result, the total density matrix after the measurement is described by a sum of subnormalized states $\rho'=\sum_j P_j \rho P_j$.\\
A slightly different usage of subnormalized states is that $\text{Tr}(\rho)$ gives the probability of success of preparation of the state. In case of failure, no system is output at all. This point of view can be related to the projective measurement from before. Only if measurement $j$ occurs, the system is described by $\frac{P_j \rho P_j}{\text{Tr}(P_j \rho P_j)}$. If another outcome occurred, outcome $j$ failed.\\
This notation can be used to hide further conditions or to include implicit conditions. For example in the projective measurement, $\rho_j := P_j \rho P_j$ is used to perform another projective measurement $\mathcal Q$ with projectors $Q_k$. The probability for outcome $k$ in the $\mathcal Q$-measurement is given by $\text{Tr}(Q_k \rho_j Q_k)$. This can be rewritten as:
\begin{align*}
	\text{Tr}(Q_k \rho_j Q_k) &= \text{Tr}\left(Q_k \frac{\rho_j}{\text{Tr}(\rho_j)}\right)\cdot \text{Tr}(\rho_j)\\ &= \text{Prob(outcome $k$ in $\mathcal Q$| outcome $j$ in $\mathcal P$)} \cdot \text{Prob(outcome $j$ in $\mathcal P$)}\\
		&= \text{Prob(outcome $j$ in $\mathcal P$, afterwards outcome $k$ in $\mathcal Q$)}   
\end{align*}
	Thus now all probabilities calculated with $\rho_j$ contain the additional event that outcome $j$ in the first measurement is obtained.\\
	All these applications show, that subnormalized states are not really necessary, but helpful to simplify notation and to put more content into a simpler expression. We will use a function $u_A$ to specify the normalization. Using the interpretation that the normalization gives the success probability, there should only be one state normalized to zero. This state corresponds to certain failure/no output at all. Furthermore, we will also consider ``supernormalised'' states. We do not give them a physical meaning. However, introducing them has many mathematical advantages, allowing us to use the full framework provided by cones from convex geometry.\\
	\\
	Now we collect our results to define:
	\begin{defi}
		A tripel $(A,\Omega_A,u_A)$ is called an \textbf{abstract state space} iff the following conditions hold:
			\begin{enumerate}
				\item A is a finite-dimensional, real vector space.
				\item $\Omega_A \subset A$ is a convex, compact subset.
				\item $A_+ := \mathbb R_{\ge 0} \cdot \Omega_A$ is a closed, generating cone.
				\item $u_A \in A^*$ is strictly positive on the non-zero elements of the cone, i.e. $u_A(w)>0$ for all $w\in A_+$ with $w\ne 0$.
				\item For $w \in A_+$: $u_A(w) = 1 \Leftrightarrow w \in \Omega_A$.
			\end{enumerate}
		$\Omega_A$ is called the \textbf{set of (normalized) states}, $A_+$ the \textbf{cone of unnormalized states} and the order unit $u_A$ gives the normalization. Furthermore, $\Omega_A^{\le 1} := \{w\in A_+ | u_A(w)\le 1\}$ is called the set of \textbf{subnormalized states}.
	\end{defi}

An example for a state space is shown in Figure \ref{Fig:StateSpace}. We note, that the definition of abstract state spaces is overcomplete - some properties are consequences of other properties.
\begin{center}
	\includegraphics[width=0.3\textwidth]{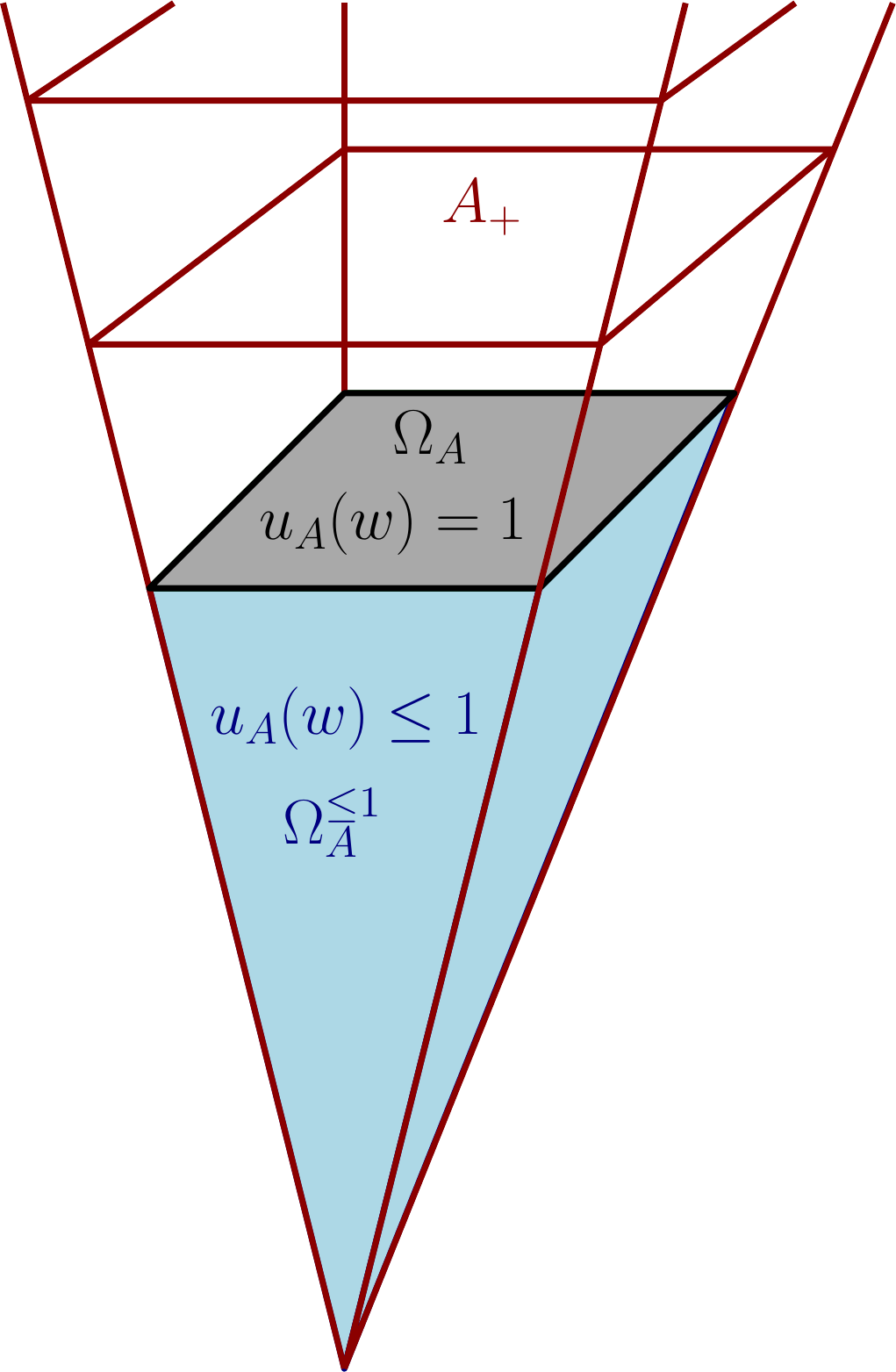}
	\captionof{figure}{\textit{\small A state space consists of a cone of unnormalized states $A_+$, where the normalization is defined by an order unit $u_A$. The set of normalized states is given by the states with $u_A(w)=1$. The set of subnormalized states $\Omega_A^{\le 1}$ is given by those states with $u_A(w)\le 1$. The (sub)normalized states have a physical interpretation, the normalization gives the probability of success of preparation.}}
	\label{Fig:StateSpace}
\end{center}

\subsection{Measurements}
So far, we have only defined state spaces. We also want to describe actions on the system, especially measurements and transformations.\\
\\
At first we consider measurements:\\
Assume there is a system in the state $w\in \Omega_A$. We wish to perform a measurement with $n$ different outcomes on the system. As the idea of a state is that it fully determines the outcome probabilities of all measurements, this measurement will be no exception. Thus it is possible to define functions $e_j : \Omega_A \to [0,1]$ that give the probabilities $e_j(w)$ of the measurement outcomes.\\
Consider a black-box preparation device, which with probability $p$ prepares a system in the state $w_1$, and with probability $1-p$ in the state $w_2$. The total state is $w=pw_1 +(1-p)w_2$. The probability for the $j$-th outcome is $e_j(w)=e_j(pw_1 +(1-p)w_2)$. But there is also another way to think about the black-box: With probability $p$, the system is in the state $w_1$. In that case, the probability for the $j$-th outcome is $e_j(w_1)$. In the other case, which happens with probability $1-p$, the probability for the $j$-th outcome is $e_j(w_2)$. The law of total probability states that the total probability is given by $p e_j(w_1) +(1-p) e_j(w_2)$. As both points of view describe the same situation, $e_j(pw_1 +(1-p)w_2) = p e_j(w_1) +(1-p) e_j(w_2)$. Thus the functions $e_j$ are convex-linear. As $\mathbb R_{\ge 0} \cdot \Omega_A = A_+$ and $\rm{span}(A_+) = A$, it is reasonable to assume that the $e_j$ can be convex-linearly continued to $A$. Furthermore, the state $0$ should give $0$ for all probabilities, i.e. $e_j(0) = 0$. By Proposition \ref{Proposition:LinearConvexLinear}, $e_j$ are linear functions on $A$. Note, that the linearity implies $e_j(w) =u_A(w)e_j \left( \frac{w}{u_A(w)}\right) $: If there is no system prepared, we will not measure anything. In Appendix \ref{Chapter:ConvexLinearToLinear}, we will discuss the technical details that the intuition is right, i.e. that it is possible to extend the effects to linear functions on $A$.\\

\begin{defi}
	For an abstract state space, the \textbf{set of effects} is defined by 
	\begin{equation}
		E_A := \{ e \in A^* | 0\le e(w) \le u_A(w) \ \forall w \in A_+ \}
	\end{equation}
	i.e. effects are linear maps $e:A \rightarrow \mathbb R$ with $0 \le e(w) \le 1$ for all $w \in \Omega_A$.
\end{defi}

\vspace{1mm}

\begin{defi}
	For effects $e,f \in E_A$, we write
	\begin{equation}
		e \le f \quad \Leftrightarrow \quad e(w) \le f(w) \quad \forall w \in \Omega_A
	\end{equation}
	or equivalently
		\begin{equation}
		e \le f \quad \Leftrightarrow \quad e(w) \le f(w) \quad \forall w \in A_+
	\end{equation}
	Likewise, $\ge$ is defined.\\
	Furthermore, we define
		\begin{equation}
		e < f \quad \Leftrightarrow \quad e(w) < f(w) \quad \forall w \in \Omega_A
	\end{equation}	
	and analogously $>$.
\end{defi}

Furthermore, measurement probabilities on properly normalized states should sum to $1$.

\begin{defi}
	A \textbf{measurement} is a set $\mathcal M = \{e_1,...,e_n\}$ of effects such that $\sum_{j=1}^n e_j = u_A$.
\end{defi}

\begin{note}
	While we have defined measurements in a mathematical sense, it is not clear whether these measurements can actually be implemented in an actual experiment. Thus it is not clear whether the measurements are physically allowed. A typical assumption is the ``\textit{no-restriction}''-hypothesis, which claims that all mathematically well-defined measurements are physically possible. Such an assumption can be justified as follows:\\
	Of course there can be practical limitations (insufficient control, too expensive,...), but also conceptual problems which forbid a measurement. An example for the latter one could be given by space-like separated systems, such that it is not possible to act on both systems at the same time. However, all these limitations are not introduced by quantum theory itself, but by the choice of physical system it is applied to. Just like with a two-level system in quantum theory (qubit), many different physical systems might be described by the same abstract state space. Thus it can happen, that the same mathematical measurement might be impossible for one physical system, but possible for another physical system. Thus one decides not to exclude any well-defined measurement beforehand. However, a choice of a specific physical system might render some measurements impossible. This means, that not the GPT makes a measurement impossible, but the physical system it is applied to. As we want our framework to be as powerful as possible, one does not exclude any measurement without a reason.\\
	We will not use the no-restriction hypothesis here, as the postulates used by us will ensure that all effects are physically allowed. But in the general case, if the no-restriction hypothesis is neither a postulate nor a consequence of the postulates one chooses, then one has to introduce an extra set which specifies the allowed effects/measurements.\\ However, we will assume that all well-defined measurements formed by allowed effects also are allowed measurements.\\
	\\
	The most prominent example for an impossible measurement is to measure position and momentum of a particle in non-relativistic quantum mechanics. As the position eigenstates form a basis, the measurement of the position is already normalized to one. Adding effects of the form ``\textit{Is the particle's momentum found in $[p_1,p_2]$ ?}'' would lead to a mathematically ill-defined measurement, because the total probability would be larger than $1$. Thus this important example, which is constructed from allowed effects, is already mathematically forbidden because of wrong normalization. 
\end{note}

\begin{ass}
	Let $e_1,...,e_n$ be physically allowed effects with $u_A \ge \sum_{j=1}^n e_j$. Then $\{e_1,...,e_n\}$ can appear in a common physically allowed measurement.
\end{ass}

Furthermore, we assume that for any event described by an effect $e$, also the counter-event described by $u_A-e$ is physically allowed:
\begin{ass}
	If $e$ is a physically allowed effect, then so is the effect $u_A-e$.
\end{ass}
If two effects $e_1,e_2$ can appear in a common measurement, then also $e_1+e_2$ should be a physically valid effect. It can be obtained by assigning a new combined outcome to $e_1$ and $e_2$ which does not distinguish any more if $e_1$ or $e_2$ was triggered.
\begin{ass}
	If $e_1,e_2$ are physically allowed effects with $e_1+e_2\le u_A$, then also $e_1+e_2$ is a physically valid effect.
\end{ass}

\begin{defi}
	A set of states $w_1,...,w_n \in \Omega_A$ is called \textbf{perfectly distinguishable}, if there is a set of allowed effects $e_1,...,e_n$, which can appear in a common measurement with $\sum_j e_j \le u_A$ and for which 
	\begin{equation}
		e_j(w_k) =\delta_{jk}
	\end{equation}
\end{defi}
\begin{prop}
	If $w_1,...,w_n$ are perfectly distinguishable, there also exists a properly normalized allowed measurement $e_1',...,e_n'$ with $\sum_j e_j' = u_A$ and $e_j'(w_k) = \delta_{jk}$. \label{Prop:SubnormalizedDistinguishable}
\end{prop}
\begin{proof}
	For $j<n$, set $e_j' := e_j$. These trivially fulfil $e_j'(w_k) = \delta_{jk}$. Furthermore, set $e_n' := u_A - \sum_{j=1}^{n-1} e_j$. Thus we obtain a properly normalized measurement. $e_n'$ is a effect, because $u_A$ and $e_j$ are linear and $0 \le u_A -\sum_{j=1}^{n-1} e_j \le u_A$.\\
	By our assumptions, $e_n'$ is a physically valid effect, and
	\begin{equation}
	e_n'(w_j) = 1-\sum_{k=1}^{n-1}e_k(w_j) = 1-\sum_{k=1}^{n-1} \delta_{jk} = \delta_{nj}
	\end{equation}
	The last equality holds, because for $j < n$,  $\sum_{k=1}^{n-1} \delta_{jk}= 1$, while for $j= n$ we find $\sum_{k=1}^{n-1} \delta_{jk} = 0$.
\end{proof}

Physically, this measurement can be implemented by using the measurement which includes $e_1$,...,$e_n$. If we do not obtain outcome $1$,...,$n-1$, we say we obtained outcome $n$. Thus the event with outcome $n$ is the counter-event for the event \textit{outcome 1,...,n-2 or n-1 obtained} and thus is indeed described by $u_A -\sum_{j=1}^{n-1} e_j$.

\vspace{1mm}

\begin{exam} \label{Example:FigureToEffect}
	We consider the affine hyperplane given by $\Omega_A$, i.e. $u_A^{-1}(1)$ and choose an origin in this plane such that $u_A^{-1}(1)$ can be considered a vector space. All vectors we will consider now will be vectors in $u_A^{-1}(1)$.
	Effects are convex-linear, and therefore of the form $e=L(\cdot)+y$ with $y$ a constant, $L$ a linear map. Thus there exists a vector $\vec v$ with $e(\vec w) = \vec v^T \cdot \vec w+y$. The sets with constant values $e^{-1}(a) = \{ \vec w | \vec v^T \cdot \vec w+y = a \}$ define affine hyperplanes in $u_A^{-1}(1)$ with $\vec v$ as normal vectors. 
	\\ \\ 
	Vice versa, two parallel hyperplanes in $u_A^{-1}(1)$ can be used to define an unique linear functional in $A^*$ (see also Figure \ref{Fig:FigureToEffect}): \\
	 In the vector space $u_A^{-1}(1)$, we think of two parallel hyperplanes: $\{\vec w \ | \ \vec v^T \cdot \vec w = a'\}$ and $\{ \vec w \ | \ \vec v^T \cdot \vec w = b'\}$. Hereby, $\vec v$ is the normal vector of the hyperplanes. We consider the convex-linear map $\check e : u_A^{-1}(1) \rightarrow \mathbb R$:
	\begin{equation}
		\check{e}(\vec w) = \frac{\vec v^T \cdot \vec w - a'}{b'-a'}\cdot b + \frac{\vec v^T \cdot \vec w - b'}{a'-b'}\cdot a
	\end{equation}
	One can directly check $\check{e}(\vec w) = a$ if $\vec v^T \cdot \vec w = a'$ and $\check{e}(\vec w) = b$ if $\vec v^T \cdot \vec w = b'$. Thus all points on the hyperplane $\{\vec w \ | \ \vec v^T \cdot \vec w = a'\}$ are assigned the value $a$, the points on the other hyperplane get the value $b$.\\ 
	By Appendix \ref{Chapter:ConvexLinearToLinear}, we know that this convex-linear map can be linearly extended to $A$. If the values $a$ and $b$ are chosen carefully, one can define effects in this way. A typical construction is to choose the parallel hyperplanes such that $\Omega_A$ is found between them, while one hyperplane defines the states with $e(w)= 0$ and the other one the states with $e(w)=1$, as shown in Figure \ref{Fig:FigureToEffect}.
	\begin{center}
		\includegraphics[width=0.5\textwidth]{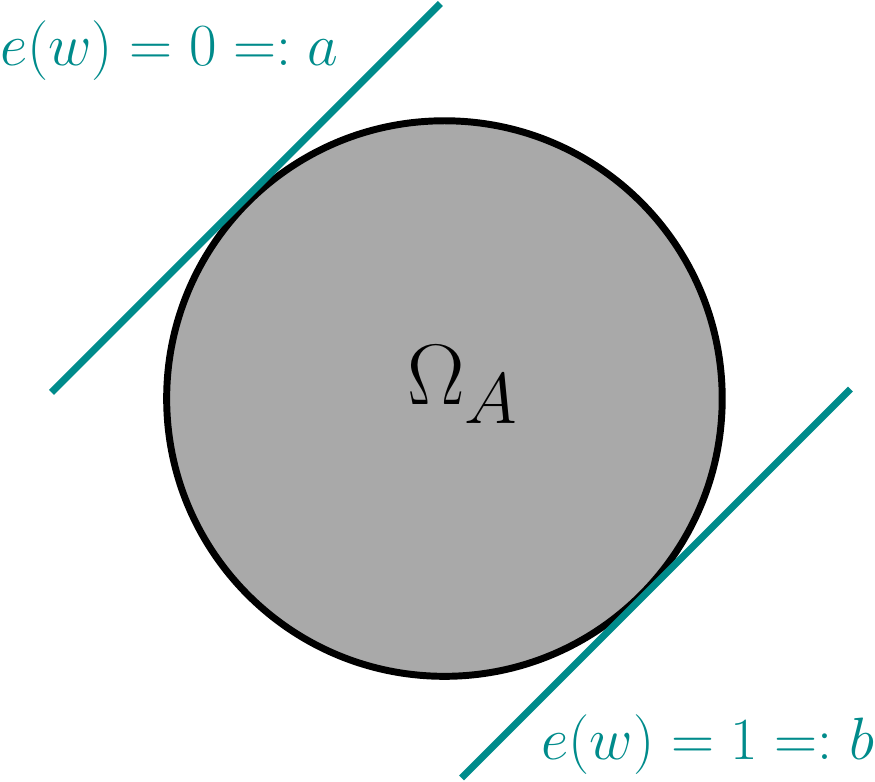}
		\captionof{figure}{\small \textit{This figure shows an abstract state space and two parallel hyperplanes. These two planes can be used to define a convex-linear map on $\Omega_A$ or a linear map on $A$.}}
		\label{Fig:FigureToEffect}
	\end{center}
\end{exam}
	
\subsection{Transformations and operations}

Next we consider transformations:\\
A transformation converts one state of a physical system into another state (possibly of another system): $T: A_+ \to B_+$. With the same argument like for effects, transformations have to be convex-linear and can be extended to linear maps $T: A \to B$. Furthermore, transformations should not increase the normalization, as otherwise a physically meaningful normalized state could be changed into an unphysical supernormalized state. Furthermore, a transformation should convert states into states.

\begin{defi}
	Let $A,B$ be two abstract state spaces. A \textbf{transformation} is a map $T: A \to B$ which satisfies:
	\begin{enumerate}
		\item $T(A_+) \subset B_+$. This property is called \textbf{``T is positive''}.
		\item $u_B \circ T \le u_A$, i.e. $u_B(T(w)) \le u_A(w)$ $\forall w \in A_+$. 
		\item $T$ is linear
	\end{enumerate}
\end{defi}

\vspace{1mm}

\begin{note}
	Just like with measurements and effects, not all well-defined transformations have to be physically allowed. Especially, if the no-restriction hypothesis is not satisfied, then for every physically allowed effect $e$ and every physically allowed transformation $T$, also $e\circ T$ has to be a physically allowed effect. This requirement is stronger than just $u_B \circ T \le u_A$. Analogously, for a physically allowed measurement $\{e_r | r \in R\}$, also $\{e_r\circ T | r \in R\}$ has to be a physically allowed measurement. We will clarify the physical assumptions used in this thesis when stating the postulates in Chapter \ref{Section:Postulates}. \\
	In Quantum Theory, one usually demands a stronger property than positivity: complete positivity. This means that for all types of composite systems, also the transformation $T\otimes \mathbbm 1$ has to be positive. This map means, that on one part of the composite system, the transformation $T$ is applied, while the other part is not changed at all. However, we do not consider composite systems here. In many axiomatic derivations of quantum theory, one uses a postulate called \textit{Local Tomography}. It states that all states of composite systems can be characterised by local measurements and correlations between them. In \cite{Barrett}, Barrett shows how from such an assumption, a tensor product rule for composite systems can be derived. However, we will use different postulates because we do not consider composite systems in this work. Note that all completely positive transformations are also transformations in our sense, thus the results we derive will also be valid for completely positive transformations. 
\end{note}

\vspace{1mm}

\begin{defi}
	A transformation $T: A\rightarrow B$ is called \textbf{reversible}, if $T^{-1}$ exists and is transformation too.
A physically allowed transformation $T$ is called \textbf{(physically) reversible}, if $T^{-1}$ exists and is a physical transformation. The set of physically allowed physically reversible transformations is denoted by $\mathcal G_A$.
\end{defi}

Next we consider operations.\\
We consider a collection $\{T_1,...,T_n\}$ of transformations. We have a device which randomly applies exactly one of the transformations to any system which enters the device. With our usual interpretation, $u_A\circ T_j(w)$ gives the probability that the $j$-th transformation is applied to an incoming state $w$, i.e. with probability $u_A\circ T_j(w)$ the system will be in the state $\frac{T_j(w)}{u_A\circ T_j(w)}$. With probability $1-u_A\circ T_j(w)$, this procedure fails and one of the other transformations is applied. As the total probability should be given by $u_A$ (i.e. $1$ for a properly normalized state), the transformations should satisfy $\sum_{j=1}^n u_B\circ T_j =u_A$. In case of a black-box device which does not tell us which result it obtained for $j$, the state of the system after leaving the device will be described by:
\begin{equation}
	w' = \sum_{j: \ u_A\circ T_j(w) \ne 0} u_A\circ T_j(w) \cdot \frac{T_j(w)}{u_A\circ T_j(w)} = \sum_{j = 1}^n T_j(w)
\end{equation}

\begin{defi}
	An operation is a collection of transformations $\mathcal O = \{T_1,...,T_n\}$ which satisfies
	\begin{equation}
		\sum_{j=1}^n u_B\circ T_j =u_A
	\end{equation}
\end{defi}

The most famous examples for operations are given by projective measurements in quantum theory: Here, the transformations are given by projections. Thus it is possible to model projective measurements by using operations.

\begin{note}
	So far, our notion of an abstract state space only describes the structure of the set of states. This definition can be extended to what is sometimes called a \textbf{dynamical abstract state space} (see also \cite{Hardy}\cite{Barrett}): An abstract state space together with a set of physically allowed measurements and operations $(A,\Omega_A,u_A,\mathcal M_A, \mathcal O_A )$. If one assumes that all mathematically well-defined measurements composed only of physically allowed effects are also physically allowed (like we do in this thesis), then it is sufficient to specify the set of allowed effects $\mathcal E_A$ instead of the set of allowed measurements $\mathcal M_A$. If the no-restriction hypothesis is assumed to hold, then it is not necessary to specify $\mathcal M_A$, $\mathcal E_A$.\\
Similarly, if one assumes that all mathematically well-defined operations constructed from physically allowed transformations are also physically allowed, then it is sufficient to specify the set of allowed transformations $\mathcal T_A$ instead of the set of allowed operations $\mathcal O_A$. Furthermore, it is important that one of the main applications of the GPT-framework is to derive quantum theory. In many such derivations, one is not interested in the set of allowed operations, but rather in the set $\mathcal G_A$ of physically reversible transformations. The reason is that such operations map the state space onto itself in a reversible way; such symmetry transformations put many restrictions on the possible shapes of the state space and are therefore very useful in axiomatic derivations of quantum theory. So often an abstract state space is considered as a tupel $(A,\Omega_A,u_A, \mathcal G_A)$, specifying also the set of physically reversible transformations. When we state our postulates in Chapter \ref{Section:Postulates}, we will also consider $(A,\Omega_A,u_A, \mathcal G_A)$.
\end{note}

\vspace{1mm}

\begin{exam}
	In \textbf{classical probability theory}, in principle measurements can be performed without disturbing the system. In principle, it is possible to combine all measurements into one large measurement, from which all other probabilities can be derived. The assumption of finite-dimensionality means that there are only a finite number of outcomes. For example, a n-sided die is fully characterised by the probabilities for the n sides. Other probabilities, for example ``\textit{Does the die show a prime number?}'' can be deduced from that. Thus, in classical probability theory, one assumes that it is possible to find a single fiducial measurement which describes the whole system. Thus the states can be be described by listing the probabilities of all the outcomes of that measurement:
	\begin{equation}
		w = \begin{pmatrix} p_1 ,..., p_n \end{pmatrix}
	\end{equation}
	The pure states are states with a predetermined outcome: $w^{(k)}=\left(p^{(k)}_j\right)_{j=1,...,n} = \left( \delta_{jk} \right)_{j=1,...,n}$.\\
	Especially, all pure states are perfectly distinguishable.
	All other states can be obtained by statistical mixtures.\\
	\begin{equation}
		\Omega_\text{classical} := \rm{conv} \left(\begin{pmatrix} 1\\ 0\\ \vdots \\ 0 \end{pmatrix}, \begin{pmatrix} 0\\ 1\\ \vdots \\ 0 \end{pmatrix},... ,\begin{pmatrix} 0\\ 0\\ \vdots \\ 1 \end{pmatrix}\right)
	\end{equation}
	State spaces with only finitely many extreme points are called \textbf{polytopes}. Finite-dimensional state spaces are polytopes. Even more, every mixed state has a unique decomposition into pure states, so classical states spaces are \textbf{simplices}.
\end{exam}

\vspace{1mm}

\begin{exam}
	In \textbf{Quantum Theory}, the most general description of states is given by density operators, which include both irreducible quantum randomness and the classical randomness caused by ensembles or missing knowledge. Here for a finite-dimensional complex Hilbert space $\mathcal H$, we define $A=\{\text{hermitian operators on } \mathcal H \}$, $\Omega_A= \{\text{density operators on } \mathcal H \} = \{\rho \in A \ | \ \rho \ge 0, \rho^\dagger = \rho, \rm{Tr}(\rho) = 1 \}$. Note that $A$ is only a real vector space. For example, $\mathbbm 1$ is hermitian, but not $i \cdot \mathbbm 1$.\\
	There are two different cases to which the density state formalism can be applied.\\ 
	The first application is from statistical physics/thermodynamics:\\
	Here, density operators describe ensembles of quantum systems whose microstates all realize the same macrostate. If one applies a measurement of the observable $A$ to the ensemble, an average value $\braket{A} = \text{Tr}(A\rho)$ is obtained.\\
	\\
	The second application is mainly used in quantum information theory:\\
	Like explained before for more general GPTs, density operators can be used to describe quantum states whose preparation is not fully known. If we know that with probability $p_j$ the state $\ket{\phi_j}$ was prepared, then $\rho = \sum_j p_j \ket{\phi_j}\bra{\phi_j}$ is the our best description of the state of the system. While it does not make sense to consider measurements beyond average values for ensembles, for single systems of unclear preparation it makes sense to consider measurements, where only one out of several outcomes is obtained. The most general such measurements are described by \textbf{POVMs}(Positive operator valued measurements):\\ Let $E_1,...,E_n \ge 0$, $\sum_j E_j = \mathbbm 1$, $E_j^\dagger = E_j$. Then $e_j(\cdot) := \text{Tr}(E_j \cdot)$ form a measurement which is called a POVM. The correspondence between $e_j$ and $E_j$ is induced by the self-duality of quantum theory.
\end{exam}

\begin{defi}
	By $A_+^*$ we denote the set of unnormalized effects, i.e.:
	\begin{equation}
		A_+^* = \mathbbm R_{\ge 0} E_A = \{ e \in A^* | \ e(w)\ge 0 \ \forall w \in A_+\}
	\end{equation}
	$A_+^*$ is also called the \textbf{dual cone}.\\
	We say a state space is \textbf{(strongly) self-dual}, if there exists an inner product $\braket{\cdot,\cdot}$ such that:
	\begin{equation}
		A_+^* = \{\braket{\cdot,w}| \ w \in A_+\}
	\end{equation}
	Thus $A_+$ and $A_+^*$ can be identified with each other in case of self-duality.
\end{defi}

\subsection{Equivalent state spaces}
So far we have motivated abstract state spaces by lists of fiducial probabilities. However, as quantum theory and the Bloch ball suggest, sometimes other choices for the state space are more convenient. So some state spaces are physically equivalent, if they have the same convexity-structure\cite{Masanes}:
\begin{defi}
	Two state spaces $(A,\Omega_A,u_A)$ and $(B,\Omega_B, u_B)$ are \textbf{equivalent} if there exists a bijective linear map $L: A \rightarrow B$ such that $L(A_+) = B_+$ and $u_B \circ L = u_A$.
\end{defi}

\begin{note}
	If the considered state spaces also include a set of allowed effects/measurements/operations/(reversible) transformations, then the map $L$ also has to conserve these sets, e.g. $\mathcal E_B \circ L = \mathcal E_A$ for the sets of allowed effects.
\end{note}

Our definition of abstract state spaces does not start from a list of probabilities, but rather from any convex compact set. So it it important to note that every state space is equivalent to a state space that has the form of a list of probabilities:
\begin{theorem}
	Let $\Omega_A$ be a GPT with $\text{dim}(A)=N$. Then $\Omega_A$ is equivalent to a state space $\Omega_B$ such that all components can be found between $0$ and $1$, i.e. can represent probabilities.
\end{theorem}
\begin{proof}
	Here, we provide only a proof sketch.\\
	$u_A^{-1}(1)$ describes a hyperplane in $A$ which contains $\Omega_A$. We rotate this hyperplane such that it is perpendicular to the $x_N$-axis. As $\Omega_A$ is bounded, there exists a $c > 0$ such that $\Omega_A \subset [-c,c]^{N-1} \times \{1\}$. The $N-1$ vectors $(0,...,0,c,0,...,0,0)$ together with the vector $(-c,-c,...,-c,1)$ form a basis. We define a linear map by $(0,...,0,c,0,...,0,0) \mapsto (0,...,0,\frac 1 2,0,...,0,0)$ and $(-c,-c,...,-c,1)\mapsto (0,...,0,1)$. These new states also form a basis, thus the map is invertible. In particular, $[-c,c]^{N-1} \times \{1\} \rightarrow [0,1]^{N-1}\times \{1\}$. Thus the first $N-1$ components of $\Omega_A$ are now found in $[0,1]$, while the last one is $1$ and gives the normalization.
\end{proof}

\subsection{Some mathematical properties of GPTs}

\begin{lem}
	For an abstract state space, let $F\subset A_+$ be a face. Then for all $w \in F$, we have $\mathbb R_{\ge 0} \cdot w \subset F$. \label{lemma:rays}
\end{lem}

\begin{proof}
	Let $\lambda \in \mathbb R_{\ge 0}$ be arbitrary. If $\lambda = 1$, then $\lambda w \in F$ trivially.\\
	If $\lambda > 1$, then $w = \frac{1}{\lambda}\lambda w + (1-\frac{1}{\lambda}) 0$. As $F$ is a face, this implies $0 \in F$ and $\lambda w \in F$.\\
	If $\lambda < 1$, then $\lambda w = \lambda \cdot w +(1-\lambda)\cdot 0$. As faces are convex, also $\lambda w \in F$.
\end{proof}

\vspace{1mm}

There is a bijective correspondence between the faces of $\Omega_A$ and $A_+$:
\vspace{1mm}

\begin{prop}
	For an abstract state space, a face $F$ of $\Omega_A$ induces a face $\mathbb R_{\ge 0} \cdot F$ of $A_+$. Vice versa, a face $F\ne \{0\}$ of $A_+$ induces a face $\Omega_A \cap F$ of $\Omega_A$. If $w_1,...w_m\in \Omega_A$ generate the face $F$ of $\Omega_A$ ($A_+$), they also generate the corresponding face of $A_+$ ($\Omega_A$). Furthermore for a face $F\subset \Omega_A$, $\Omega_A \cap (\mathbb R_{\ge 0}\cdot F) = F$, and vice versa for a face $G\subset A_+$, $G \ne \{0\}$, we find $G= \mathbb R_{\ge 0} \cdot (\Omega_A \cap G)$. \label{Prop:FaceCorrespondence}
\end{prop}

\begin{proof}
	The proof is quite technical and is provided in Appendix \ref{Section:FaceCorrespondence}.
\end{proof}
\vspace{1mm}
The following lemma is based on Pfister's Proposition 3.36 \cite{Pfister}:
\vspace{1mm}

\begin{lem}
	Let $e: \Omega_A \to [0,1]$ be an effect such that there exists a state $w$ with $e(w) = 1$ (or $e(w) = 0$).\\ Then $e^{-1}(1)$ (or $e^{-1}(0)$) is a face of $\Omega_A$. \label{Lemma:EffectFace} 
\end{lem}
\begin{proof}
	Let $e(w) = 1$ (or $e(w) = 0$) and $w = \sum_{j=1}^n p_j w_j$ be any convex combination of states with $p_j > 0$. Then, by convex linearity, $\sum_{j=1}^n p_j e(w_j) = e(w) = 1$ (or $0$). As $p_j > 0$ and $\sum_{j=1}^n p_j = 1$ and $e(w_j) \in [0,1]$, this requires $e(w_j) = 1$ (resp. $e(w_j) = 0$). This especially holds true for $n=2$.\\
	 Furthermore, $e(w)=e(w') = 1$ (or $0$) implies $e(pw+(1-p)w') = p e(w)+(1-p)e(w') = p +1-p = 1$ (or $p\cdot 0 +(1-p)\cdot 0 = 0$). Thus, $e^{-1}(1)$ (or $e^{-1}(0)$) is a face of $\Omega_A$ (they are non-empty by requirement). \\ 
\end{proof}

\subsection{The gbit}
Now we consider our first GPT example beyond quantum and classical theory. The gbit or (square-bit) is a square-shaped set of normalized states, see Figure \ref{Fig:Gbit}. It is very important, because one can interpret it as one half of a so-called PR-box with superstrong correlations (see e.g. \cite{Barrett}, \cite{Ududec} for more explanations). Also, the gbit is one of the most simple state spaces which is often used as a counter-example. Here we will consider it to apply all the basic notions important for GPTs in a concrete example.
\begin{center}
	\includegraphics[width=0.6\textwidth]{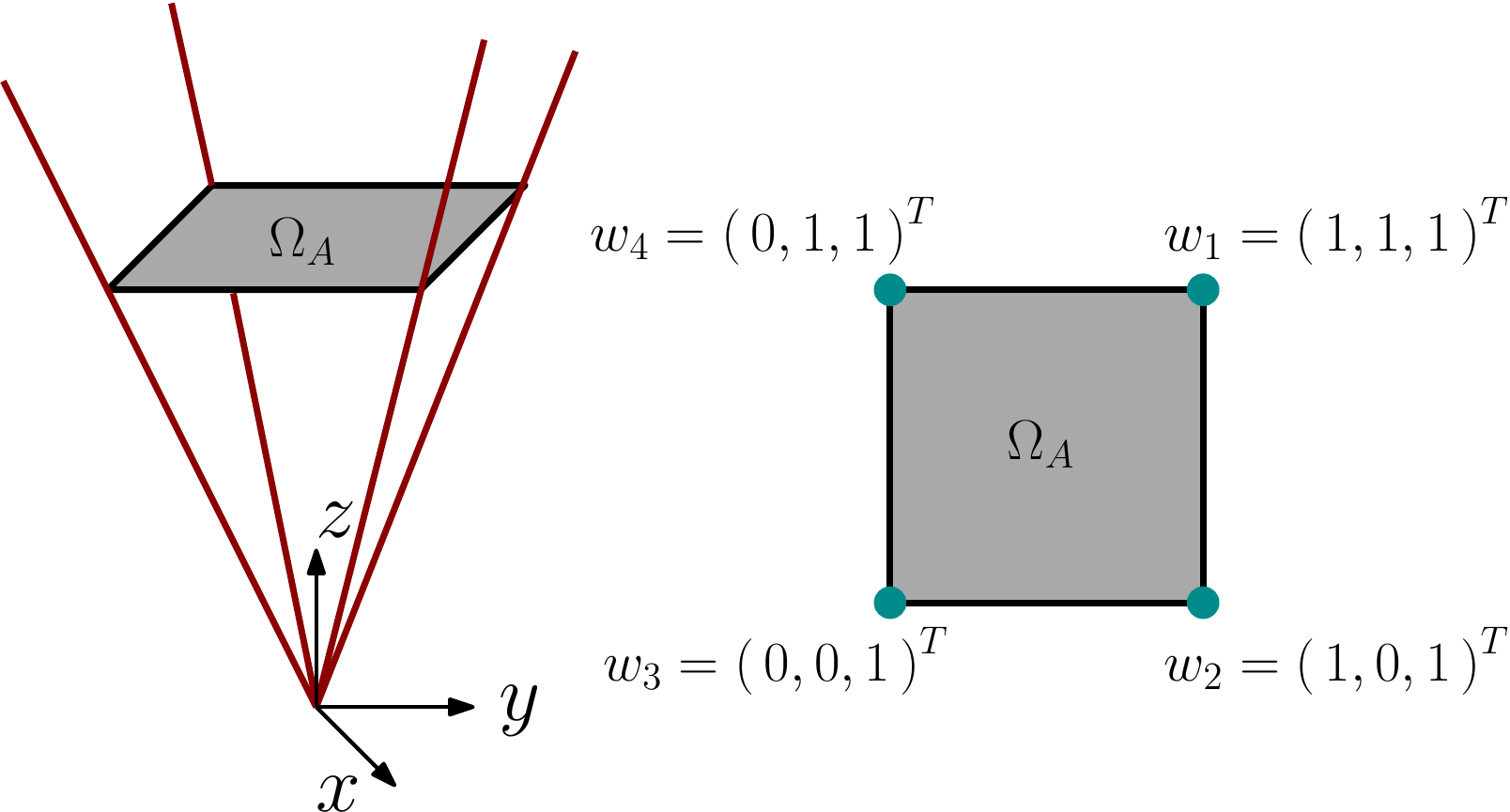}
	\captionof{figure}{\textit{\small The cone of unnormalized states $A_+$ and the square-shaped set of normalized states $\Omega_A$ for the gbit.}}
	\label{Fig:Gbit}
\end{center}

\begin{defi}
Let $A=\mathbb R^3$. We set $\Omega_A := \rm{conv}(w_1,w_2,w_3,w_4) = \left\{ \begin{pmatrix} a \\ b \\ 1 \end{pmatrix} \Big| 0\le a,b \le 1 \right\}$, where $w_1 = \begin{pmatrix} 1 \\ 1 \\ 1 \end{pmatrix}$, $w_2 = \begin{pmatrix} 1 \\ 0 \\ 1 \end{pmatrix}$, $w_3 = \begin{pmatrix} 0\\ 0\\ 1 \end{pmatrix}$, $w_4 = \begin{pmatrix} 0\\ 1\\ 1 \end{pmatrix}$. 
\end{defi}
The idea behind that choice is, that every state is described by two experiments $E_1$, $E_2$ with two outcomes $x,y$ each. Then a normalized state can be written as $\begin{pmatrix} p(x|E_1) \\ p(x|E_2) \\ 1 \end{pmatrix}$ considering that  $p(y|E_j) = 1-p(x|E_j)$. The third component gives the normalization, i.e. $u_A = \rm{pr}_3$, the projector on the third component. These states form a square, the corners $w_j$ of this square are those with definite results for the measurements and thus deserve their name \textit{pure states}, because they are states of maximal knowledge about what will happen in these measurements.\\
\\
Now we wish to find those operations $\{T_1, T_2\}$, $T_j : A \to A$, which can be used to distinguish two opposing sides. As every side is given by the convex combinations of two pure states, it is enough to consider them:
\begin{theorem}
All operations $T_1,T_2$ which can be used to distinguish the sides given by $w_1,w_4$ and $w_2,w_3$, i.e.
\begin{enumerate}
	\item $u_A(T_1(w_4)) = u_A(T_1(w_1)) = 1$
	\item $u_A(T_2(w_2)) = u_A(T_2(w_3)) = 1$
\end{enumerate} 
have the following form:
\begin{align*}
	T_2(a_1 w_1 + a_3 w_3 + a_4 w_4) = a_3 T_2(w_3) =: a_3 v'\\
	T_1(a_1 w_1 + a_2 w_2 + a_3 w_3) = a_1 T_1(w_1) =: a_1 v
\end{align*}
with $v,v' \in \Omega_A$ arbitrary.\\
However among those, only the transformations with $v \in  \{p\cdot w_1+ (1-p) \cdot w_4 | p \in [0,1]\}$ and $v' \in  \{p\cdot w_2+ (1-p) \cdot w_3 | p \in [0,1]\}$ are repeatable, i.e. 
$(u_A \circ T_j)\Big(\frac{T_j(w)}{(u_A \circ T_j)(w)}\Big) = 1 \quad \forall w \in \Omega_A$, and thus $(u_A \circ T_j)(T_j(w)) = u_A (T_j(w))$.
\end{theorem} 

\begin{proof}
As $u_A(T_1(w)) + u_A(T_2(w)) = 1 \quad \forall w\in \Omega_A$, we find $u_A(T_1(w_2)) = u_A(T_1(w_3)) = 0$ and $u_A(T_2(w_1)) = u_A(T_2(w_4)) = 0$. In particular, we find:
\begin{align}
	T_1(w_2) = 0 && T_1(w_3) = 0\\
	T_2(w_1) = 0 && T_2(w_4) = 0 \nonumber
\end{align}
Here we used that $u_A$ as an order unit is strictly positive, i.e. $0$ is the only element of $A_+$ that is mapped to $0$.\\
\\
Now we want to construct $T_j$. At first we notice that arbitrary choices of three of the $w_k$ gives us a basis, especially:
$\begin{pmatrix} 1 \\ 0 \\ 0 \end{pmatrix} = w_1-w_4 = w_2-w_3$, $\begin{pmatrix} 0 \\ 1 \\ 0 \end{pmatrix} = w_1-w_2 = w_4-w_3$, $\begin{pmatrix} 0 \\ 0 \\ 1 \end{pmatrix} = w_3$, thus $\{w_1,w_2,w_3\}$ and $\{w_1,w_3,w_4\}$ each are a basis of $A$. Thus we already know that the $T_j$ have a 2-dimensional kernel and rank 1. We can write:
\begin{align*}
	T_2(a_1 w_1 + a_3 w_3 + a_4 w_4) = a_3 T_2(w_3) =: a_3 v'\\
	T_1(a_1 w_1 + a_2 w_2 + a_3 w_3) = a_1 T_1(w_1) =: a_1 v
\end{align*}
We rewrite $a_1 w_1 + a_2 w_2 + a_3 w_3 = a_1 w_1 + a_2 (w_1-w_4+w_3) + a_3 w_3 = (a_1+a_2) w_1 +(a_2+a_3)w_3 -a_2 w_4$ to find:
\begin{align*}
	T_2(a_1 w_1 + a_2 w_2 + a_3 w_3) &= (a_2+a_3) v'\\
	T_1(a_1 w_1 + a_2 w_2 + a_3 w_3) &=  a_1 v\\
	(T_1+T_2)(a_1 w_1 + a_2 w_2 + a_3 w_3) &= a_1 v + (a_2+a_3)v'
\end{align*}
As $1 = u_A(T_1+T_2)(w) \quad \forall w \in \Omega_A$ and $(T_1+T_2)(w_1) = v$ and $(T_1+T_2)(w_2) = v'$, we know that $v = \begin{pmatrix} ? \\ ? \\ 1 \end{pmatrix}$, $v' = \begin{pmatrix} ? \\ ? \\ 1 \end{pmatrix}$. As the $T_j$ are positive, we also know that the missing components can be found in $[0,1]$. So we have $v, v' \in \Omega_A$.\\ \\
\\
\textbf{Claim:} All $v,v' \in \Omega_A$ give rise to valid operations.
\begin{enumerate}
\item $u_A(T_1(w_4)) = u_A(T_1(w_1)) = 1$, $u_A(T_2(w_2)) = u_A(T_2(w_3)) = 1$ by $v,v' \in \Omega_A$ and $w_4=w_1-w_2+w_3$
\item Now we show positivity and that the normalization does not increase:\\
	As $A_+ = \mathbb R_{\ge 0} \cdot \Omega_A$, we either have $w=0$ ($T_j(0) = 0$, so a state again), or $w= \begin{pmatrix} a \\ b \\ c\end{pmatrix} = c \cdot \begin{pmatrix} a/c \\ b/c \\ 1 \end{pmatrix}$ with $c > 0$ and $0 \le \frac a c , \frac b c \le 1$. Thus:\\
	$T_1(w) = T_1(a\cdot (w_2-w_3) + b\cdot (w_1-w_2)+c w_3) = b v$ and\\ $T_2(w) = T_2(a\cdot (w_2-w_3) + b\cdot (w_1-w_2)+c w_3) = (a-a-b+c) v' = (c-b) v'$.
	As $b \ge 0$ and $b \le c$ both results are in $A_+ = \mathbb R_{\ge 0} \cdot \Omega_A$ again. Here one can also see that the $T_j$ reduce the normalization, as before the normalization was $c$ and now it is $b\le c$ or $c-b \le c$.
\item $u_A((T_1+T_2)(a_1 w_1+a_2 w_2+a_3 w_3)) = u_A(a_1 v + (a_2+a_3) v') = a_1+a_2+a_3 = u_A(a_1 w_1+a_2 w_2+a_3 w_3)$, i.e. $u_A(T_1+T_2) = u_A$.
\item Linearity is clear.
\end{enumerate}
Thus the claim is true.\\
\\
\\
At last we consider the consequences of repeatability:\\ \\
We want those $T_j$ for which 
\begin{equation}
  (u_A \circ T_j)\Big(\frac{T_j(w)}{(u_A \circ T_j)(w)}\Big) = 1 \quad \forall w \in \Omega_A
\end{equation}
  i.e. $(u_A \circ T_j)(T_j(w)) = u_A (T_j(w))$. As $u_A\circ (T_1+T_2) = u_A$, $u_A$ strictly positive, this is equivalent to $T_1 \circ T_2 = T_2 \circ T_1 = 0$ on $\Omega_A$ (and thus on A):\\
Let $k \in \{1,2\}$, $k \ne j$. Then 
\begin{equation}
  (u_A \circ T_j)(T_j(w)) = u_A (T_j(w)) = (u_A\circ T_j)(T_j(w)) + (u_A\circ T_k)(T_j(w))
\end{equation}  
and thus $(u_A\circ T_k)(T_j(w)) = 0$, i.e. $T_k \circ T_j = 0$. Vice versa:
\begin{equation}
  u_A (T_j(w)) = (u_A\circ T_j)(T_j(w)) + (u_A\circ T_k)(T_j(w)) = (u_A \circ T_j)(T_j(w))
\end{equation}
The interpretation is clear: If we have measured result 1, a new measurement will not lead to result 2 because of the repeatability.\\
Thus we have to choose $v,v'$ such that $T_2(v) = T_1(v') = 0$.\\
As $\rm{ker}(T_1) = \rm{span}\{ w_2,w_3\}$ and $v' \in \Omega_A$, any choice $v' \in \{p\cdot w_2+ (1-p) \cdot w_3 | p \in [0,1]\}$ is valid:
We surely have $v'=a_2 w_2+a_3w_3$. As $u_A(v') = 1$, we find $a_2+a_3=1$. As $w_2 = \begin{pmatrix} 1 \\ 0 \\ 1 \end{pmatrix}$, $w_3 = \begin{pmatrix} 0\\ 0\\ 1 \end{pmatrix}$ and $v' \in \Omega_A$, we need $0 \le a_2 \le 1$ because of the first component. But then we also have $0 \le a_3 \le 1$.\\
As $\rm{ker}(T_2) = \rm{span}\{ w_1, w_2-w_3 \} = \rm{span}\{ w_1,w_4\}$ and $v \in \Omega_A$, any choice $v \in  \{p\cdot w_1+ (1-p) \cdot w_4 | p \in [0,1]\}$ is valid: We have $v=a_1w_1+a_4w_4$. As $v\in \Omega_A$ also $a_1+a_4=1$. As $w_1 = \begin{pmatrix} 1 \\ 1 \\ 1 \end{pmatrix}$, $w_4 = \begin{pmatrix} 0\\ 1\\ 1 \end{pmatrix}$ we need $0 \le a_1 \le 1$ because of the first component. But then also $0 \le a_4 \le 1$.
\end{proof}

This shows that the measurement outcome can be used to distinguish the sides spanned by $w_1, w_4$ and $w_2, w_3$. By relabelling, it should be possible to construct operations that also separate the other two lines. However, it is not possible to separate all the corners on the square (see \cite{Pfister}). The reason is, as we have seen, that $T_j$ will disturb its input.\\
\\
At last, we give examples of equivalent definitions of the gbit:\\
\begin{enumerate}
	\item $L(w_1) = w_1$, $L(w_4) = \begin{pmatrix} -1 \\ 1 \\ 1 \end{pmatrix}$, $L(w_3) = \begin{pmatrix} -1 \\ -1 \\ 1 \end{pmatrix}$ can be linearly extended. Then $L(w_2) = L(w_1-w_4+w_3) =\begin{pmatrix} 1 \\ -1 \\ 1 \end{pmatrix}$ . L is invertible. It identifies the gbit with $[-1,1]^2\times\{1\}$. 
	\item \cite{Barrett} uses a notation $\begin{pmatrix}p(x|E_1) \\ p(y|E_1) \\ p(x|E_2) \\ p(y|E_2) \end{pmatrix}$ which contains also the probabilities for the other results, but the normalization is not explicitly listed any more. It is given by $p(x|E_1) + p(y|E_1) = p(x|E_2) + p(y|E_2) =: c$.\\
	We can achieve this form by $L\left(\begin{pmatrix} a \\ b \\ c\end{pmatrix}\right) = \begin{pmatrix} a \\ c-a\\b\\c-b \end{pmatrix}$, i.e. $L = \begin{pmatrix} 1 & 0 &0\\ -1 & 0 & 1 \\ 0 & 1 & 0 \\ 0 & -1 & 1 \end{pmatrix}$. $L$ is a bijective map from $A$ to $\rm{span}\left\{  \begin{pmatrix} 1 \\ -1\\0\\0 \end{pmatrix},  \begin{pmatrix} 0 \\ 0\\1\\-1 \end{pmatrix}, \begin{pmatrix} 0 \\ 1\\0\\1 \end{pmatrix} \right\}$. 
\end{enumerate}

\section{The postulates}
\label{Section:Postulates}
After we have introduced the framework, we can finally state the postulates. These postulates are taken from \cite{Postulates}, where it was shown that in finite dimension, they single out quantum theory, i.e. quantum theory is the only theory which satisfies these postulates. \\
\\
\cite{Postulates} also makes some additional weak assumptions on the set of allowed effects, i.e. that it is convex and closed (for similar reasons like $\Omega_A$ is convex and closed) and that it has full dimension (to ensure that there are no different states that give the same probabilities for all measurements). We will not explicitly list this set as the postulates will have the consequence that all effects are allowed.

\subsection{Motivation of the postulates}
Consider $n$ perfectly distinguishable pure states $w_1,...,w_n$. The convex hull of these states has all the properties of a classical $n$-level system. The first postulate is, that all states are an element of a classical subspace:

\begin{postulate}[Postulate 1][\hspace{-3mm} Classical decomposability/ weak spectrality] \\\textit{
	For every state $w\in \Omega_A$, there exists a probability distribution $p_1,...,p_n$ and perfectly distinguishable pure states $w_1,...,w_n$ such that:
	\begin{equation}
		w = \sum_{j=1}^n p_j w_j
	\end{equation}
	}
\end{postulate}

This means that the only non-classical behaviour exists, because not all states have to belong to the same classical subspace.\\
\\
Next we consider a postulate which is important for powerful computation:\\
The computation power of a classical computer does not depend on its physical implementation. No matter if using silicon wafers, Lego, or redstone in Minecraft\cite{Minecraft}, there are many ways to build Turing machines. An important requirement is, that in terms of computation power, all classical implementations of a bit are equivalent. This property is called \textit{Bit Symmetry}. This property can be expanded to all classical $n$-level systems, where $n=2$ is the bit. quantum computation and quantum information usually are analyzed with the assumption that all quantum $n$-level systems are equivalent, e.g. all qubits are equivalent. Especially, this assumption implies that in principle, it should be possible to translate an entangled state of a composite system to a state of a single system and vice versa without any losses. Thus this assumption is crucial for the superior computation power of quantum computers. In more mathematical terms, the equivalence between $n$-level systems requires a translation-function $T$, which translates one n-level system into another. Furthermore, it should be possible to translate back. Thus \cite{Postulates} considers (dynamical) state spaces $(A, \Omega_A, u_A, \mathcal G_A)$ with $\mathcal G_A$ the Lie group of physically reversible physically allowed transformations and postulates:

\begin{postulate}[Postulate 2][\hspace{-3mm} Strong symmetry / generalised bit symmetry] \\
	\textit{
	For any $n\in \mathbb N$, let $w_1,...,w_n$ and $v_1,...,v_n$ be sets of perfectly distinguishable pure states. Then there exists a reversible transformation $T$ such that $T(w_j) = v_j \ \forall j$.
	}
\end{postulate}

To simplify notation, we refer to $\mathcal G_A$ as the set of reversible transformations. Hereby the condition that the transformations and their inverses have to be physically allowed is implicit.\\
\\
The next postulate is related to the famous two-slit experiment, which was an important step towards the discovery of quantum physics. Like for sand falling through two slits, in classical physics one would expect that the probability of electrons passing a two-slit experiment to be just the sum of the probability of passing one slit, and that the resulting intensity on a detector plane was just the sum of the single-slit intensities. The surprising experimental result was that also matter shows interference. There are multiple paths electrons could take, and all these paths contribute with a certain phase which leads to interference patterns as summarised in the path integral.\\
So one fundamental insight of quantum physics is, that one cannot just ``add'' the intensities and probabilities of two single slits to describe two-slit experiments. The reason is interference.\\
Despite the fundamental importance of the two-slit experiment for the discovery of quantum physics, long time has passed until people started to think about three-slit experiments or other multi-slit experiments. One can ask similar questions like one once did for the two-slit experiment: If we know the behaviour of the single-slit and the two-slit experiments, can we also infer the behaviour of the three-slit experiment? Or will there be a no-trivial interference which only appears when at least three slits are involved?\\
The surprising answer is that there is no non-trivial third- or higher-order interference.

\begin{postulate}[Postulate 3][\hspace{-3mm} No third-order interference] \\
	\textit{
	There is no non-trivial third- or higher-order interference.
	}
\end{postulate} 

To state this postulate in an exact way, one needs a lot of technical formalization and abstract definitions. As we will not directly use this postulate, we will not make the effort. The details can be found in \cite{Postulates},\cite{Ududec}, and the original framework of higher-order interference was first introduced by Sorkin \cite{Sorkin}.\\
The last postulates gives rise to Hamiltonian dynamics:\\

\begin{postulate}[Postulate 4][\hspace{-3mm} Observability of energy] \\
	\textit{
	There is non-trivial reversible continuous time evolution and the generator of every such evolution can be associated to an observable (energy) , which is a conserved quantity.
	}
\end{postulate}

This postulate is special in the sense, that most axiomatic derivations of quantum theory do not define any time evolution and also do not talk about time at all, not even in the final results. Just like before, we will not directly use this postulate and thus do not explain it in full detail.

\subsection{First consequences of the postulates} 
The most important result from \cite{Postulates} is, that the 4 postulates single out the state space structure of quantum theory together with the unitary transformations:
\begin{theorem}
	The 4 postulates imply that the state space is an N-level state space of standard complex quantum theory for some $N \in \mathbb N$, and all conjugations $\rho \mapsto U\rho U^\dagger$ with $U\in \rm{SU}(N)$ are contained in the group of reversible transformations.
\end{theorem}
\begin{proof}
	See Theorem 31 from \cite{Postulates}.
\end{proof}

We will now focus on the consequence of the first two postulates, i.e. classical decomposability and strong symmetry. We  will often call them Postulates 1 and 2. So it is interesting to have some examples for non-trivial state spaces that satisfy these postulates. Especially it is important to know that there are non-quantum and non-classical state spaces fulfilling the postulates:
\begin{theorem}
	The possible state spaces satisfying Postulates 1, 2 and 3 which have a non-trivial connected component $\mathcal G_0$ of their reversible transformation groups are the following:
	\begin{enumerate}
		\item The $d$-dimensional ball state spaces $\Omega_d := \{(1,r)^T|r\in \mathbb R^d, ||r||\le 1\}$ with $d\ge 2$, and either $\mathcal G_0 = \rm{SO}(d)$, or $\mathcal G_0 = \rm{SU}(d/2)$ if $d=4,6,8,...$, or $\mathcal G_0 = \rm{U}(d/2)$ if $d=2,4,6,8,...$, or $\mathcal G_0 = \rm{Sp}(d/4)$ if $d= 8,12,16,...$, or $\mathcal G_0 = \rm{Sp}(d/4) \times \rm{U}(1)$ if $d=8,12,16,...$, or $\mathcal G_0 = \rm{Sp}(d/4)\times \rm{SU}(2)$ if $d=4,8,12,...$, or $\mathcal G_0 = \rm{G}_2$ if $d=7$ or $\mathcal G_0 = \rm{Spin}(7)$ if $d=8$ or $\mathcal G_0 =\rm{Spin}(9)$ if $d=16$.
		\item $N$-level real quantum theory with $N \ge 2$ and $\mathcal G_0 = \{\rho \mapsto O\rho O^T | O \in \rm{SO}(N)\}$
		\item $N$-level complex quantum theory with $N \ge 2$ and $\mathcal G_0 = \{\rho \mapsto U \rho U^\dagger | U \in \rm{SU}(N)\}$
		\item $N$-level quaternionic quantum theory with $N\ge 2$ and $\mathcal G_0 \cong \rm{Sp}(N)/\{-1,+1\}$ 
		\item $3$-level octonionic quantum theory with $\mathcal G_0 \cong \rm{F}_4$.
	\end{enumerate}
	However, among those, only the complex quantum theory state spaces (including $\Omega_3$, the qubit) satisfy Postulate 4, that is, observability of energy.
\end{theorem}
\begin{proof}
	Lemma 33 from \cite{Postulates}.
\end{proof}
The previous theorem tells us what the state spaces satisfying the first \textbf{three} postulates are, i.e. one more than we are going to consider. It is not known whether the third postulate is separate from the first two or if there exist state spaces that satisfy Postulates 1 and 2, but not 3.\\
\\
The following definitions and results also are from \cite{Postulates}. The proofs are too long and technical to repeat them here. Also the definitions and results themselves are quite technical but will be needed in this thesis. Therefore, examples from quantum theory are used to explain them. When referring to ``the Hilbert space'', we mean the Hilbert space of pure states. This comparison is possible, because our postulates provide some of the structure that the Hilbert space of pure states provides for the set of all states (i.e. pure and mixed states) in quantum theory.
\\
Bit symmetry, i.e. the special case of strong symmetry for 2-level systems, has the important result that the GPT is self-dual:
\begin{theorem} \label{Theorem:Selfdual}
	Postulates 1 and 2 imply that $A_+$ is self-dual.
	The inner product can be chosen such that all of the following properties hold:
	\begin{enumerate}
		\item $\braket{Tw,Tv} = \braket{w,v}$ for all reversible transformations $T$
		\item $0 \le \braket{w,v} \le 1$ for all $w,v \in \Omega_A$
		\item $\braket{w,w} = 1$ for all pure $w\in \Omega_A$ and $\braket{v,v} < 1$ for all mixed $v \in \Omega_A$
		\item $\braket{w,v} = 0$ for all $w,v \in \Omega_A$ which are perfectly distinguishable. This means that all perfectly distinguishable states are orthogonal. 
	\end{enumerate}			
\end{theorem}
\begin{proof}
	See Proposition 3 from \cite{Postulates}, Theorem 1 from \cite{Skalarprodukt}, Proposition 5.19 from \cite{Ududec}.
\end{proof}
This inner product will be our most important tool. In quantum theory, the self-dualizing inner product on the space $A$ of hermitian operators is given by $\braket{M,B} = \text{Tr}(M^\dagger B) = \text{Tr}(MB)$.\\ 
\\
In quantum theory, orthonormal pure states $\ket{1}\bra{1}$, $\ket{2}\bra{2}$,...,$\ket{n}\bra{n}$ are perfectly distinguishable. We call such a set a $n$-frame and generalize this definition as follows:
\begin{defi}
	A set of $n$ perfectly distinguishable pure states $w_1,...,w_n \in \Omega_A$ is called a \textbf{$n$-frame}.
\end{defi}

Postulate 2 implies that two $n$-frames $w_1$,...,$w_n$ and $v_1$,...,$v_n$ can be reversibly transformed into each other, i.e. there is a reversible transformation $T$ with $T w_j = v_j$. This implies that there is a maximal frame size, and that all smaller frames can be completed to a maximal frame.\\

\begin{lem}
	Let $\rm{dim}(A) = n$, $\{w_j\}_{j \in J}$ be a set of perfectly distinguishable (pure) states. Then $|J| \le n$.
\end{lem}
\begin{proof}
	Let $\{e_j\}_{j \in J}$ be the effects that distinguish the states: $e_k(w_j) = \delta_{kj}$. The $w_j$ are linearly independent, because if they were linearly dependent, we had $w_k = \sum_{j\ne k} a_j w_j$ for some $k$, $a_j \in \mathbb R$ and thus $1=e_k(w_k) =  \sum_{j\ne k} a_j e_k(w_j) = 0$. Thus the number of perfectly distinguishable (pure) states is no larger than the dimension of $A$.
\end{proof}

In quantum theory, all bases contain the same number of states and have therefore the same frame-size. This generalizes to our GPTs:
\begin{lem}
	 All maximal sets of perfectly distinguishable pure states have the same size, that is:\\
	If $\{w_1,...w_m\}$, $\{w_1',...w_n'\}$ are both sets with perfectly distinguishable pure states and $m < n$, we can find pure states $w_{m+1},...,w_n$ such that $\{w_1,...w_n\}$ is perfectly distinguishable. 
\end{lem}
\begin{proof}
	Let $T$ be the reversible transformation (Postulate 2) taking $\{w_1',...w_m'\}$ to $\{w_1,...w_m\}$, i.e. $T(w_j') = w_j \quad \forall j \le m$.
	Let $e_1',...,e_n'$ be the effects with $e_j'(w_k') = \delta_{j,k}$, $\sum_j e_j' \le u_A$. As $T$ is reversible, $T^{-1}$ must be normalization-preserving ($T$ is not allowed to be normalization-increasing), i.e. $u_A \circ T^{-1} = u_A$. For $e_j := e_j' \circ T^{-1}$ we have $\sum_{j=1}^n e_j = \sum_{j=1}^n e_j' \circ T^{-1} \le u_A \circ T^{-1} = u_A$. The $e_j$ are linear as composition of linear functions, and im$(e_j) \subset $ im$(e_j') \subset [0,1]$ on $\Omega_A^{\le 1}$  as $T^{-1}(\Omega_A^{\le 1}) = \Omega_A^{\le 1}$ by reversibility and positivity. I.e. the $e_j$ are also effects.\\
$e_j(T(w_k')) = e_j'(w_k') = \delta_{jk}$, i.e. the $e_j$ perfectly distinguish $\{T(w'_1),...,T(w'_n)\} = \{w_1,...,w_m,T(w_{m+1}'),..., T(w_n')\}$.\\
Hereby, the $T(w_j')$ are also pure: Let $p\in (0,1)$, $T(w_j')= p w + (1-p) w'$. Then by linearity and bijectivity $w_j' = p T^{-1}(w) + (1-p) T^{-1}(w') $. As $w_j'$ pure, $w_j' =  T^{-1}(w) = T^{-1}(w')$. Thus $w = w' = T(w_j')$.
\end{proof} \noindent

Similarly to the bases of sub-Hilbert spaces in quantum theory, also faces can be identified by sets of perfectly distinguishable pure states. The rank generalizes the dimension.   
\begin{prop}
	Postulates 1 and 2 imply that every face of $\Omega_A$ is generated by a frame. Any two frames that generate the same face $F$ have the same size, called the \textbf{rank} of $F$, and denoted by $|F|$. Moreover, if $F \subset G$ and $F\ne G$, then $|F| < |G|$. Every frame of size $|F|$ in $F$ generates $F$.
\end{prop}
\begin{proof}
	Proposition 2 from \cite{Postulates}
\end{proof}

The no-restriction hypothesis is satisfied:
\begin{prop}
	Postulate 1 and 2 imply that all effects are allowed.
\end{prop}
\begin{proof}
	Proposition 1 from \cite{Postulates}.
\end{proof}

The following proposition is analogous to the fact, that every orthonormal set (of pure states) can be extended to an orthonormal basis of the Hilbert space.
\begin{prop}
	Postulates 1 and 2 imply that every frame $w_1,...,w_n$ can be extended to a frame $w_1,...,w_{|A_+|}$ that generates $A_+$.
\end{prop}
\begin{proof}
	Proposition 5 from \cite{Postulates}
\end{proof}

Orthogonal pure states already are perfectly distinguishable, i.e. for pure states, orthogonality and perfect distinguishability are equivalent. Furthermore, maximal frames define measurements that are similar to POVM constructed from a non-degenerate projective measurement in quantum theory. Removing some frame states corresponds to leaving out some measurement results and thus gives a subnormalized measurement:
\begin{prop} \label{Prop:FramesAddUp}
	Postulates 1 and 2 imply that if $w_1,...,w_n$ are mutually orthogonal pure states, then they are a frame and $\sum_{j=1}^n \braket{w_j,\cdot} \le u_A$. Furthermore, every maximal frame $w_1,...,w_{|A_+|}$ adds up to the order unit, i.e. $\sum_{j=1}^{|A_+|} \braket{w_j,\cdot} = u_A$.
\end{prop}
\begin{proof}
	Proposition 6 from \cite{Postulates}
\end{proof}

Also for faces it is possible to extend frames to generating frames:
\begin{prop} Suppose Postulates 1 and 2 are satisfied.
	If $w_1,...,w_n$ is any frame contained in some face $F$ of $A_+$, then it can be extended to a frame $w_1,...,w_{|F|}$ of $F$ which generates $F$.
\end{prop}
\begin{proof}
	Proposition 7 from \cite{Postulates}.
\end{proof}

Just like the projectors in quantum theory, also the projectors considered here map states to states:
\begin{prop}
	Postulates 1 and 2 imply that for every face $F$ of $A_+$, the orthogonal projection $P_F$ onto the linear span of $F$ is positive. \label{Prop:PositiveProjector}
\end{prop}
\begin{proof}
	Theorem 8 from \cite{Postulates}.
\end{proof}

Now consider a projective measurement in quantum theory with orthogonal projectors $P_1,...,P_m$. They map states according to $\rho \mapsto P_j \rho P_j$. Then $\rm{Tr}(P_j \rho P_j) =\rm{Tr}(\rho P_j)$ is the probability that after the measurement, the state is found in the subspace given by $P_j$. A generalization of the functional $\rm{Tr}(P_j \ \cdot)$ is defined as follows:
\begin{defi}
	Let $A$ be any system satisfying Postulates 1 and 2. Then, to every face $F$ of $A_+$, define the \textbf{projective unit} $u_F$ as
	\begin{equation}
		u_F := u_A \circ P_F
	\end{equation}
	where $P_F$ is the orthogonal projection onto the linear span of $F$.
\end{defi}

By using the self-duality and the symmetry of $P_F$, $\braket{u_F,w} = u_F(w) = u_A\circ P_F (w) = u_A(P_F w) =  \braket{u_A, P_F w} = \braket{P_F u_A,w}$. Thus one can also write $u_F = P_F u_A$, which is the original definition from \cite{Postulates}.\\
\\
In quantum theory, $\rm{Tr}(P_j \ \cdot)$ are well-defined probability functionals. Furthermore, for a spectral decomposition $P_j = \sum_k \ket{k_j}\bra{k_j}$ we find $\rm{Tr}(P_j \ \cdot) = \sum_k \rm{Tr}(\ket{k_j}\bra{k_j} \cdot )$, i.e. $\rm{Tr}(P_j \ \cdot)= \sum_k \ket{k_j}\bra{k_j}$ using self-duality. Here, the $\ket{k_j}$ span the subspace $P_j$ projects onto. Furthermore, two projectors can only appear in a common measurement if they are orthogonal. Similar results also hold for our generalizations:

\begin{prop} \label{Prop:ProjectiveUnit}
	Let $A$ be any system satisfying Postulates 1 and 2. $u_F$ is an effect $0\le u_F \le u_A$ with $u_F(w) = 1 \ \forall w\in F \cap \Omega_A$. If $w_1,...,w_{|F|}$ is any frame that generates $F$, then
	\begin{equation}
		u_F = \sum_{j=1}^{|F|} w_j
	\end{equation}
	We have $u_F + u_G \le u_A$ if and only if $F$ and $G$ are orthogonal.
\end{prop}
\begin{proof}
	Lemma 11 from \cite{Postulates}.
\end{proof}

In quantum theory, a projective measurement $P_1,...,P_m$ only gives a predictable outcome $k$ if the considered system is already found in the subspace onto which $P_k$ projects. Something similar holds for our generalization:

\begin{prop}
	Assume Postulates 1 and 2. Then every face $F$ of the set of normalized states $\Omega_A$ can be written as:
	\begin{equation}
		F = \{ w \in \Omega_A | \braket{u_{[ \mathbb R_{\ge 0}\cdot F]},w}=1 \}
	\end{equation}
\end{prop}
\begin{proof}
	Proposition 5.29 from \cite{Ududec}.
\end{proof}
\section{Von Neumann's thought experiment}
\label{Section:Neumann}
\subsection{The plan}
In quantum theory, consider a density operator $\rho$ with orthonormal eigenbasis $\ket{j}$ and eigenvalues $p_j$, i.e. $\rho = \sum_j p_j \ket{j}\bra{j}$.\\
The von Neumann entropy is defined as 
\begin{equation}
	S(\rho) = - k_B \sum_j p_j \ln p_j
\end{equation}
where $0 \ln 0 := 0$ by continuity. In quantum theory, all pure states are of the form $\ket{j}\bra{j}$. As the eigenstates are orthogonal, they are perfectly distinguishable.\\
Now we consider GPTs which satisfy Postulates 1 and 2. For any state $w \in \Omega_A$, we consider a classical decomposition $w=\sum_j p_j w_j$ with $w_j$ pure and perfectly distinguishable. A natural generalization of the von Neumann entropy is:
\begin{equation}
	S(w) = -k_B \sum_j p_j \ln p_j
\end{equation}
Of course, the analogy to quantum theory is enough to motivate why this entropy definition is natural and interesting. However, it is important for us that von Neumann obtained his entropy by thermodynamic considerations\cite{Neumann}. Realizing an idea by J. Barrett \cite{Conference}, we will see that these considerations can be applied to many other GPTs as well. While it is relatively easy to introduce operationally/information-theoretically motivated entropies in GPTs (see e.g. \cite{WehnerEntropy}, \cite{BarnumEntropy}), there is no straight-forward way to introduce thermodynamics to GPTs. Thus von Neumann's thought experiment is an important step to provide a deep connection between information theory and thermodynamics also for other GPTs.

\subsection{Combining GPTs and ideal gases}
We consider a GPT ensemble $[S_1,S_2,...,S_N]$, which we will call the $S_j$-ensemble. Now we consider the following trick introduced by Einstein and applied by von Neumann: Imagine we take $N$ hollow, small boxes $K_1$, $K_2$,$...$, $K_N$ and put one of the systems into each of the boxes. We do this in a such way that the internal system $S_j$ has no interaction with its box $K_j$ or anything else - the systems are completely isolated, such that they are not perturbed and thus the ensemble is not changed. The boxes are assumed to form an ideal classical gas. The internal GPT systems have no impact on the behaviour of the classical gas, as the internal systems are hidden from any interaction; that is except for two steps, where we will actually open the box and measure its content or transform it.\\ \\
So we have a classical ideal gas whose particles function as the carriers of internal, passive GPT systems. A key idea is that the inner state will behave like a classical label. It sounds absurd because something like that would be extremely hard to realize in experiment - but this is a thought experiment, and thermodynamics should also be capable of describing such well-defined thought experiments.\\
\\
The basic idea for the derivation of the thermodynamic entropy of a GPT ensemble is that of consistency: We will perform a reversible operation. We already know how the classical gas and its entropy will behave, and the difference between the total entropy change and the entropy change of the classical gas (or the heat reservoir) must be caused by the GPT ensemble. 

\subsection{Relation between the entropies of the GPT ensemble and the gas}
We assume we have a $w$-ensemble, where $w$ is the state of $[S_1,...,S_N]$, and a $v$-ensemble, where $v$ is the state of $[S'_1,...,S'_N]$. Then later on, we need to be sure, that the entropy difference between the gases is the same as the entropy difference between the internal GPT ensembles, if both gases are considered at the same conditions (i.e. same temperature and same volume of the tank). The idea is, that the gases are almost equal, the only difference in entropy being caused by the internal GPT entropy.\\
So if we consider low temperatures, the movement of the boxes freezes out and the entropy of the gas is just given by its internal GPT ensemble. In this limit, we can imagine the gas as just a bunch of GPT systems that do not see each other, which is how ensembles are typically introduced in textbooks. So in this limit, the statement is true. \\
Now we heat the two gases to the same, arbitrary temperature $T$. The boxes are assumed to be completely equal. And as the internal GPT systems are completely isolated, only the boxes can take any work or heat, while the internal GPT system is unaffected. Thus both gases have the same specific heat $C_V = \frac{\delta Q}{\delta T}$. So in order to heat them by $\delta T$, both need the same heat $\delta Q$. As $dS = \frac{\delta Q}{T}$, both gases have the same change in entropy. Thus the entropy difference is still given by the GPT ensemble.\\
Therefore for two such gases with same $T,V,N$: 
\begin{equation}
	S_{w\text{-gas}} - S_{v\text{-gas}} = S_{w\text{-ensemble}} - S_{v\text{-ensemble}}
\end{equation}

\subsection{Tool 1: Semipermeable membranes}
For von Neumann's reasoning, one needs semipermeable membranes. In quantum theory, if we start with a density operator $w$, then we can diagonalize it. Let $w = \sum_j p_j \ket{j}\bra{j}$ with $\ket{j}$ orthonormal. We can realize $w$ by preparing a lot of quantum systems, $p_j$ being the probability that the system is prepared in the state $\ket{j}$. This ensemble now is our $[S_1,...,S_N]$-ensemble. As the $\ket{j}$ are orthonormal, one can imagine a semipermeable membrane: This membrane opens the boxes of incoming particles and measures the internal quantum state. This measurement is a projective measurement in the orthonormal basis $\ket{j}$. We know that this measurement does not perturb the internal quantum state and always gives the right result. Depending on the state, the box is allowed to pass (a window opens) or is reflected (window remains closed). Von Neumann also gives a thermodynamic reasoning, that such a semipermeable membrane can only exist for orthogonal states. As this is also a standard result from quantum information theory, we will not reproduce it here. Also it is enough to know that there exists one preparation procedure for which the single-system states can be distinguished by a semipermable membrane. \\
\\
The orthonormal states $\ket{j}\bra{j}$ from above have the important property, that we can distinguish them without perturbation or error. I.e. if know that the system is one of the states $\ket{j}\bra{j}$, we can find out for sure in which one of the states the systems is, and we can do so without destroying the state. This reminds of classical physics\footnote{The convex hull of such perfectly distinguishable pure states forms a simplex and thus is a classical system in the sense of GPTs}. The property, that a (mixed) state can be prepared by using only perfectly distinguishable pure states is thus called classical decomposability. Reproducing the von Neumann argument for more general GPTs is the reason, why we are so much interested in this postulate.\\
\\
Thus if a GPT fulfils Postulate 1 we can prepare arbitrary (mixed) states by only using perfectly distinguishable pure states $w_j$, which replace the eigenbasis $\ket{j}$ from the quantum case. Then we can consider a semipermeable membrane, which uses the effects $e_i$ with $e_i(w_j) = \delta_{ij}$ to find out the internal state of the box. One moment of thought is needed considering post-measurement states: We don't want the semipermeable membrane to change the internal state of the boxes. However, Postulates 1 and 2 do not make any statement about post-measurement states. Thus we have to add the additional assumption, that a perfectly distinguishing measurement can be implemented without disturbing the states it distinguishes. This assumption is well motivated. First of all, one could assume that the membrane prepares the box in the same state it just has measured before, undoing any perturbation caused by the measurement. In more details, we might consider a measurement described by the operation $T_1,...,T_n$ with $u_A \circ T_j(w_k) = \delta_{jk}$ and $T_j w_k =: \delta_{jk} v_j$. We assume that the measurement is perfect/noiseless in the sense that the pure $w_j$ are mapped to pure states, i.e. $v_j$ pure. Then by Postulate 2, there is a reversible transformation $T_j'$ with $T_j' v_j = w_j$. As $T_j'$ is reversible, $u_A \circ T_j' = u_A$. Thus the operation $T_1'\circ T_1,...,T_n'\circ T_n$ induces the same measurement, but it does not disturb the $w_j$. Also we will see in Chapter \ref{Section:ProjectiveMeasurements}, that projective measurements in the style of quantum theory can be defined. It is reasonable, that the operations are described by projective measurements, as they should be repeatable. This is also the motivation for Pfister's one simple postulate\cite{Pfister}: \\
\begin{postulate}[Postulate][\hspace{-3mm} Pfister's one simple postulate] \\
\textit{
If we can predict the outcome of a measurement with certainty, we can perform the measurement without altering the state: Let $M = \{e_1,...,e_n\}$ be a pure measurement with corresponding operation $\{T_1,...,T_n\}$.  If $w$ is a state with certain outcome, i.e. $e_k(w) = 1$ for some k, then $T_k(w) = w$.
}
\end{postulate}

We will implement this assumption by refining the postulate of classical decomposability. The basic idea of the classical decomposability is that every state is part of a classical subspace and that the non-classical properties only exist because there are different classical subspaces. Thus if we stay within a classical subspace, the key idea is that this subspace behaves classically. Especially, a measurement which perfectly distinguishes the pure states can be implemented such that it does not disturb the pure state. Thus we postulate:
\begin{postulate}[Postulate 1'][\hspace{-3mm} Classical decomposability and classical behaviour of classical subspaces]\\
\textit{
	For any state $w\in \Omega_A$, there exists a probability distribution $p_1$,...,$p_n$ and perfectly distinguishable pure states $w_1$,...,$w_n$ with $w = \sum_{j=1}^n p_j w_j$. A measurement which perfectly distinguishes the pure states can be implemented by a physically allowed operation $T_1,...,T_n$ which does not disturb the pure states: $T_j w_k = \delta_{jk} w_k$.
	}
\end{postulate}
Furthermore we note that this refined postulate is only necessary for the thought experiment. For the mathematical definition of the entropy and the proof of its properties, Postulates 1 and 2 from \cite{Postulates} will be sufficient. As already mentioned, we will construct projective measurements in Chapter \ref{Section:ProjectiveMeasurements}. Thus the new postulate is not stronger than the original postulate concerning the state space structure. However, it does tell us that the projective measurement or a measurement with similar non-destructive properties is physically allowed. Thus we know that even without the stronger postulate, there is at least one mathematically well-defined operation to perfectly distinguish a frame. The stronger version of this postulate thus just adds that this or a similar operation is physically allowed, but does not add anything to the state space structure.\\
\\
Long story short:\\
Postulate 1, refined by the assumption of non-destructive distinguishing measurements, is enough to obtain a semipermeable membrane. 
  
\subsection{Tool 2: Reversible pure state conversions}
In the end we will transform pure states reversibly into other pure states. In quantum physics, this can be achieved by unitary time evolution $e^{i H t}$ or a complicated sequence of infinitely many measurements with infinitesimal perturbation.\\
We will simply use Postulate 2 from \cite{Postulates}. In fact, we only need a weaker form of this postulate. For us it is enough to know that any pure state can be reversibly transformed to any other one.

\subsection{The main argument: Deriving the entropy}
\begin{center}
\includegraphics[width=\textwidth]{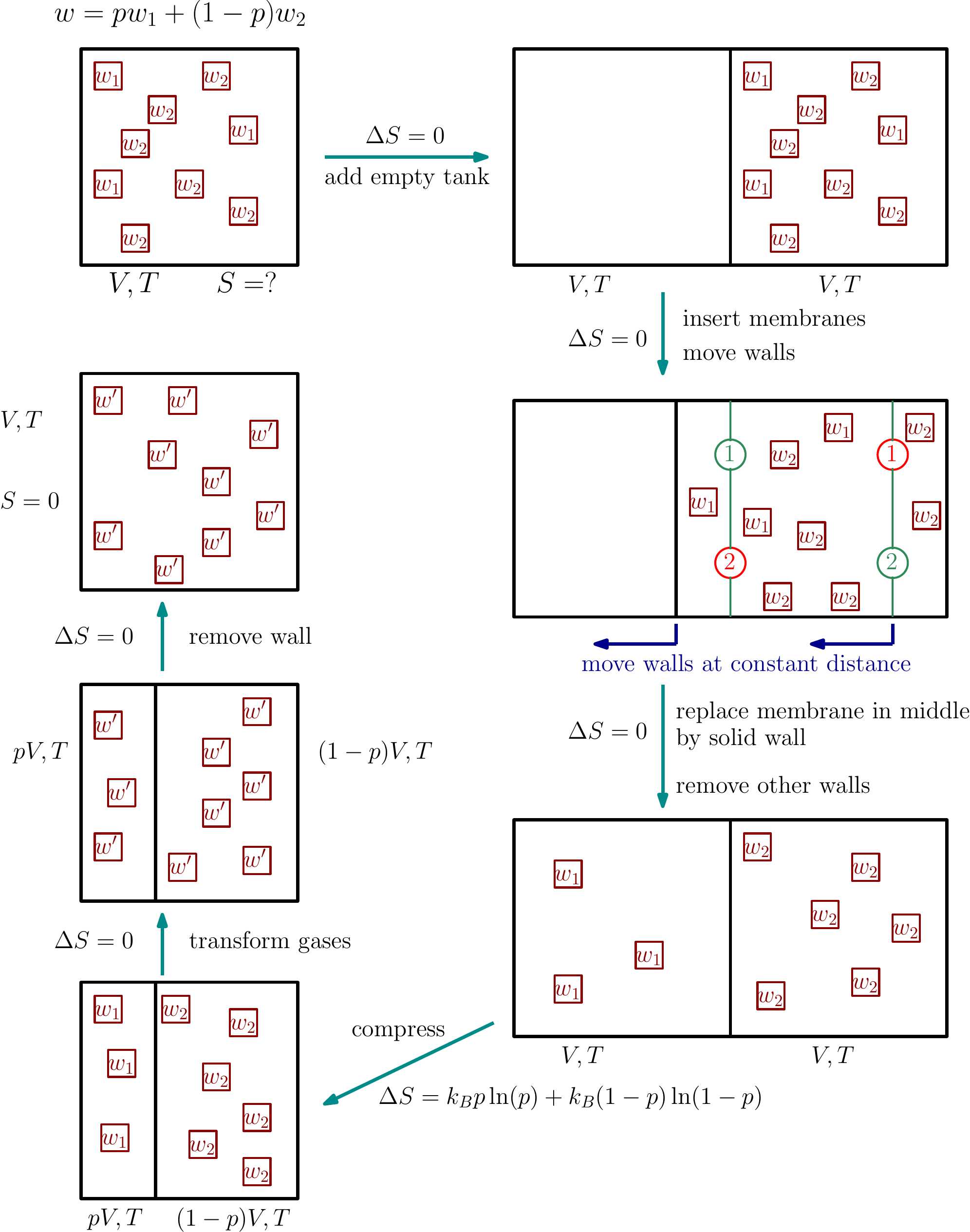}
\captionof{figure}{\small \textit{This figure visualizes all steps of von Neumann's thought experiment.}}
\label{Fig:vonNeumannTotal}
\end{center}
Now we can finally perform von Neumann's thought experiment. All steps of the thought experiment are shown in Figure \ref{Fig:vonNeumannTotal}. It is important to recognize that, except for Tool 1 and 2, the underlying GPT, from which the internal ensemble is taken, plays no role. The reason is that using the semipermeable membrane and the reversible state conversion, the inner GPT-state just behaves like a classical label. This is an important consequence of the classical decomposability, i.e. the idea that every system belongs to a classical subspace.\\
\\
\begin{center}
\includegraphics[width=0.9\textwidth]{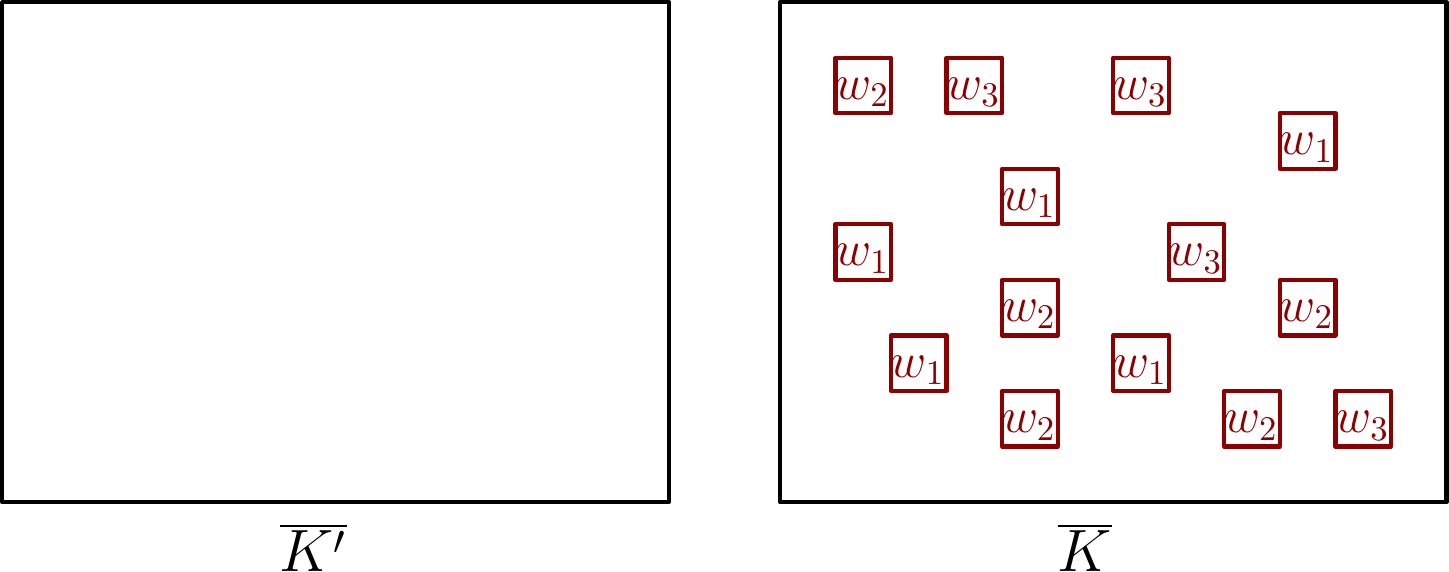}
\captionof{figure}{\small \textit{Tank $\overline K$ contains our gas. In the beginning, the tank $\overline{K'}$, which is a clone of  $\overline{K}$, is empty.}}
\label{Fig:TanksBeginning}
\end{center}
We assume we start with a $w$-ensemble $[S_1,...S_N]$. Following Postulate 1', we will assume that the systems $S_j$ are prepared using only pure distinguishable states $w_j$, i.e. $w = \sum_j p_j w_j$. Choosing $N$ very large, we can assume that $p_j \cdot N$ systems have the internal state $w_j$. Thus we have access to Tool 1, the semipermeable membrane. This ensemble is implanted into a classical gas at temperature T confined in tank $\overline K$ of volume $V$. To the left, we add a tank $\overline{K'}$ of the same volume $V$, but empty (vacuum), see Figure \ref{Fig:TanksBeginning}.

We assume that we have two neighbouring walls separating the two tanks. The wall to the left is a standard wall, not letting through anything. We call it $1$. The wall to the right is semipermeable (Tool 1): The boxes with internal state $w_1$ can pass through the semipermeable membrane, the other ones are reflected. This wall we call $2$. Furthermore, we have another semipermeable membrane (Tool 1) at the right end of tank $\overline K$: It is transparent for all $w_j$ with $j \ne 1$  and only reflects $w_1$. This wall we call $3$. The whole situation is shown in Figure \ref{Fig:Walls}.\\
\\
Now we move the standard wall 1 and the right semipermeable membrane (i.e. 3) to the left while keeping them at constant distance. We do so until wall $1$ collides with the left end of tank $\overline{K'}$.\\
The boxes with $w_j$, $j \ne 1$ are not influenced by this procedure at all. As the walls are moved at same velocity, the $w_1$-gas is also kept at constant volume, and we do not need to perform any work. The basic idea is, that the $w_2,w_3,...$-gas has the same pressure on wall 3 from both sides and thus can be neglected. The pressure of the $w_1$-gas one has to work against at wall 3 (right), is the same pressure that moves wall 1 (left), thus here the energy difference is also zero (the work needed at the right end can be regathered from the left end).\\ This way of arguing is justified by Dalton's law (see e.g. \cite{Schneider} or \cite{NoltingThermo} Chapter 3.5): Different types of ideal gases behave, as if they were alone: For a gas which can access a container of volume $V$, the partial pressure of that gas is given by $p V= N k_B T$, where $N$ is the number of particles of that gas. The total pressure is given by the sum of all the partial pressures of the different types of gases. This law is a consequence of the fact, that the Hamiltonian of the ideal gas is modelled to include no particle-particle-interaction.

\begin{center}
\includegraphics[width=0.9\textwidth]{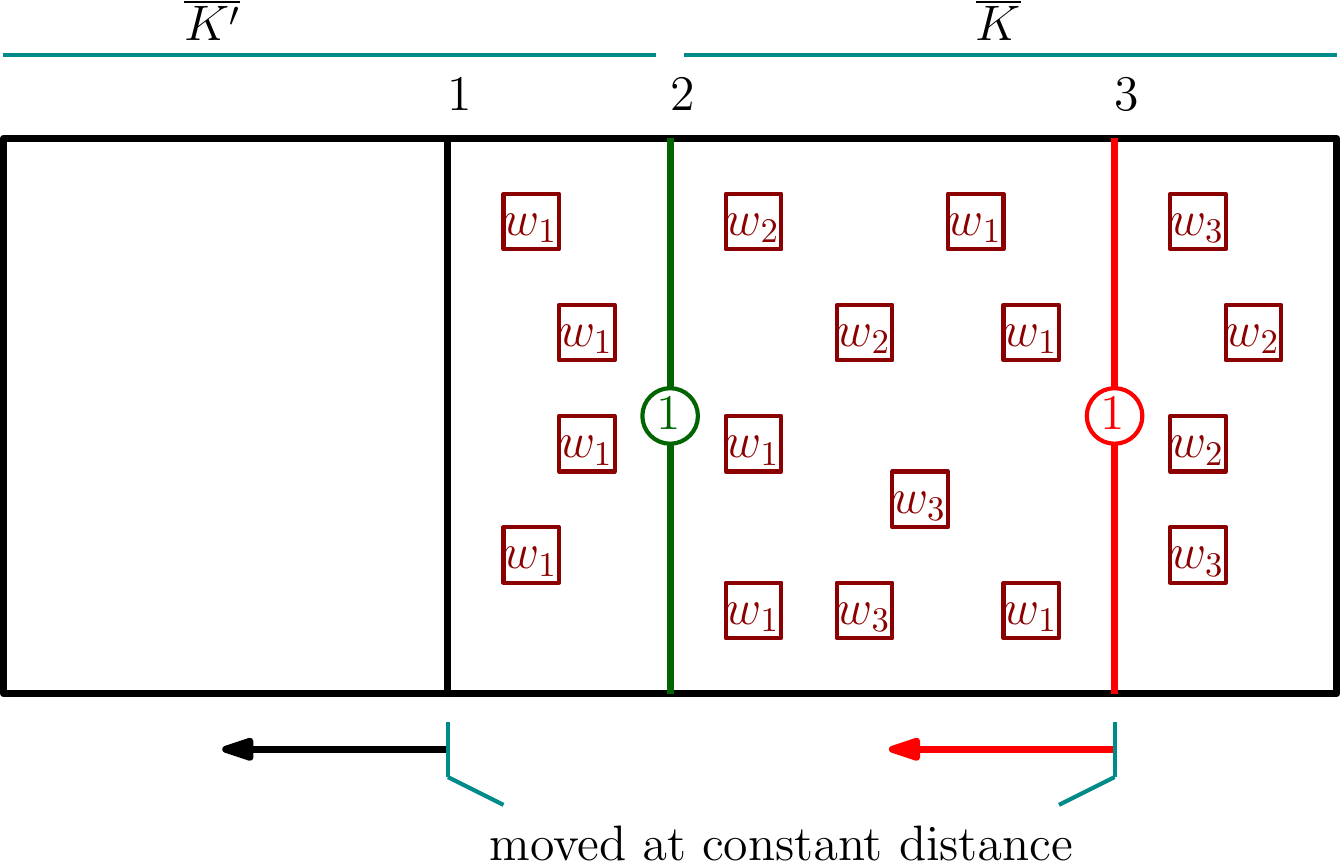}
\captionof{figure}{\textit{The $w_1$-boxes are separated from the rest by using three walls: The green wall inscribed with label 1 is a semipermeable wall which lets only $w_1$-boxes pass. The red wall with label 3 is a semipermeable wall which lets all $w_j$ pass, except $w_1$ (green $\rightarrow$ go, red $\rightarrow$ no-go). Wall 3 and the standard wall 1 are moved to the left at constant distance.}}
\label{Fig:Walls}
\end{center}

Now, all $w_1$-boxes are in $\overline{K'}$, while all the other boxes are in $\overline K$. We separate the two tanks.\\
\\
Thus now we have (reversibly, without any work or heat exchange!) successfully isolated the $w_1$-gas from the rest. We repeat this procedure so often, that each $w_j$-gas ends up in its own tank.\\
\\
Now we isothermally compress each tank to the volume $p_j \cdot V$, shown in Figure \ref{Fig:Compress}.

\begin{center}
\includegraphics[width=\textwidth]{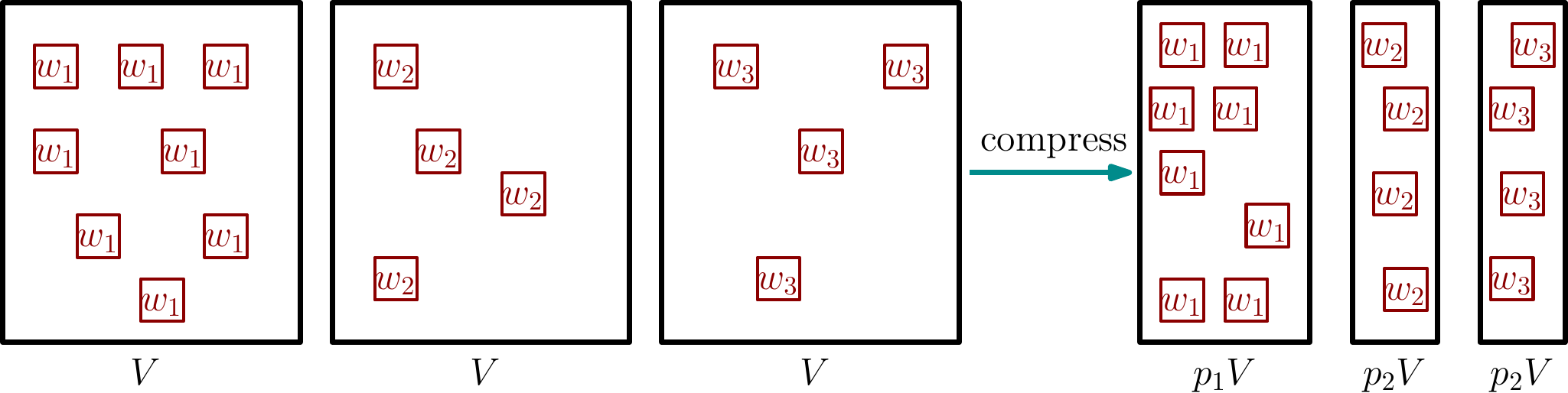}
\captionof{figure}{\small \textit{Each tank containing a $w_j$-gas is compressed to the size $p_j V$.}}
\label{Fig:Compress}
\end{center}

The work needed for this is:\\
\begin{align*}
	\delta W &= -\sum_j \int\limits_{V}^{p_j V} p \rm{d} V' = -\sum_j \int\limits_{V}^{p_j V} \frac{p_j N k_B T}{V'} \rm{d} V' = -\sum_j \Big[ p_j N k_B T \ln (V') \Big]_{V}^{p_j V} \\
		&= -\sum_j p_j N k_B T [\ln(p_j V) -\ln(V)]\\
		&= -N k_B T \sum_j p_j \ln(p_j)
\end{align*} 
 The mean energy $E \propto T = const$ is constant in isothermal procedures. Thus the work performed on the gas is given as heat to the heat reservoir. 
Following d$S = \frac{\delta Q}{T}$, the entropy of the reservoir increases by $-N k_B \sum_j p_j \ln(p_j)$, vice versa the entropy of the bunch of gases 
is increased by $N k_B \sum_j p_j \ln(p_j)$ (negative !), as the collection of gases loses the heat $\delta Q$. \\
\\
Now we apply Tool 2: All the gases are reversibly transformed into the same pure state $w'$. These gases all have the same density $\frac{p_j N}{p_j V} = \frac N V$. We define that an ensemble whose systems all have the same pure state has entropy 0. We can do this, because all pure states can reversibly be transformed into each other (Tool 2), i.e. no entropy change here. It makes sense to define the pure ensemble-entropy as 0: All particles have the same label, it is trivial.\\
\\
The last step is that we merge all the tanks to one tank of volume $V$ and take away the separating walls, see Figure \ref{Fig:end}. Of course, we can put the walls back in, no entropy change here as all the tanks contained the same gases at same density anyway. \\
Overall, we have reversibly transformed our original $w$-gas to a pure gas at same temperature and volume. The entropy change of the gas was $N k_B \sum_j p_j \ln(p_j)$. As we have already reasoned, this entropy change is the entropy change of the internal ensemble. As now the internal ensemble has entropy 0, our original $w-$ensemble had the entropy
\begin{equation*}
	S_\text{GPT} = - N k_B \sum_j p_j \ln(p_j)
\end{equation*}
\begin{center}
\includegraphics[width=0.8\textwidth]{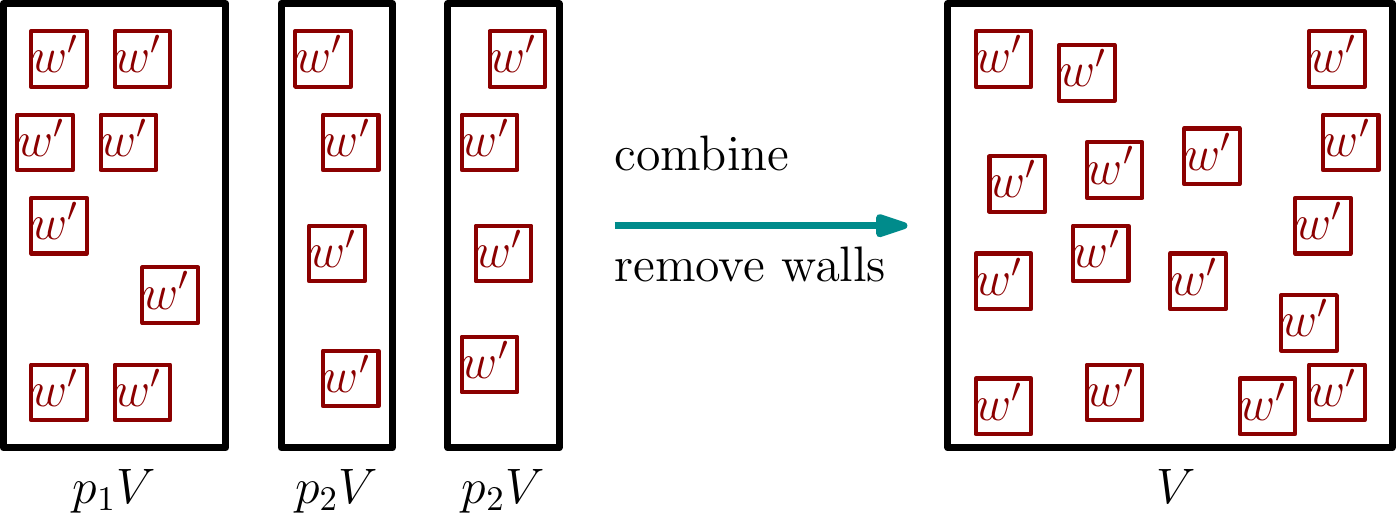}
\captionof{figure}{\small \textit{In the end, all tanks contain the same pure gas and are merged to one tank of size V. The resulting gas differs from the original $w-$gas only in its internal GPT state and the fact, that the entropy has changed by $N k_B \sum_j p_j \ln(p_j)$.}}
\label{Fig:end}
\end{center}

For the special case of quantum theory, $p_j$ are the eigenvalues of the density operator $w$ and we can also write $S_\text{QT} = - N k_B \text{Tr}(w \ln w)$. Furthermore, we can consider the entropy per particle
\begin{equation*}
	s_\text{GPT} = \frac{S_\text{GPT}}{N} = - k_B \sum_j p_j \ln(p_j)
\end{equation*}

\subsection{Entropy from combinatorial considerations}
The derivation of the entropy was based on a purely thermodynamic argument, not using any combinatorial arguments as introduced by Boltzmann in statistical physics. We will now give a short combinatorial argument for an isolated system which gives us the same entropy equation as in the thermodynamic derivation. Arguments of that form are often used in statistical physics, e.g. a related argument can be found in \cite{Schwabl} or \cite{Daijiro}.\\
\\
Once more we consider a GPT-ensemble realized with perfectly distinguishable pure states, each system being put into a small box. There are $N$ systems in total, $N_j$ of them in the state $w_j$. Once more we assume that these boxes form a classical ideal gas. However, this time we assume that the container is perfectly isolated instead of being surrounded by a heat bath.\\
The basic idea is that the inner GPT-states behave like a classical label, e.g. instead of distinguishing between different GPT-states, we could have boxes in different colors or different molecules.\\
\\
The number of states in this situation is given by 
\begin{equation}
	\Omega_\text{total} = \Omega_\text{gas} \cdot \Omega_\text{GPT-configuration} \label{eq:OmegaFaktoren}
\end{equation}

There are two slightly different points of view with same result, see also Figure \ref{Fig:BoltzmannTwoViews}.\\

\begin{center}
	\includegraphics[width=\textwidth]{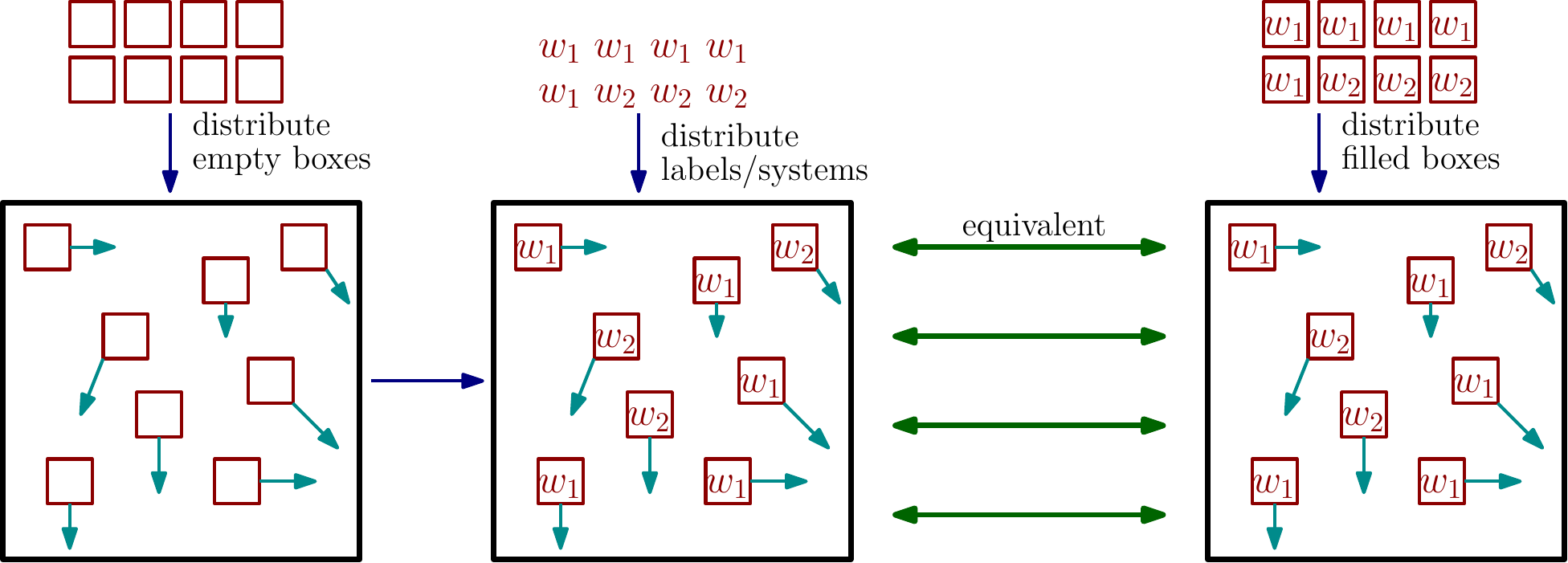}
	\captionof{figure}{\textit{\small This figure illustrates the two different points of view of the combinatorial argument. In the first one, we start by distributing empty (i.e. label-less) boxes across the phase space, and then put systems (i.e. labels) into the boxes. In the second one, we start with boxes containing systems/labels and distribute them in the phase space.}}
	\label{Fig:BoltzmannTwoViews}
\end{center}

The first point of view is that we start with empty boxes. They are indistinguishable, and $\Omega_\text{gas}$ is the standard number of states for a monoatomic ideal gas without labels (see e.g. \cite{Nolting} Equations (1.129)(1.71)), i.e.: 
\begin{equation}
  \Omega_\text{gas} = \int\limits_{E_0 < E < E_0 + \delta E} \frac{\mathrm d^{3N}p \ \mathrm d^{3N}q}{N! \mathrm{h}^{3N}} 
\end{equation}
  Now, we put the GPT-systems into the boxes. As the boxes can be distinguished by their position and momentum (i.e. the phase space coordinates), there are $\Omega_\text{GPT-configuration} = \frac{N!}{N_1! N_2!...}$ ways to do this.\\
  Thus with Equation (\ref{eq:OmegaFaktoren}):
  \begin{equation}
  \Omega_\text{total} = \int\limits_{E_0 < E < E_0 + \delta E} \frac{\mathrm d^{3N}p \ \mathrm d^{3N}q}{\mathrm{h}^{3N}N_1! N_2!...} \label{eq:combi}
\end{equation}
The other point of view is to directly incorporate the idea that the inner GPT-states serve as classical label, see e.g. \cite{Nolting} Equations (1.130)(1.71). Then some of the boxes are distinguishable and we directly find Equation (\ref{eq:combi}).\\
\\
Next, we use Stirling's formula 
\begin{equation}
  \ln(n!) \approx n \ln(n) -n
\end{equation}  
together with $\sum_j N_j = N$ and $p_j = \frac{N_j}{N}$, the probability that a random system in the ensemble is found in the state $w_j$:

\begin{align}
	S_\text{total} &= k_B \ln(\Omega_\text{total}) = k_B \ln(\Omega_\text{gas}) + k_B \ln(\Omega_\text{GPT-configuration})\\
		&= S_\text{gas} + k_B \ln(N!)-k_B\sum_j \ln(N_j!)\\
		& \approx S_\text{gas} + k_B N \ln N - k_B N - k_B \sum_j N_j \ln N_j + k_B \sum_j N_j\\
		&= S_\text{gas} - k_B \sum_j N_j (\ln(N_j) - \ln(N)) = S_\text{gas} - k_B N \sum_j \frac{N_j}{N}\ln\left(\frac{N_j}{N}\right)\\
		&= S_\text{gas} - N k_B \sum_j p_j \ln p_j
\end{align}
Thus again we find 
\begin{align}
	S_\text{GPT} = - N k_B \sum_j p_j \ln p_j
	&&s_\text{GPT} = -k_B \sum_j p_j \ln p_j
\end{align}

The assumption that the $w_j$ are perfectly distinguishable entered by the analogy with classical labels: If the $w_j$ were not perfectly distinguishable, we could not distinguish between all the boxes that contain different $w_j$ at the same time. Especially it would not be possible for an experimenter to find out how many boxes are of type $w_1$, how many of type $w_2$,...\\
For example in a gbit, all states are described by statistical mixtures of the corners. And these corners are pairwise distinguishable, however they are not perfectly distinguishable as a whole \cite{Pfister}. Let the corners be called $w_1,w_2,w_3,w_4$. If we know we get either $w_i$ or $w_j$, $i,j \in \{1,2,3,4\}$ and $i\ne j$ , we always can find out which one it is. But if we are only told it is one of the states $w_1$, $w_2$, $w_3$ or $w_4$, it is not possible to find out which one of these states we got. There is no combinatorial rule from statistical mechanics how to count such ``semi-distinguishable'' configurations, that involve pairwise distinguishable states which are not distinguishable as a whole. The important property of gbits is that there are states without classical decomposition.\\
\\
So we can see in this example, that even for more general GPTs, the considerations from statistical physics and thermodynamics still agree.\\
\\
For the rest of Chapter \ref{Section:Neumann}, $S$ will denote the entropy already divided by the number of systems, i.e. the entropy we called $s_\text{GPT}$ before. 

\subsection{Entropy in classical and quantum physics}
In classical physics, the continuous generalization 
\begin{equation}
	S(\rho) = -k_B \int \mathrm d \Gamma \rho \ln \rho,
\end{equation}
where $\mathrm d \Gamma$ denotes the phase space integral, is called the Gibbs entropy and is the usual entropy in classical equilibrium statistical physics, see \cite{Schwabl} Equation (10.6.5). \\
Like in the classical case, $S(\rho) = -k_B \text{Tr}(\rho \ln \rho)$ is also known as the Gibbs entropy and used for equilibrium physics, see \cite{Schwabl} Equation (10.6.1).

\subsection{Consistency}
The equation which we derived for the entropy depends on the coefficients of the classical decomposition. To be exact, so far we have only derived the entropy of a particular realisation of a state $w$ (via a particular ensemble). In the worst case, the same state might have different classical decompositions with different entropies. This would mean, that the ``state'' $w$ is not a complete description of the the thermodynamic properties. ``The state does not describe the state''. Adapting a proof found by Howard Barnum and Markus Mueller\cite{Erloesung} that the coefficients of a classical decomposition majorize the coefficients of all convex decompositions of same size, we can show, that all GPTs satisfying Postulates 1+2 give rise to consistent thermodynamic entropies.\\

\begin{theorem}
	By Postulate 1 and 2 from \cite{Postulates}, the entropy of a state $w$ is well-defined, i.e. it does not depend on the choice of classical decomposition.
\end{theorem}
\begin{proof}
According to Theorem \ref{Theorem:Selfdual} (or \cite{Skalarprodukt}, Proposition 3 from \cite{Postulates}), there is an inner product $\braket{\cdot,\cdot }$.\\
Let $w = \sum_{j=1}^m p_j w_j = \sum_{j=1}^n q_j w_j'$ be a state with two classical decompositions, i.e. the $w_j$ are pure and perfectly distinguishable, the $p_j$ form a probability distribution, analogously for the $w_k'$ and $q_j$. Wlog we assume that the $w_j$ and $w_k'$ are frames of maximal size ( adding terms of the form $0\cdot \ln(0) = 0$ does not change the entropy). We will now adapt the proof from \cite{Erloesung} for our own purpose.\\
For perfectly distinguishable states $a,b$ we have $\braket{a,b} = 0$. For pure states $w_j'$ we have $\braket{w_j',w_j'} = 1$. Thus
\begin{equation}
	q_i = \braket{w_i', w} = \sum_j p_j \cdot \braket{w_i',w_j}
\end{equation}
We define $r_{ij} := \braket{w_i',w_j}$, probability vectors $q,p$ given by $\{q_i\}_i,\{p_i\}_i$ and a matrix $R=(r_{ij})_{i,j}$ such that $q=R\cdot p$.\\
Now we show that $R$ is doubly-stochastic:\\
The state-space is self-dual, and by Proposition \ref{Prop:FramesAddUp} (or Proposition 6 from \cite{Postulates}) we have $\sum_j \braket{w_j',\cdot} \le u_A$, $\sum_i \braket{w_i,\cdot } \le u_A$. Every maximal frame adds up to the order unit, i.e. we have $\sum_j \braket{w_j',\cdot} = u_A$, $\sum_i \braket{w_i,\cdot } = u_A$. Thus
\begin{equation}\sum_j r_{ij} = \sum_j \braket{w_i',w_j} = u_A(w_i') = 1\end{equation}
and also $\sum_i r_{ij} = 1$. For all states $a,b$ we have $\braket{a,b} \ge 0$, i.e. $r_{ij} \ge 0$.\\
Thus $q \prec p$ by a theorem of Hardy, Littlewood and Polya (see e.g. Lemma I.2.B.2 of \cite{Majorization}) and in the same way one shows $q \succ p$. Following \cite{Majorization}, eq. I.1.B.(2),  this implies that $q$ is a permutation of $p$, and thus the entropies agree. 
\end{proof}
\section{Generalized thought experiment}
\label{Section:Petz}
So far, we have only considered the von Neumann argument for a decomposition into perfectly distinguishable \textbf{pure} states. The interpretation behind that is: In QT, the pure states are of the form $\ket{\psi}\bra{\psi}$. This state can be realized by preparing a system in the state $\ket{\psi}$. The density operators are then used to describe ensembles of systems in such quantum states. We expect that the states of maximal knowledge, i.e. the pure states, should be the basic states one can prepare on a single system. But we can combine the two interpretations of states, i.e. missing knowledge of a single system versus ensemble of many fully known systems: In principle it makes sense that we do not know the exact pure state a system is prepared in. Then it makes sense to describe single systems with mixed states. Indeed, this is the interpretation normally applied when talking about GPTs. And we can still put these only partially known systems into boxes. Following that idea, we will now apply a generalized version of the thought experiment as found in \cite{Petz} for the quantum case. Already Petz suggested that the argument can be used beyond quantum theory, as long as orthogonality makes sense.\\
\\
\begin{center}
	\includegraphics[width=\textwidth]{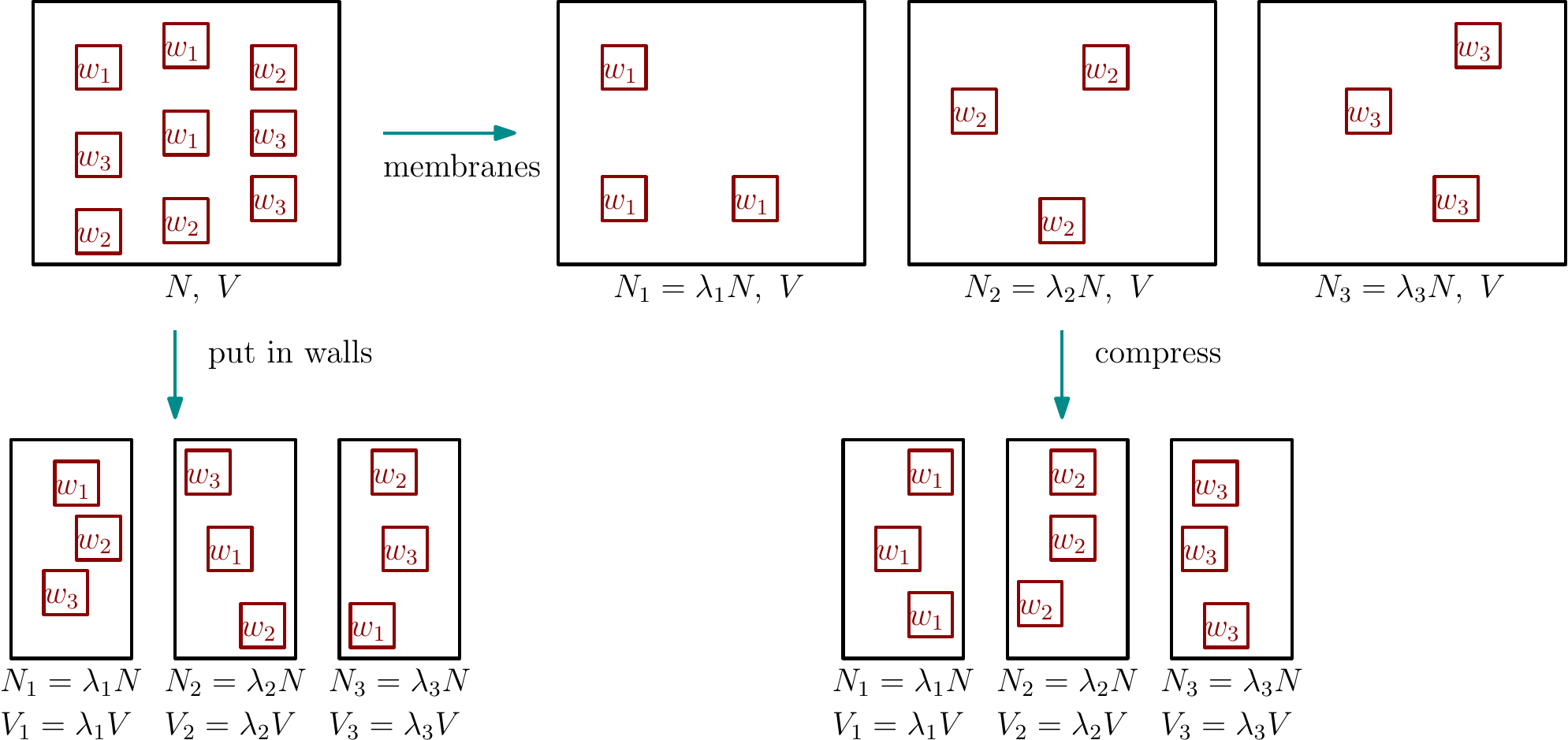}
	\captionof{figure}{\small \textit{Visualization of the argument which relates the thermodynamic entropy of a state to the entropies of a decomposition into perfectly distinguishable mixed states.}}
\end{center}

So let us assume we have $w=\sum_{j=1}^n \lambda_j w_j$ a convex combination ($\sum_j \lambda_j = 1$ and $\lambda_j \in [0,1]$) with $w_j$ perfectly distinguishable ($e_k(w_j) = \delta_{jk}$, $\sum_{k} e_k = u_A$), but not necessarily pure.\\
Analogously to the thought experiment by von Neumann, we consider the following situation: In a tank with volume V, N boxes form an ideal gas. For each $j$, $N_j$ of these boxes contain the GPT-state $w_j$. Thus we obtain an ensemble $w = \sum_j \frac{N_j}{N} w_j$, $\lambda_j = \frac{N_j}{N}$, in a volume V.\\
\\
Just like before we can construct a semi-permeable membrane. However, we need to assume that the membrane can be constructed such that it does not disturb the $w_j$. Like in von Neumann's thought experiment, we use that membrane to separate the constituents: every $w_j$-gas with its $N_j$ particles now is in a separate tank with volume V. We compress each tank to the volume $\lambda_j V$. Hereby we need to perform 
\begin{equation}
  \delta W = - N k_B T \sum_j \lambda_j \ln (\lambda_j).
\end{equation}
As $E \propto T = const.$, the bath gets 
\begin{equation}
\delta Q = \delta W = - N k_B T \sum_j \lambda_j \ln (\lambda_j),
\end{equation}
the gases lose $\delta Q = \delta W$. This means the gases have ``obtained'' \begin{equation}\Delta S = N k_B \sum_j \lambda_j \ln (\lambda_j) \le 0 .\end{equation} All the tanks now have the same density $n_j =\frac{N_j}{\lambda_j V} = \frac{N}{V}$ which coincides with the density at the beginning. \\
\\
Now we go back to the very beginning: One single tank of volume V with a $w$-gas of $N$ particles. We insert walls such that we end up with $n$ tanks of volumes $\lambda_j V$. This step is reversible. We assume that the entropy is extensive and additive, i.e we have $\sum_{j=1}^n S_j(w) = S(w)$, where $S_j(w)$ is the entropy of the $w$-gas in the $j$-th tank. In the situation from the paragraph before, we have $S(w)+\Delta S = \sum_{j=1}^n S_j(w_j)$ (reversible!). Thus:\\
\begin{equation} S(w) = \sum_{j=1}^n S_j(w) = \sum_{j=1}^n S_j(w_j) -  N k_B \sum_{j=1}^n \lambda_j \ln (\lambda_j)\end{equation} or:
\begin{equation} S(w) -\sum_{j=1}^n S_j(w_j) = \sum_{j=1}^n S_j(w)-\sum_{j=1}^n S_j(w_j) = -  N k_B \sum_{j=1}^n \lambda_j \ln (\lambda_j)\end{equation}
As the $j$-th tank in both situations has the same macroscopic conditions (volume, particle number, temperature), the only difference in entropy of the $j$-th tank can be caused by the internal GPT entropy. Thus we find:
\begin{align*}
	\sum_{j=1}^n \big(S_{GPT,j}(w) - S_{GPT,j}(w_j)\big) = -  N k_B \sum_{j=1}^n \lambda_j \ln (\lambda_j)
\end{align*}
We also assume that $S_{GPT}$ is extensive with $\sum_{j=1}^N S_{GPT,j}(w) = S_{GPT}(w)$ and $S_{GPT,j}(w_j) = \frac{N_j}{N}S_{GPT}(w_j)$, i.e. that a homogeneity relation holds. Here, $S_{GPT}(w')$ refers to a tank of volume V filled with N boxes that form a $w'$-ensemble\footnote{This is clear, if we consider the equation $S_{GPT}(w') = -N k_B \sum_j q_j \ln q_j \propto N$ where $q_j$ are the coefficients of a classical decomposition. However, one should be aware that here we are checking for self-consistency of the entropy. The self-consistency might fail, even if the entropy itself is well-defined as a function. Furthermore, we later wish to apply this generalized von Neumann thought experiment to systems that don't always have classical decompositions, like the gbit.}. This makes sense, if we assume that the GPT-entropy is additive:\\
\begin{center}
\includegraphics[width=\textwidth]{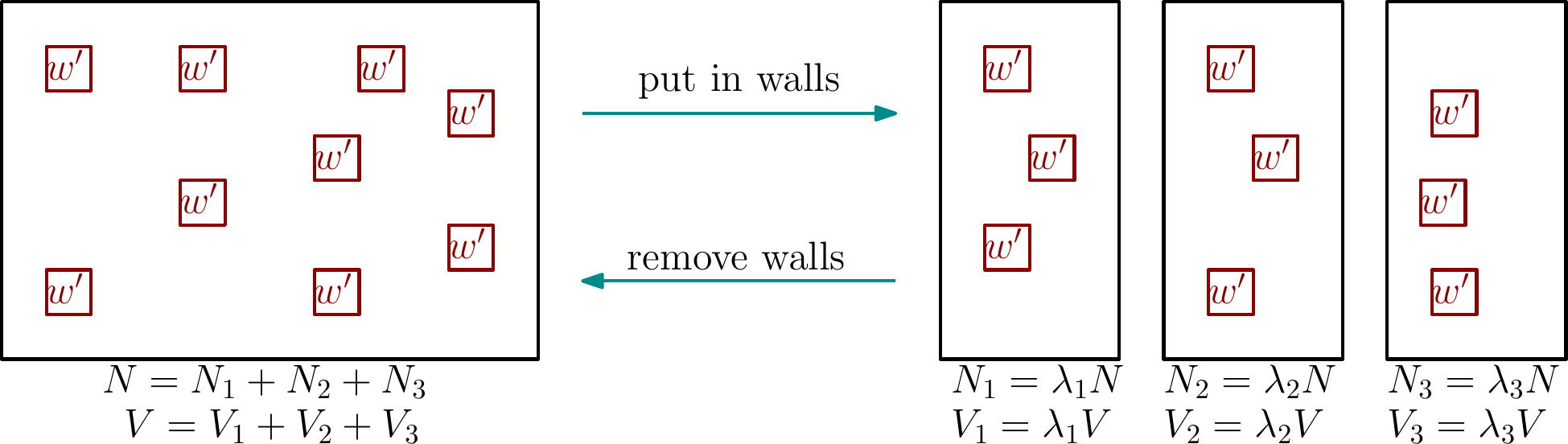}
\captionof{figure}{\small \textit{This figure visualizes the ``wall removing/putting back in''-argument.}}
\label{Fig:WandWegHin}
\end{center}
Tanks with volume $V_j := \lambda_j V$, $\sum_j \lambda_j =1$, filled with a $w'$-gas at density $\frac N V$ can be combined and their walls removed, see Figure \ref{Fig:WandWegHin}. As they had the same content at the same density, this step is reversible by putting the walls back in. Thus it makes sense that 
\begin{equation}
  \sum_{j=1}^n S_{GPT}(w',\lambda_j N, \lambda_j V) = S_{GPT}(w', N, V).
\end{equation} 
Especially in the cases where all $\lambda_j$ are equal ($\lambda_j = \frac 1 n$) we find 
\begin{align*} 
  &n S_{GPT}\left(w',\frac N n, \frac V n\right) = S_{GPT}(w', N, V)  && S_{GPT}\left(w',\frac N n, \frac V n\right) = \frac 1 n S_{GPT}(w', N, V).
\end{align*} 
As by the same \textit{wall removing-putting back in}-argument 
\begin{equation*} 
  m S_{GPT}(w',N, V) = S_{GPT}(w', mN, mV),
\end{equation*} 
$m \in \mathbb N$, we find in total 
\begin{equation} 
  S_{GPT}\left(w', \frac m n N, \frac m n V\right) = \frac m n S_{GPT}(w',N,V),
\end{equation} 
thus\footnote{If you feel uncomfortable about using these extensivity/additivity arguments for the GPT entropy, you can instead use it on the total entropy in $\sum_{j=1}^n S_j(w)-\sum_{j=1}^n S_j(w_j) = -  N k_B \sum_{j=1}^n \lambda_j \ln (\lambda_j)$, obtaining $\sum_{j=1}^n \lambda_j S(w)-\sum_{j=1}^n \lambda_j S(w_j) = -  N k_B \sum_{j=1}^n \lambda_j \ln (\lambda_j)$. Then, we can still identify $S(w)- S(w_j) = S_{GPT}(w)-S_{GPT}(w_j)$.} by continuity 
\begin{equation}
  S_{GPT}(w',pN,pV) = p\cdot S_{GPT}(w',N,V).
\end{equation} 

Therefore we find $\sum_{j=1}^n S_{GPT,j}(w) = S_{GPT}(w)$ and $\sum_{j=1}^n S_{GPT,j}(w_j)=\sum_{j=1}^n \lambda_j S_{GPT}(w_j)$ which leads to our end result:
\begin{align}
	S_{GPT}(w) &= \sum_{j=1}^n \lambda_j \cdot S_{GPT}(w_j) - N k_B \sum_{j=1}^n \lambda_j \ln (\lambda_j)\\
	s_{GPT}(w) &= \sum_{j=1}^n \lambda_j \cdot s_{GPT}(w_j) -  k_B \sum_{j=1}^n \lambda_j \ln (\lambda_j)
\end{align}
where again $s := \frac{S}{N}$. There is no $V$-dependence any more. This result coincides with that of Petz\cite{Petz}. 

It is important that we have not applied any strong symmetry argument, especially we have NOT assumed that the $w_j$ can be reversibly transformed into each other. So far, we have only expressed the entropy of one state by the entropies of other states. To get a final numerical value for the entropy, we need to reduce the entropy of a state to states whose entropies are known (or chosen, by convention). For GPTs satisfying the classical decomposition postulate, one can reduce the entropy to those of pure states, by the strong symmetry one argues that all the pure states have the same entropy which we set to 0. If we had chosen another value, the entropy would not necessarily be extensive, unless we also chose that this additional summand also scales with $N$.\\

\subsection{Classical mixtures and their entropy}
In classical physics one finds for classical labels $j$  (see e.g. \cite{NoltingThermo} Equations (3.57) (3.56) (3.55) (3.54)):
\begin{align}
	 &S(T,V,N,w) =  \sum_{j=1}^{n} S(T,V_j,N_j,w_j) - k_B \sum_{j=1}^{n} N_j \ln\left( \frac{N_j}{N} \right)
\end{align}
If we use $\lambda_j = \frac{N_j}{N} = \frac{V_j}{V}$ and that the entropy is homogeneous, i.e. $S(T,N_j,V_j,w_j) = \lambda_j S(T,N,V,w_j)$ by Equation (3.39) from \cite{NoltingThermo}, then we find:
\begin{align}
	 &S(T,V,N,w) =  \sum_{j=1}^{n} \lambda_j S(T,V,N,w_j) - k_B N \sum_{j=1}^{n} \lambda_j \ln\left( \lambda_j \right)
\end{align}
This result coincides with the result from the generalized von Neumann argument. The last term is the mixture entropy. It also shows, that we really can just add the mixing/internal entropy to the external entropy.\\
\\
For the rest of Chapter \ref{Section:Petz}, $S$ will denote the entropy already divided by the number of systems, i.e. the entropy we called $s_\text{GPT}$ before. Furthermore, we omit $k_B$.

\subsection{Quantum mixtures and their entropy}
Let $\rho = \sum_{j=1}^{n'} p_j \rho_j$ be a convex combination with $\rm{Tr}(\rho_j \cdot \rho_k) \propto \delta_{jk}$. We want to prove 
\begin{equation} S(w) = \sum_{j=1}^{n'} p_j \cdot S(w_j) - \sum_{j=1}^{n'} p_j \ln (p_j) \label{Eq:QTPetz} \end{equation}
for the von Neumann entropy ($k_B =1$, $N$ already divided out).\\
\\
Let $\ket{1_j},...\ket{n_j}$ be an eigenbasis of $\rho_j$ with eigenvalues $\lambda^{(j)}_{k}$. Let $\ket{1_j},...,\ket{f(j)_j}$, where $f(j) \in \{1,...n\}$, be the eigenvectors with $\lambda^{(j)}_{k} \ne 0$. Thus we find ($j\ne k$):
\begin{align*}
	0 &= \text{Tr}(\rho_j \rho_k) = \sum_{a=1}^n \bra{a_j}\rho_j\rho_k\ket{a_j} = \sum_{a=1}^{f(j)} \lambda^{(j)}_a \bra{a_j}\rho_k\ket{a_j}
\end{align*}
Thus $\bra{a_j}\rho_k\ket{a_j}=0$ for $a \le f(j)$ because $\rho_k$ is positive semi-definite. As $1=\text{Tr}(\rho_k)$ we have $1 = \sum_{a=f(j)+1}^n \bra{a_j}\rho_k\ket{a_j}$. For $\rho_k = \sum_b \lambda^{(k)}_b \ket{b_k}\bra{b_k}$ we thus find:
\begin{align*}
1 &= \sum_{a=f(j)+1}^n \sum_{b=1}^n \lambda^{(k)}_b \cdot |\braket{a_j|b_k}|^2 =  \sum_{b=1}^n \lambda^{(k)}_b \sum_{a=f(j)+1}^n |\braket{a_j|b_k}|^2 \\& = \sum_{b=1}^{f(k)} \lambda^{(k)}_b \sum_{a=f(j)+1}^n |\braket{a_j|b_k}|^2
\end{align*}
Therefore $\sum_{a=f(j)+1}^n |\braket{a_j|b_k}|^2 = 1$ for all $b$ with $\lambda^{(k)}_b \ne 0$, which means those $\ket{b_k}$ that actually appear in the eigendecomposition of $\rho_k$. We used 
$\sum_{a=1}^n  |\braket{a_j|b_k}|^2 = 1$ and $\sum_{b=1}^{f(k)} \lambda^{(k)}_b = 1$. Thus  $\sum_{a=f(j)+1}^n |\braket{a_j|b_k}|^2 \le 1$ but we need $=1$ because $\sum_{b=1}^{f(k)}\lambda_{b}^{(k)} q_b < 1$ for $q_b \in [0,1]$ and one $q_b < 1$. \\
We find $|\braket{a_j|b_k}|^2=0$ for all $a \in \{1,..f(j)\}$, $b\in \{1,...,f(k)\}$. In nicer words: The eigenvectors of $\rho_j$ and $\rho_k$, which belong to non-vanishing eigenvalues, are orthogonal to each other. Especially, the eigenvectors of $\rho_j$ and $\rho_k$ that actually appear in the decomposition are orthogonal. We have shown, that $\rho_i$ have support on orthogonal subspaces. By Theorem 11.8(4) from \cite{BestBookEver} we finally find Equation (\ref{Eq:QTPetz}).

\subsection{Maximal consistency for GPTs}

\begin{theorem}
	Consider a GPT which satisfies Postulate 1 (Classical Decomposition) and Postulate 2 (Strong Symmetry) from \cite{Postulates}.\\ Then the thermodynamic entropy is \textbf{fully von Neumann argument}-compatible, i.e. it is fully compatible with decompositions into perfectly distinguishable (mixed) states and their semipermeable membrane\cite{Petz}:\\
Let $w= \sum_j p_j w_j$ be a convex combination of perfectly distinguishable states (they are allowed to be mixed). Then:
\begin{equation}
	S(w) = \sum_j p_j S(w_j) - \sum_j p_j \ln p_j
\end{equation}
\end{theorem}
\vspace{0.5cm}
\begin{proof}
	By Postulate 1, the $w_j$ have classical decompositions $w_j = \sum_k q_k^{(j)} w_k^{(j)}$. Without loss of generality, we assume $q_k^{(j)} > 0$. By Proposition \ref{Prop:SubnormalizedDistinguishable}, leaving out effects from the measurement causes no problems. There are effects with $e_j(w_a) = \delta_{ja}$. By Lemma \ref{Lemma:EffectFace}, $e_j^{-1}(1)$ is a face of $\Omega_A$. By our proof of Lemma \ref{Lemma:EffectFace} or Proposition 2.7 from \cite{Pfister}, we find $w_k^{(j)} \in e_j^{-1}(1) \ \ \forall k$ , i.e. $e_j(w_k^{(j)})=1$. Furthermore, as $e_a(w_j) = 0 \ \ \forall a\ne j$ and $e_a^{-1}(0)$ is a face too, also $e_a(w_k^{(j)})=0$. Thus in total, $e_a(w_k^{(j)})=\delta_{ja}$.\\
So far, by perfect distinguishability we have $\braket{w_j, w_k} \propto \delta_{jk}$, $\braket{w_k^{(j)}, w_a^{(j)}} = \delta_{ka}$ (equality holds because pure states are normalized to 1 by $\braket{\cdot,\cdot}$). As $e_a(w_k^{(j)}) = \delta_{aj}$, for $j\ne k$ also $w_a^{(j)}$ and $w_b^{(k)}$ are perfectly distinguishable by $e_j$ and $e_k$. Thus we also have $\braket{w_a^{(j)},w_b^{(k)}} = 0 \ \ \forall j \ne k$. Thus in total $\braket{w_a^{(j)},w_b^{(k)}} = \delta_{ab} \delta_{jk}$ (equality because pure).\\
Therefore all the $w_a^{(j)}$ form a frame and are perfectly distinguishable.\\
Thus $w = \sum_{j,k} p_j q_k^{(j)} w_k^{(j)}$, $\sum_{j,k} p_j q_k^{(j)} = \sum_j p_j = 1$, $ p_j q_k^{(j)} \in [0,1]$, is a classical decomposition. Especially, by our definition of the thermodynamic entropy: 
\begin{align} S(w) &= - \sum_{j,k} p_j q_k^{(j)} \ln (p_j q_k^{(j)})= -\sum_{j,k} p_j q_k^{(j)} \ln (q_k^{(j)})-\sum_{j,k} p_j q_k^{(j)} \ln (p_j) \nonumber \\
		&= \sum_j p_j S(w_j) -\sum_j p_j \ln(p_j)
\end{align}
\end{proof}

\subsection{Gbits do not satisfy maximal consistency}
The gbit does not allow a classical decomposition in general. But every state can be decomposed into perfectly distinguishable mixed states found in opposing edges of the gbit. As we have already seen in Chapter \ref{Section:GPT}, the corresponding membrane can be constructed such that it does not perturb the states being distinguished by it: the transformations corresponding to effects distinguishing opposing edges (and only those can be distinguished) can be chosen such that the faces collapse to an arbitrary state in that edge. And those states have a classical decomposition. So the idea is as follows, see Figure \ref{Fig:EntropyGbit}: By using the equation of the von Neumann argument for decompositions into mixed states, we reduce the entropy of an arbitrary state to the entropy of states in the boundary:\\
For a convex combination $w= \sum_j p_j w_j'$ with $w'_j$ perfectly distinguishable but not necessarily pure, we have 
\begin{equation*} S(w) = \sum_{j=1}^n p_j \cdot S(w_j') - k_B \sum_{j=1}^n p_j \ln (p_j) \end{equation*}
We choose the $w_j'$ such that they are found in opposing edges of the gbit. Then we use the same equation to reduce the entropies $S(w_j')$ to the entropies of the pure states, i.e. the corners. As a corner can be reversibly transformed into any other corner by a rotation, we assume that all corners have the same entropy ( 0 by convention). This makes sense, as no corner is special.\\
\\
The basic question is: Is this a self-consistent way to define the thermodynamic entropy for a gbit? It is not, as we will show now.

\begin{centering}
\includegraphics[width= 0.6\textwidth]{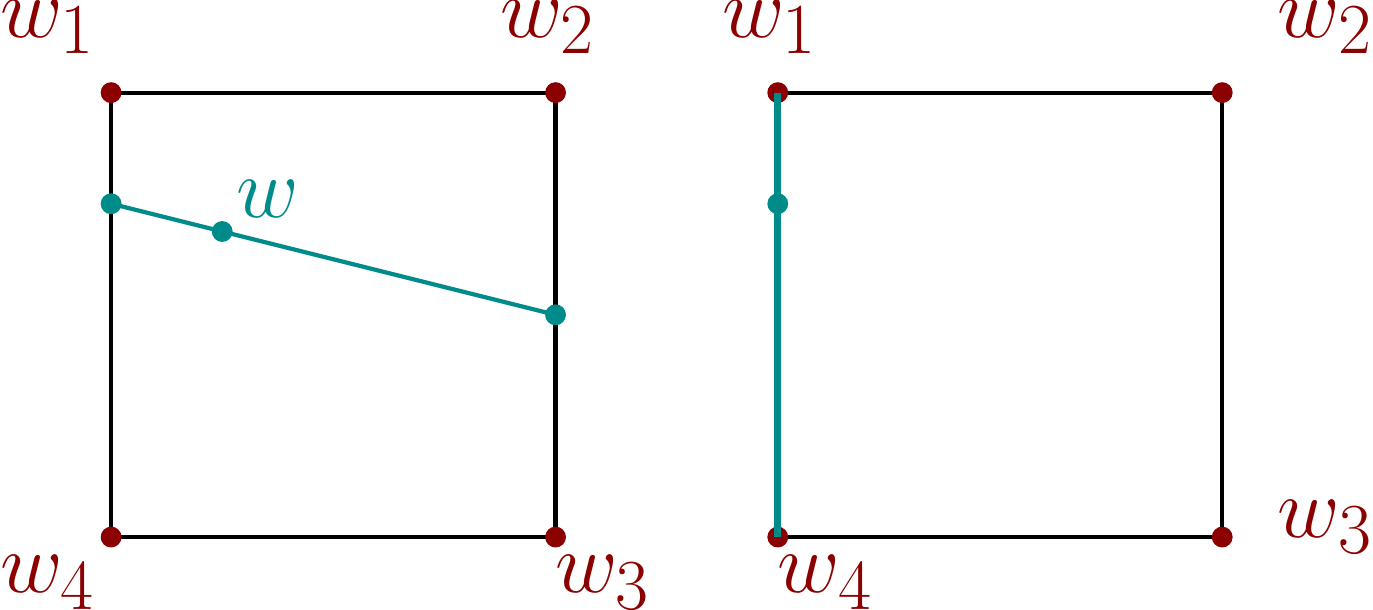}
\captionof{figure}{\small \textit{Example for how to decompose an arbitrary non-pure state into perfectly distinguishable (mixed) states. If more than one such state is necessary, they can always be chosen in opposing edges of the gbit.}}
\label{Fig:EntropyGbit}
\end{centering}

Clockwise starting in the upper-left corner, we denote the corners by $w_1,w_2,w_3,w_4$.\\
We consider the ``maximally mixed state'' 
\begin{align}
  w&=\frac 1 2 \left(\frac 1 2 w_1+\frac 1 2 w_2\right)+\frac 1 2 \left(\frac 1 2 w_3+\frac 1 2 w_4\right)\\
  &=\frac 1 4 w_1+\frac 1 4 w_2+\frac 1 4 w_3+\frac 1 4 w_4 = \frac 1 2 w_1 + \frac 1 2 w_3 = \frac 1 2 w_2 + \frac 1 2 w_4
\end{align}
found in the center of the square. There are many ways how this state can be decomposed into states found in the boundary, it also does have classical decompositions. For arbitrary $a \in [0,1]$ we define $v_a:= a \cdot w_1 + (1-a)\cdot w_2$. Then, we also define $v_a':= a \cdot w_3 + (1-a)\cdot w_4$. Thus we find:
\begin{align*}
 \frac 1 2 v_a + \frac 1 2 v_a' &= a \left(\frac 1 2 w_1 + \frac 1 2 w_3\right) + (1-a) \left(\frac 1 2 w_2 +\frac 1 2 w_4\right) = a w +(1-a) w = w
\end{align*}

\begin{center}
\includegraphics[width=0.3\textwidth]{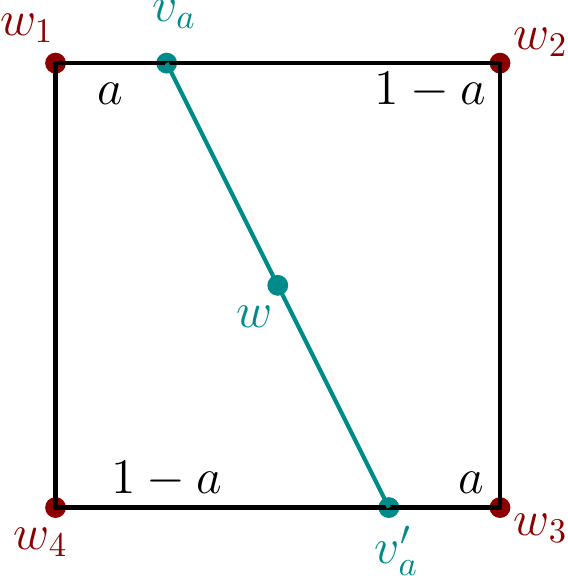}
\captionof{figure}{\small \textit{This figure visualizes how we choose $v_a$ and $v_a'$ as decompositions of the ``maximally mixed state''.}}
\end{center}

$v_a$ and $v_a'$ are perfectly distinguishable, because they are found in opposing edges of the square. So our entropy should be:
\begin{equation*} S(w) = \frac 1 2 \cdot S(v_a)+\frac 1 2 \cdot S(v_a') - k_B 2 \cdot\Big( \frac 1 2 \ln \left(\frac 1 2\right)\Big) =\frac 1 2 \cdot( S(v_a)+ S(v_a')) +k_B \ln 2  \end{equation*}
While the equation of the entropy did not require that $v_a$ and $v_a'$ can be reversibly transformed into each other, here it is possible by a 180 degree rotation: $T(w_1) = w_3, T(w_2) = w_4, T(w_3) = w_1, T(w_4) = w_2$, i.e. $T(v_a) = v_a', T(v_a') = v_a $. This has the consequence that both states have the same entropy, which we will now determine:
\begin{align}
S(v_a) &= a \cdot S(w_1)+(1-a)\cdot S(w_2) - k_B a \ln (a)-k_B (1-a) \ln (1-a) \nonumber \\
	&= - k_B a \ln (a)-k_B (1-a) \ln (1-a)
\end{align}
We used our convention $S(w_j) = 0$ and that neighbouring corners can be perfectly distinguished, because there are two opposing edges which contain these corners. In the same way: $S(v_a') = - k_B a \ln (a)-k_B (1-a) \ln (1-a)$.
So in total, we find for entropy of the center:
\begin{equation}
	S(w) = - k_B a \ln (a)-k_B (1-a) \ln (1-a) + k_B \ln 2
\end{equation}
This result is not self-consistent, as it still depends on $a$, i.e. on the chosen decomposition. This shows that it is not possible to define an entropy on the states of the gbit which is compatible with von Neumann's thermodynamic argument. This entropy varies from $\ln 2$ for $a= 0,1$ to $2\ln 2$ for $a=\frac 1 2$.
\\

\section{Projective measurements for GPTs}
\label{Section:ProjectiveMeasurements}
So far we have considered membranes as in von Neumann's original argument for the derivation of the von Neumann entropy. While there are measurements which can perfectly distinguish the states considered in the argument, it is not clear what happens to these states. Now we want to make the measurement more concrete, by modelling it as an operation. Like in quantum theory, there exist projective measurements that have the desired properties. At first we will consider projectors onto ``minimal'' faces, each of them generated by only one pure state. These projectors can be used to perfectly distinguish the states of (maximal) frames. Afterwards we will consider more general projective measurements as needed to distinguish perfectly distinguishable states which in general are not pure. 

\subsection{Projective measurements for frames}
Here, we consider projective measurements which perfectly distinguish the elements of (maximal) frames. In quantum theory, these measurements correspond to the measurement of non-degenerate observables as we will explain in Chapter \ref{Section:2ndLaw}.\\
Remember from Section \ref{Section:Postulates} that every pure state generates a face while every face is generated by a frame of one or more pure states. For every face $F$ there is a positive, symmetric projector $P_F$ onto the linear span of the face. $u_F = u_A \circ P_F$ is an effect and is called the projective order unit of the face. If $v_1,...,v_k$ is a frame that generates $F$, then $u_F = \braket{\sum_{j=1}^k v_j,\cdot}$ or $u_F = \sum_{j=1}^k v_j$ using self-duality. 

\begin{lem}
	Suppose that strong symmetry and classical decomposability are satisfied. The face $F:= \mathbb R_{\ge 0} \cdot \{w\}$ for $w$ pure has $u_F = w$ because $F$ is generated by $w$. Furthermore, for a frame $w_1,...,w_n$ and $j \ne k$, we find $u_A \circ P_j(w_k) = u_j(w_k) = \braket{w_j,w_k} = 0$, where $u_j$ and $P_j$ are the projective unit and the orthogonal projector of the face $\mathbb R_{\ge 0} \cdot \{w_j\}$. As $P_j$ is positive and $0 \in A_+$ is the only state normalised to 0, we find $P_j(w_k) = 0$. \label{Lemma:FrameProjectors}
\end{lem}

\begin{theorem}
	Assume strong symmetry and classical decomposability. Let $w_1,...,w_n$ be a maximal frame\footnote{Remember that every frame can be extended to a maximal frame.}. Let $u_j$ be the projective unit, $P_j$ the orthogonal projector corresponding to the face $F_j := \mathbb R_{\ge 0} \cdot \{w_j\}$. Then $\{P_1,...,P_n\}$ form a valid operation. They induce a measurement which perfectly distinguishes the elements of the frame: $u_j(w_k) = u_A \circ P_j(w_k) = \delta_{jk}$. Furthermore, $P_j(w_k) = \delta_{jk} w_k$. Thus the measurement does not disturb the frame.
\end{theorem}

\begin{proof}
	In Lemma \ref{Lemma:FrameProjectors} we have seen $P_j(w_k) = \delta_{jk} w_k$ for $j \ne k$. As projectors are surjective onto the linear spans of the corresponding faces (as implied by the word ``onto''), there is a vector $w \in A$ with $w_j = P_j w$. As $P_j$ is a projector and thus $P_j P_j = P_j$, we find $P_j w_j = P_j P_j w = P_j w = w_j$ , i.e. $P_j(w_k) = \delta_{jk} w_k$ for $j = k$.\\
\\
As $0 \le u_j = u_A \circ P_j \le u_A$, the projectors induce valid effects and especially are normalization-non-increasing. Furthermore, we already know that the projectors are positive. Thus the projectors are valid transformations. As maximal frames add up to the order unit, we find $\sum_j u_A\circ P_j = \sum_j u_j = \sum_j w_j = u_A$, i.e. a full measurement is obtained. So in total we have a full operation. It perfectly distinguishes the frame, because $u_j(w_k) = u_A \circ P_j(w_k) = \delta_{jk}$. 
\end{proof}

\begin{cor}
	The operation constructed above satisfies all properties needed for the membrane in the von Neumann thought experiment.
\end{cor}

\begin{note}
	While we have shown, that the operation is mathematically well-defined, it is not clear whether the projectors actually are physically allowed. The postulates from \cite{Postulates} do not consider non-reversible transformations and especially do not assume that the projectors which model the M-slit experiments in Postulate 3 are actually physically allowed. The von Neumann thought experiment gives a good reason to assume that projectors are physically allowed.\\
One might motivate this assumption similarly to Pfister's ``one simple postulate''\cite{Pfister}: As each of the frame elements will have a clear, predetermined outcome, it should be possible to do that measurement without disturbing. Another motivation is, that a frame generates a classical subspace (see motivation for Postulate 1). If we only consider the $w_j$ and the $u_j$, then we only work in that classical subspace, i.e. we have a classical behaviour. In classical physics, in principle measurements can be done without disturbing the system. Furthermore, as M-slit experiments are built by a slit-plane followed by a detector-plane, it is natural to assume that the system survives the slit-plane. 
\end{note}

\subsection{General projective measurements}
While the measurements from the previous chapter fit to the original von Neumann argument, the generalized version found in Petz \cite{Petz} needs similar statements for more general projective measurements. This means now we want to find measurements that perfectly distinguish some perfectly distinguishable states, which might be mixed, but without disturbing these states.

\begin{lem}
	Assume strong symmetry and classical decomposability. Let $w_1,...,w_n \in \Omega_A$ be perfectly distinguishable, but not necessarily pure. Let $F_j \subset \Omega_A$ be the minimal faces that contain $w_j$. Then $F_j \subset e_j^{-1}(1)$, $F_j \subset e_k^{-1}(0)$ for $j \ne k$ where $e_j$ are effects that perfectly distinguish the $w_k$, i.e. $e_j(w_k) = \delta_{jk}$. Furthermore, the faces $F_j$ are orthogonal to each other. The same is true for the corresponding faces of $A_+$. \label{Lemma:FacesOrthogonalDegenerate}
\end{lem}

\begin{proof}
	 By the definition of perfectly distinguishable, there exist effects with $e_j(w_k) = \delta_{jk}$. For $j\ne k$, $e_k^{-1}(0)$ and $e_j^{-1}(1)$ are faces (see Lemma \ref{Lemma:EffectFace}) which contain $w_j$. As $F_j$ is the minimal face which contains $w_j$, we find  $F_j \subset e_j^{-1}(1)$, $F_j \subset e_k^{-1}(0)$ for $j \ne k$. \\
	 As the faces $F_j$ are contained in $ e_j^{-1}(1), e_k^{-1}(0)$ for $j \ne k$, these effects also perfectly distinguish the faces $F_j$, i.e. they are orthogonal by Theorem \ref{Theorem:Selfdual}.\\ Thus also $\braket{\mathbb R_{\ge 0} \cdot F_j,\mathbb R_{\ge 0} \cdot F_k} = \{0\}$ for $j \ne k$.
\end{proof}

\begin{lem}
	Suppose strong symmetry and classical decomposability are satisfied. Let $w_1,...,w_n \in \Omega_A$ be perfectly distinguishable, but not necessarily pure. Let $F_j \subset \Omega_A$ be the minimal faces that contain $w_j$. Let $P_j$ be the symmetric projection onto the linear span of $\mathbb R_{\ge 0} \cdot F_j$. Then $P_j w_k = \delta_{jk} w_k$. 
\end{lem}

\begin{proof}
	We apply the same tricks as before:\\
	As $w_j \in \rm{im}(P_j)$, $\exists w \in A: P_j w = w_j$. Thus $P_j w_j = P_j P_j w = P_j w = w_j$.\\
	Let $w_1^{(j)},...,w_{f(j)}^{(j)}$ be frames that generate $\mathbb R_{\ge 0} \cdot F_j$ and $u_j$ the corresponding projective unit. Then $u_j = \sum_{k=1}^{f(j)} w_{k}^{(j)}$. Especially $u_A \circ P_j = u_j = \sum_{k=1}^{f(j)} w_{k}^{(j)}$. By Lemma \ref{Lemma:FacesOrthogonalDegenerate} the faces are orthogonal and thus $u_A \circ P_j(w_k) = 0$ for $j \ne k$, i.e. $P_j(w_k) = 0$ for $j \ne k$ by positivity.
\end{proof}

\begin{theorem}
	Assume strong symmetry and classical decomposability. Let $w_1,...,w_n \in \Omega_A$ be perfectly distinguishable, but not necessarily pure or maximal. Then there exists a projective measurement that perfectly distinguishes the $w_k$ in the following sense:\\
$\exists P_1,...,P_{n+f(n+1)}$ orthogonal positive projectors that form an operation with\\ $\sum_{k=1}^{n+f(n+1)} u_A \circ P_k = u_A$ and $u_A \circ P_j (w_k) = \delta_{jk}$ for $j \in \{1,...,n+f(n+1)\}$ and $k \in \{1,...,n\}$.
\end{theorem}

\begin{proof}
	We consider the minimal faces $F_j \subset \Omega_A$ that contain $w_j$.\\ 
	Let $w_1^{(j)},...,w_{f(j)}^{(j)}$ be frames that generate $\mathbb R_{\ge 0} \cdot F_j$ and $u_j$ the corresponding projective unit, i.e. $u_j = \sum_{k=1}^{f(j)} w_{k}^{(j)}$. As the $w_k^{(j)}$ are pairwise orthogonal pure states, by Proposition \ref{Prop:FramesAddUp} (or Proposition 6 from \cite{Postulates}) they are a frame, which can be extended to a maximal frame. We will call the new frame elements $w_1^{(n+1)},...,w_{f(n+1)}^{(n+1)}$. Also consider the positive orthogonal projectors $P_{n+k}$ onto the faces $\mathbb R_{\ge 0} \cdot {w_k}^{(n+1)}$.\\ 
Then $0 \le u_j = u_A \circ P_j \le u_A$ for $j=1,2,...,n+f(n+1)$. As the projectors are also positive, they are valid transformations.\\
Furthermore:\\ $\sum_a u_A \circ P_a = \sum_{j=1}^{n} u_A \circ P_j + \sum_{k=1}^{f(n+1)} u_A\circ P_{n+k} = \sum_{j=1}^n \sum_{k=1}^{f(j)} w_k^{(j)}+\sum_{k=1}^{f(n+1)} w_k^{(n+1)} = u_A$ because maximal frames add up to the order unit. Thus we obtain a full projective operation.\\
For $P_1,...,P_n$ we already know $P_j(w_k) = \delta_{jk} w_k$. As $u_A \circ P_{n+k} w_j = \braket{w_k^{(n+1)}, w_j} = 0$ (alternatively, use that $u_A\circ P_j(w_j) = 1$ implies $u_A\circ P_{n+k}(w_j) = 0$ by having a normalized operation),  we find that only the projector $P_j$ has a non-zero probability to be performed on $w_j$, and it does not disturb $w_j$.
\end{proof}

\begin{note}
	One can drop the last $f(n+1)$ effects if one allows the measurement to be normalized to less than $u_A$. Then one still finds $u_A \circ P_j (w_k) = \delta_{jk}$.\\
\\
The same considerations from the previous chapter also apply here. I.e. it is not clear whether the projections are physically allowed, it has to be assumed as a well-motivated postulate. Furthermore, the projective measurement created in this chapter is perfectly suited for the generalized version of von Neumann's thought experiment found in \cite{Petz}.
\end{note}

\section{Second law}
\label{Section:2ndLaw}
The most famous property of the entropy is that it fulfils the second law of thermodynamics. At first we will discuss whether the second law holds in time evolution. Afterwards we will explain why even in quantum theory some quantum operations decrease entropy. This implies that an increase in entropy can only be shown for some subsets of the set of all quantum operations. We will therefore analyse projective measurements, as they are a central part of quantum theory. At last, we show that the entropy also never decreases during mixing processes.\\
\\
From now on, $S$ will always denote the entropy divided by the number of systems, i.e. the entropy we called $s_\text{GPT}$ before. Furthermore we omit $k_B$. The only exception will be Chapter \ref{Section:Mixing}, where the situation from the thermodynamic thought experiment will be used one more time.

\subsection{Time evolution}
Postulate 4 from \cite{Postulates} implies that time evolution is reversible and thus does not change the entropy. In particular, the 2nd law of thermodynamics is valid in dynamics. However, it is clear that Postulate 4 is much stronger than what we actually need in order to ensure that time evolution does not violate the second law. As long as time evolution is described by a reversible transformation, the entropy is conserved and thus does not decrease.\\
However, Postulates 1+2 do not say anything about time evolution, thus we will need an extra postulate specifying time evolution (e.g. Postulate 4). As long as we only consider Postulates 1+2, we cannot check if time evolution respects the second law because time evolution itself remains undefined. For now, we neglect the question if time evolution respects the second law; we consider the second law as a consistency requirement that any definition of time evolution should satisfy. 

\subsection{Issues concerning the second law in measurements and transformations}
In this section we will explain why there are some processes that are able to decrease the entropy, making it necessary to focus on some special transformations and operations and proving the entropy increase for them.\\
Consider an operation $O=\{T_1,...\}$ with transformations $T_j$ and $\sum_j u_A \circ T_j = 1$. Consider the action of the operation on an ensemble described by the state $w$. With probability $u_A\circ T_j (w)$ the state after the operation will be $\frac{T_j(w)}{u_A\circ T_j (w)}$ . This induces a new ensemble: 
\begin{equation}
w' = \sum_{ \{k\ | \ u_A\circ T_k (w) \ne 0\} } u_A\circ T_k (w) \cdot \frac{T_k(w)}{u_A\circ T_k (w)} = \sum_k T_k(w)
\end{equation}
In general, we cannot assume that such operations only increase entropy. If one considers a system coupled to an environment and performs a transformation on the composite system, there is no reason left why on the small system alone entropy should not be allowed to decrease. So if we want our definition of transformations to be as general as possible, especially to cover transformations induced by larger systems, we have to accept that some of them decrease entropy.\\
This is already true in quantum theory, as the following example shows: Consider the following SWAP-operation, the mathematical details are explained in Appendix \ref{Section:SWAP}: Whenever an electron approaches a black-box device, that device absorbs that electron and instead emits a new electron in a known pure state. As every incoming electron is replaced by a new electron, this transformation is already properly normalized and its physical implementation is also clear. But us this transformation is able to convert a mixed state into a pure state, in general it will decrease entropy. 
\\
Also the projective measurements play an important role for the second law in quantum theory. In Exercise 11.15 in \cite{BestBookEver}, one has to show that a measurement described by $M_1 := \ket{0}\bra{0}$, $M_2 := \ket{0}\bra{1}$ and ensemble state after the measurement $M_1 \rho M_1^\dagger+M_2 \rho M_2^\dagger$ decreases entropy. This is quite clear, because the state after this measurement is $\ket{0}\bra{0}$, i.e. a pure state.
\\
Because of all these conceptual problems, we decide to prove the second law only for projective measurements in analogy to quantum theory. But different to quantum theory, we will not postulate that these are the fundamental measurements.

\subsection{GPT observables}
The purpose of this section is to motivate why we will consider certain projective measurements for checking the second law. The basic idea is to consider projective measurements that correspond to measuring observables.\\
\\
We consider projective measurements in analogy to quantum theory. In quantum theory, observables are of the form $\mathcal A =\sum_{a} a P_a$ with eigenvalues $a$ and projectors onto eigenspaces $P_a =\sum_j\ket{j;a}\bra{j;a}$. $\ket{j;a}\bra{j;a}$ are pairwise orthogonal pure states. For our GPTs, we generalize observables as $\mathcal A =\sum_a a \sum_j w_{j;a}$ where $ w_{j;a}$ are pairwise orthogonal pure states and form a maximal frame. For fixed $a$, the face $F_a$ generated by all the $ w_{j;a}$ replaces the eigenspace, i.e. observables now have eigenfaces. For the corresponding symmetric projectors $P_a$ onto the linear spans of the faces $F_a$, we find $u_a := u_A \circ P_a = \sum_j w_{j;a}$ and thus $u_A \circ \sum_a P_a = \sum_a u_{a} = \sum_a \sum_j w_{j;a} =u_A$.  As these projectors are positive, they are furthermore normalization-non-increasing because of $u_A \circ \sum_a P_a = u_A$. Thus the eigenprojectors form a valid operation. By Lemma \ref{Lemma:OrthogonalProjectors}, these projectors are mutually orthogonal. Thus $\braket{P_j v,P_k w} = \braket{v, P_j P_k w} = 0$, which implies that also the faces they project onto are orthogonal. Thus the eigenfaces are orthogonal. The induced effects, i.e. the projective units, are $u_a = \sum_{j}w_{j;a}$. Such effects are also called \textit{sharp effects} \cite{Postulates}. \\

Vice versa, consider an operation given by sharp effects $u_a = \sum_j w_{j;a}$ for pure states $w_{j;a}$. As $1\ge u_a(w_{k;a}) = \sum_j \braket{w_{j;a},w_{k;a}} = 1+ \sum_{j\ne k} \braket{w_{j;a},w_{k;a}}$, for same $a$ the $w_{j;a}$ are pairwise orthogonal. As $\sum_a u_a = u_A$ and $u_a(w_{j;a}) =1$, we find $u_b(w_{j;a}) = 0$ for $b\ne a$ and thus $\braket{w_{k;b},w_{j;a}}=\delta_{ab}\delta_{jk}$. Thus the $w_{j;a}$ can be used to define an observable $M = \sum_a a \sum_j w_{j;a}$ and thus give rise to a projective measurement. In that sense, there is a correspondence between sharp measurements and projective measurements of observables.  \\
\\
The important difference to quantum theory is that it is never specified, what the pure states actually are. In QT, we know that the pure states are induced by a complex Hilbert space; we do not assume this for our GPTs, thus our treatment is more general.\\
\\
We do not assume as a postulate that observables enter the GPT in this way. However, the natural generalization from quantum theory motivates to analyze the consequences of such an assumption in greater detail.\\
Non-degenerate observables $\mathcal A = \sum_j a_j w_j$ with $a_j \in  \mathbb R$, $a_j \ne a_k$ for $j \ne k$ and $\{w_j\}$ a maximal frame correspond to measurements with projective units $u_j= w_j$, i.e. faces generated by single pure states. We will call such projective measurements \textbf{non-degenerate}. Projective measurements corresponding to degenerate observables will be called \textbf{degenerate}.\\
\\
We need to show that our notion of observables is well-defined, i.e. that the eigenvalues and eigenfaces do not depend on the choice of decomposition:

\begin{theorem}
	Let $\mathcal A =\sum_{x=1}^{n_a} a_x \sum_j w_{j,a,x}$ be arbitrary with $a_x \in \mathbb R$ pairwise unequal and $w_{j,a,x}$ a maximal frame. Assume $\mathcal A =\sum_{x=1}^{n_b} b_x \sum_j w_{j,b,x}$ is a similar decomposition. Then $n_a = n_b$ and, except for permutation, $a_x = b_x$ and $\sum_j w_{j,a,x} = \sum_j w_{j,b,x}$. 
\end{theorem}
\begin{proof}
	Wlog assume that the $a_x$ and $b_x$ are ordered by size ($a_1 < a_2 <...$) and are numbered with $x=1,2,...$. Consider the smallest eigenvalues and assume $a_1 \ne b_1$, wlog $a_1 < b_1$. Then we use that maximal frames add up to the order unit:
\begin{align}
	a_1 &= \braket{w_{k,a,1}, \mathcal A} = \braket{w_{k,a,1},\sum_x b_x \sum_j w_{j,b,x}} \\ &= \sum_x b_x \braket{w_{k,a,1}, \sum_j w_{j,b,x}} > a_1 \sum_x \braket{w_{k,a,1}, \sum_j w_{j,b,x}} = a_1 u_A(w_{k,a,1}) = a_1
\end{align}
The $>$ is not a $\ge$, because all the $b_x > a_1$ and at least one of the $\braket{w_{k,a,1}, \sum_j w_{j,b,x}} > 0$ because of $u_A(w_{k,a,1})=1$ and $u_A = \sum_{x,j} w_{j,b,x}$. But now we reached a contradiction. Thus our assumption was false. Thus $a_1 = b_1$. \\
\\
The $w_{j,a,1}$ generate a face $F_1$ and the $w_{j,b,1}$ generate a face\footnote{The $F_1'$ is used here simply as another face, and is not meant to imply that $F_1'$ is complementary or orthogonal to $F_1$.} $F_1'$. The projective order units are given by $u_1=\sum_j w_{j,a,1}$ and $u_1'=\sum_j w_{j,b,1}$. 

\begin{align}
	a_1 &= \braket{w_{k,a,1},\mathcal A} = \braket{w_{k,a,1},\sum_x b_x \sum_j w_{j,b,x}} = \sum_x b_x \braket{w_{k,a,1}, \sum_j w_{j,b,x}} \\ & = a_1 \braket{w_{k,a,1}, \sum_j w_{j,b,1}} + \sum_{x\ne 1} b_x \braket{w_{k,a,1}, \sum_j w_{j,b,x}}\\
		& \ge a_1 \braket{w_{k,a,1}, \sum_j w_{j,b,1}} + a_1\sum_{x\ne 1} \braket{w_{k,a,1}, \sum_j w_{j,b,x}} \\
		& = a_1 \braket{w_{k,a,1}, \sum_x \sum_j w_{j,b,x}} =a_1 u_A(w_{k,a,1})= a_1
\end{align}
In the above equation, the $>$ holds if $\braket{w_{k,a,1}, \sum_j w_{j,b,x}} \ne 0$ for a $x\ge 2$. But then we had a contradiction. Therefore $\braket{w_{k,a,1}, \sum_j w_{j,b,x}} \ne 0$ for a $x\ge 2$. Thus because of normalization and maximal frames, $\braket{w_{k,a,1}, \sum_j w_{j,b,1}}=1$, i.e. $u_1'(w_{k,a,1}) = 1$. In the same way show $u_1(w_{k,b,1}) = 1$.\\
By Proposition 5.29 from \cite{Ududec}, $\Omega_A \cap F_1 = \{w \in \Omega_A | u_{F_1}(w) = 1\}$ (Besides, this shows that the symmetric projections onto linear span of faces are neutral). Thus the frame $w_{j,a,1}$ is found in $F_1'$ and the frame $w_{j,b,1}$ is found in $F_1$. As generating frames have the same size, we find $|F_1| \ge |F_1'|$ and $|F_1'| \ge |F_1|$. Thus $|F_1|=|F_1'|$. As all frames within a face with generating size generate the face, $F_1 = F_1'$. Thus also $u_1 = u_1'$, i.e. $\sum_j w_{j,a,1} = \sum_j w_{j,b,1}$. \\
\\
We want to obtain an inductive proof.\\
The easiest way might be like this: As $\sum_j w_{j,a,1} = \sum_j w_{j,b,1}$, we modify the operator:
\begin{align}
	\mathcal A' &:= \mathcal A + L \sum_j w_{j,a,1} = \mathcal A + L\sum_j w_{j,b,1}\\
	&= \sum_x (a_x+\delta_{x,1}\cdot L)\sum_j w_{j,a,x} = \sum_x (b_x+\delta_{x,1}\cdot L)\sum_j w_{j,a,x}
\end{align}
Here, $L$ is a very large number, such that now $a_2$ and $b_2$ are the smallest eigenvalues of $\mathcal A'$. We rename the index:\\
$a'_1 := a_2$, $a'_2 := a_3$,..., and $a'_{n_a} := a_1+L$ is the last because it is the largest eigenvalue (for $b'_x$ similarly). The $a'_x$ are still ordered by size. Now we repeat exactly the same argument as before to find $a_2 = a'_1 =b'_1 = b_2$ and $\sum_j w_{j,a,2} = \sum_j w_{j,b,2}$. We repeat this procedure until we are done. It is important to note, that as the $ w_{j,a,x}$ and the $w_{j,b,x}$ form maximal frames and as we prove for each index $x$, that the corresponding frames have the same size, necessarily $n_a = n_b$. I.e. ``we will not run out of b's while we still have a's left or vice versa''.\\ \\
\end{proof}

\begin{cor}
	The eigenvalues and eigenfaces of observables are well-defined: The eigenvalues are uniquely determined and $u_x :=\sum_j w_{j,a,x}=\sum_j w_{j,b,x}$ generate the same eigenface $F_x$.\\
As generating frames of faces have a unique size, the statement also shows that the probability distribution of classical decompositions of states is uniquely determined.\\
Also for $w = \sum_x p_x \sum_j w_{jx} = \sum_x p_x \sum_j w_{jx}'$ classical decompositions of a states into maximal frames, we find $\sum_j w_{jx} = \sum_j w_{jx}'$. Thus $\log w := \sum_x \log(p_x) \sum_j w_{jx} = \sum_x \log(p_x) \sum_j w_{jx}'$ is independent of the choice of classical decomposition.\\
Likewise for any function $f: \mathbb R \to \mathbb R$ and $\mathcal A=\sum_x a_x \sum_j w_{j,a,x}$ as before, $f(\mathcal A) := \sum_x f(a_x) \sum_j w_{j,a,x}$ is well-defined and independent of the decomposition into maximal frames.
\end{cor}

It is important to note, that the close similarity to quantum observables is caused by Postulates 1 and 2, especially the fact that eigenvalues and eigenfaces are well-defined.\\
By Lemma 5.46 from \cite{Ududec} every element $w$ of $A$ does have a generalized classical decomposition of the form $w=\sum_j p_j w_j$ with $\{w_j\}$ a frame and $p_j \in \mathbb R$. Thus every element of $A$ can be interpreted as an observable. This fact is in direct analogy to quantum theory where every hermitian operator is interpreted as an observable.\\
\\
Sometimes \cite{Masanes}, observables are introduced to GPTs in a different but equivalent way:\\
When measuring an observable on ensembles, one obtains an average value which agrees with the expectation value in the thermodynamic limit. Thus it makes sense to introduce observables as functions $A_+ \rightarrow \mathbb R$. With the same argument as for effects, these functions should be convex-linear and can be extended to linear functions $A \rightarrow  \mathbb R$. Therefore observables are sometimes defined as elements of $A^*$.\\
By self-duality, for every observable $\mathcal A_* \in A^*$ there exists a $\mathcal A \in A$ with $\mathcal A_* = \braket{\mathcal A, \cdot}$. Vice versa, for every $\mathcal A \in A$ we find that $\mathcal A_* := \braket{\mathcal A,\cdot}$ is an element of $A^*$. Thus the elements of $A^*$ and $A$ are in bijective correspondence, and the definitions of observables as elements of $A$ or as elements of $A^*$ are equivalent.\\
\\
When postulating the standard axioms of quantum theory, there is always a mysterious quantum-classical transition in the measurement in the following sense: The measurement device, possibly quantum itself, measures a quantum property and returns a classical result (a digital number on a screen). In our framework considering the motivation for Postulate 1, we can justify this as follows: Every observable is part of a classical subspace. The measurement of the observable perfectly distinguishes the eigenspaces of the observable and is a consistent generalisation of the classical measurement to other classical subspaces. Thus the quantum optics explanation (see e.g. \cite{OpenQuantumSystems} or \cite{Decoherence}) by decoherence of an open quantum system (the measurement device being the environment) reduces the original state to a state of the classical subspace of the observable leading to a completely classical behaviour of the observable and the state, when only actions on this classical subspace are performed.

\subsection{Second law for non-degenerate projective measurements}
At first, we show the second law for non-degenerate projective measurements:\\
	Let $w \in \Omega$ be an arbitrary state, $w= \sum_j p_j w_j$ a classical decomposition into a maximal frame. Let $P_j$ be positive symmetric projectors with effects $u_j = u_A \circ P_j$ which form a properly normalized measurement, i.e. the $P_j$ form an operation. Furthermore, assume $P_j$ are the projectors onto the span of $w_j'$, where the $w_j'$ form a maximal frame (this means the projectors correspond to rank-1 projectors from quantum theory, i.e. no degeneracy of the measured observable). The measurement is conducted on all systems of the ensemble, i.e. the ensemble state after the measurement is given by 
	\begin{equation}
		w' = \sum_{\{j \ | \ u_A\circ P_j(w)\ne 0\} } u_A\circ P_j(w) \cdot \frac{P_j(w)}{u_A\circ P_j(w)} = \sum_j P_j(w)
	\end{equation}
\begin{theorem}
	Suppose Classical Decomposability and Strong Symmetry are satisfied. Then non-degenerate projective measurements as defined above never decrease the entropy of an ensemble:
\begin{equation}	
	S(w') \ge S(w)
\end{equation}	
\end{theorem}

\begin{proof}
Like described above, let $w=\sum_j p_j w_j$ be a classical decomposition into a maximal frame. Furthermore $P_j$ is positive and projects onto the linear span of $w_j'$, i.e. $\rm{im}^+(P_j) = \mathbb R_{\ge 0} \cdot \{w_j'\}$ (where $\rm{im}^+(P_j) := \rm{im}(P_j)\cap A_+$). Via the projective units we see that $\braket{w_j',\cdot } = u_A \circ P_j$. \\
 Thus $w' = \sum_j P_j(w) =: \sum_j q_j w_j'$ where $q_j= u_A \circ P_j(w)$:\\ 
 If $q_j=u_A \circ P_j(w)=0$, then $P_j w = 0$. Otherwise, $\frac{P_j(w)}{u_A\circ P_j(w)} = w_j'$ by proper normalization and therefore $P_j(w) = u_A \circ P_j(w) \cdot w_j' $. Furthermore:\\
$q_j =  u_A\circ P_j(w) = \sum_k u_A \circ P_j(w_k) p_k =: \sum_{k} M_{jk} p_k$.
As maximal frames are considered, $M_{jk} = u_A \circ P_j(w_k)$ is a square matrix with non-negative entries. Also:\\
$\sum_j M_{jk} = \sum_j \braket{w_j',w_k} = u_A(w_k) = 1 = u_A(w_j')=\sum_k \braket{w_j',w_k} = \sum_k M_{jk}$. Thus $M_{jk}$ is doubly stochastic, and in analogy to \cite{WehnerEntropy} we find that entropy increases:\\
$S(w) = S(\sum_j p_j w_j) = H(\vec p)$ where $H(\vec p) = -\sum_j p_j \ln p_j$ is the Shannon entropy\footnote{In this proof, the Shannon entropy is defined with $\ln$ instead of $\log_2$. Furthermore, we set $k_B = 1$ to simplify notation.}. By Birkhoff's theorem, $M=\sum_{\sigma \in S_N}a_\sigma \cdot \sigma$ is a convex combination of permutations. By Schur concavity of the Shannon entropy, $H(\vec q) \ge \sum_{\sigma \in S_N} a_\sigma H(\sigma(\vec p)) = H(\vec p)$. $S(w') = H(\vec q)$ because the $w_j'$ form a frame. Thus $S(w') \ge S(w)$.
\end{proof}

\subsection{Second law for degenerate projective measurements}
Next we wish to consider also degenerate measurements, i.e. measurements that correspond to degenerate observables in quantum theory. For doing so, we will adapt a proof for the quantum case from \cite{BestBookEver} using the relative entropy, especially we will adapt the proof from \cite{BestBookEver} that the relative entropy is non-negative. However, as our projective measurements are not necessarily induced by projectors on an underlying pure state Hilbert space, we have to find a different way to take advantage of the fact, that the relative entropy is non-negative, than in \cite{BestBookEver}.

\begin{defi}
	Assume strong symmetry and classical decomposability are fulfilled. Then the \textbf{relative entropy} is defined as:
		\begin{equation}
			S(w||v) := -S(w)-\braket{w, \ln v}
		\end{equation}
	Here, for $v=\sum_j q_j v_j$ a classical decomposition into a maximal frame, $\ln v := \sum_j \ln(q_j) v_j$.
\end{defi}

\vspace{1mm}
In quantum theory, our definition reproduces the standard definition of relative entropy from \cite{BestBookEver}. Like there, we will show that the relative entropy is never negative:
\vspace{1mm}

\begin{theorem}[Klein's inequality]
	For all $v,w \in \Omega_A$:
	\begin{equation}
	S(w||v) \ge 0 
	\end{equation}
\end{theorem}

\begin{proof}
	Consider classical decompositions $w=\sum_j p_j w_j$ and $v=\sum_k q_k v_k$ into maximal frames. Then:
		\begin{equation}
			S(w||v) = \sum_j p_j \ln p_j -\sum_{j,k} p_j \ln q_k \braket{w_j, v_k} =: \sum_j p_j \left( \ln p_j - \sum_k P_{jk} \ln q_k \right)
		\end{equation}
	Here, $P_{jk} := \braket{w_j,v_k} \ge 0$ and $\sum_j P_{jk} = \sum_k P_{jk} = 1$ because maximal frames add up to the order unit. Define $r_j := \sum_k P_{jk} q_k = \braket{w_j,v}$. $\ln$ is strictly concave and thus $\sum_k P_{jk} \ln q_k \le \ln r_j$. Thus:
	\begin{equation}
		S(w||v) \ge \sum_j p_j \ln\left( \frac{p_j}{r_j} \right)
	\end{equation}	
	 As $r_j = \sum_k P_{jk} q_k \ge 0$ and $\sum_j r_j = \sum_k \sum_{j} P_{jk} q_k = \sum_k q_k = 1$, the $r_j$ also form a probability distribution. By positivity of the classical relative entropy (see \cite{BestBookEver},Theorem 11.1), we thus find
	\begin{equation}
		S(w||v) \ge 0
	\end{equation}

	Alternative proof from Peres' book \cite{Peres}:
	\begin{equation}
		S(w||v) = \sum_j p_j \left( \ln p_j - \sum_k P_{jk} \ln q_k \right) = \sum_{jk} p_j P_{jk}\ln\left(\frac{p_j}{q_k}\right)
	\end{equation}
	We use $\ln(x) \ge 1-\frac 1 x$ with equality exactly for $x=1$:
	\begin{equation}
		S(w||v) \ge \sum_{jk} p_j P_{jk}\left[1-\frac{q_k}{p_j}\right] = \sum_j p_j -\sum_j r_j = 0
	\end{equation}
\end{proof}

\begin{lem}
	Assume strong symmetry and classical decomposability are fulfilled. Let $P_j$ be symmetric, positive projectors which form an operation. Then $P_k P_j = \delta_{jk} P_j$, i.e. the $P_j$ are mutually orthogonal. \label{Lemma:OrthogonalProjectors}
\end{lem}

\begin{proof}
	If $P_j w = 0$, then trivially $P_k P_j w = 0$. If $P_j w \ne 0$, then 
	\begin{align}
	1 &= u_A\left( \frac{P_j w}{(u_A\circ P_j)(w)}\right) = \left(u_A\circ\sum_k P_k \right)\left( \frac{P_j w}{(u_A\circ P_j)(w)}\right)\\&= u_A\left( \frac{P_j w}{(u_A\circ P_j)(w)}\right) + \left(u_A\circ\sum_{k\ne j} P_k \right)\left( \frac{P_j w}{(u_A\circ P_j)(w)}\right).
	\end{align}
	Thus $\left(u_A\circ\sum_{k\ne j} P_k \right)\left( \frac{P_j w}{(u_A\circ P_j)(w)}\right)=0$. By positivity of $P_k$, $\left(u_A\circ P_k \right)\left( \frac{P_j w}{(u_A\circ P_j)(w)}\right)=0$ for $k\ne j$. As $0$ is the only state with $u_A = 0$, $P_k P_j w = 0$. As the cone is generating ($\rm{Span}(A_+) = A$), $P_k P_j = 0$, especially $\braket{P_k w,P_j w} = 0$.
\end{proof}

\begin{lem}
	Let $P$ be a positive, symmetric, normalization-non-increasing projector which projects onto the linear span of a face $F$. For states $w$, $Pw$ is always found in the face $F$. \label{lemma:face_instead_of_span}
\end{lem}
\begin{proof}
	$w$ a state. If $Pw=0$, then trivially $Pw \in F$. If $Pw \ne 0$:  By surjectivity onto $\rm{Span}(F)$, $\frac{Pw}{u_A(Pw)} = \sum_j p_j w_j$ for $p_j \in \mathbb R$ and $w_j \in F$. We assume $p_j \ne 0$ and $w_j \ne 0$ (otherwise they do not contribute) and normalize: $w_j':= \frac {w_j}{u_A(w_j)}$, $p_j' := u_A(w_j) p_j$. $w_j'$ is still found in $F$, because $w_j = \frac{1}{\lambda}(\lambda w_j) + (1- \frac{1}{\lambda})\cdot 0$ for $\lambda > 1$ and $w_j = \lambda w_j +(1-\lambda)\cdot 0$ for $0<\lambda < 1$ together with $F$ being a face of $A_+$ imply that $\mathbb R_{\ge 0} \cdot w_j \subset F$. Now we find $\frac{Pw}{u_A(Pw)} =  \sum_j p_j' w_j'$. As all states are normalized, $\sum_j p_j' = 1$. Thus, the state $\frac{Pw}{u_A(Pw)}$ is an affine combination of states in $F$. By Proposition 2.10 from \cite{Pfister}, $F=\text{aff}(F)\cap A_+$. Thus $Pw \in F$. 
\end{proof}

\begin{lem}
	Assume Postulates 1 and 2. Let $F\ne \{0\}$ be a face of $A_+$ and $w \in F$. Then there exists a classical decomposition $w=\sum_j p_j w_j$ which only uses states in $F$, i.e. $w_j \in F \ \forall j$. \label{lemma:FaceClassicalDecomposition}
\end{lem}
\begin{proof}
	Let $w= \sum_j p_j w_j$ be a classical decomposition. Wlog $p_j > 0 \ \forall j$. As $F$ is a face, we find $w_j \in F \ \forall j$. 
\end{proof}

Now we finally consider the entropy in degenerate projective measurements.\\
We remember: We consider observables $\mathcal A = \sum_a a \sum_j w_{j;a}$. Here $a \in \mathbb R$ are the eigenvalues and the $w_{j;a}$ form a maximal frame. The $u_a := \sum_{j} w_{j;a}$ are the projective units of the eigenfaces $F_a$ with positive symmetric projector $P_a$.\\
We know that the projective units are valid effects with $u_a = u_A\circ P_a$. Furthermore, $\sum_a u_a = \sum_{a,j} w_{j;a} = u_A$ by having a maximal frame.\\

Vice versa, consider an operation given by symmetric projectors $P_a$ which project onto the linear spans of faces $F_a$. We know by Proposition \ref{Prop:PositiveProjector} that these projectors are positive. By Lemma \ref{Lemma:OrthogonalProjectors}, the $P_a$ are mutually orthogonal and thus also the faces $F_a$ are orthogonal. For all $a$, let $w_{j;a}$ be frames that generate $F_a$. Then $u_a := u_A \circ P_a$ is the corresponding projective unit and by that a valid effect. As the $P_a$ form an operation, $\sum_{a,j} w_{j;a} = \sum_a u_a = \sum_a u_A \circ P_a = u_A$, i.e. the $w_{j;a}$ form a maximal frame in total. Thus we can define an observable $\mathcal A = \sum_{a}a u_a = \sum_a a \sum_j w_{j;a}$.

\begin{theorem}
	Suppose strong symmetry and classical decomposability are satisfied. Let $P_a$ be symmetric projectors which form a valid operation and project onto the linear spans of faces $F_a$. Then the induced measurement with post-measurement ensemble state $w'=\sum_a P_a w$ does not decrease entropy: $S(w')\ge S(w)$
\end{theorem}

\begin{proof}
	Let $u_a:= u_A \circ P_a$ be the projective order units and $w_{j;a}$ frames that generate the $F_a$. We have already argued that the projectors and the faces are mutually orthogonal. We have also seen that the $w_{j;a}$ form a maximal frame in total.\\ 
	\\
	We consider $-S(w)-\braket{w,\ln w'} =S(w||w') \ge 0$. Like in \cite{BestBookEver} Theorem 11.9, we claim $\braket{w,\ln w'} = -S(w')$. If that claim is true, then $S(w') \ge S(w)$. Thus we only have to prove this claim, but as our theories cannot be assumed to have an underlying pure state Hilbert space, we will use a different proof.\\
	\\
	As the $P_a$ are mutually orthogonal, so are the $P_a w$. Because of Lemma \ref{lemma:face_instead_of_span}, $P_a w \in F_a$. If $P_a w = 0$ we use the decomposition $P_a w = \sum_j u_A(P_a w)\cdot r_{aj}\cdot w_{aj}$ with $w_{aj}:=w_{j;a}$ and $r_{aj} := \delta_{1,j}$.  If $P_a w \ne 0$, we perform a classical decomposition $\frac{P_a w}{u_A(P_a w)} = \sum_k r_{ak} w_{ak}$ for $r_{ak} >0$. Because of the classical decomposition, the $w_{ak}$ are mutually orthogonal for same $a$. The $w_{ak}$ are found in the face $F_a$ corresponding to the projectors $P_a$. Therefore, we add terms $r_{aj}\cdot w_{aj}=0\cdot w_{aj}$ to the classical decomposition of $\frac{P_a w}{u_A(P_a w)}$ to complete the $w_{aj}$ to a generating frame of $F_a$. As all generating frames of a face have the same size, also the new generating frames add up to the order unit in total, i.e. can be combined to a maximal frame. Thus $P_a w_{ak} = w_{ak}$ and as these faces are mutually orthogonal, also for different $a$ the $w_{ak}$ are mutually orthogonal. By orthogonality, $P_b w_{ak}=P_b P_a w_{ak}=0$ for $b\ne a$.\\
	In total, we have found a classical decomposition $w' = \sum_{aj} u_A(P_a w) \cdot r_{aj}\cdot w_{aj}$ with $P_a w_{bj} = \delta_{ab} w_{bj}$.
	With $\ln(w') = \sum_{aj} \ln(u_A(P_a w) \cdot r_{aj}) w_{aj}$ we find:
	\begin{equation}
	\sum_a P_a \ln(w') = \sum_{abj} \ln(u_A(P_b w) \cdot r_{bj}) P_a w_{bj} = \sum_{aj} \ln(u_A(P_a w) \cdot r_{aj}) w_{aj} = \ln(w')
	\end{equation}
	Finally, using symmetry of the projectors and 
	\begin{equation}
	  -S(w') = \sum_{aj} \big[u_A(P_a w) \cdot r_{aj}\big] \ln\big[u_A(P_a w) \cdot r_{aj}\big] = \braket{w',\ln w'}
	\end{equation} 
	for the same classical decomposition:
	\begin{align}
		\braket{w,\ln w'} = \braket{w, \sum_a P_a \ln w'} = \braket{\sum_a P_a w, \ln w'} = \braket{w',\ln w'} =-S(w')
	\end{align}
\end{proof}

\begin{note}
	If you read through the proofs carefully, you will see that the proofs would also work if $\log(w)$ depended on the choice of classical decomposition.
\end{note}

\subsection{Mixing processes}
\label{Section:Mixing}
Now we consider the entropy for a mixing process, i.e. what happens if we mix some gases which contain GPT-ensembles, see also Figure \ref{Fig:MixingProcess}. For the last time in this thesis, we use the convention that $S$ is the entropy proportional to the number of systems, while $s= \frac{S}{N}$ is the entropy divided by the number of systems, and we explicitly list $k_B$.\\ 

Like in von Neumann's thought experiment, we consider boxes, each of them filled with a GPT system, forming an ideal gas. We consider $n$ tanks at temperature $T$. In the $j$-th tank one finds $N_j$ boxes and each of these boxes contains the state $w_j$. Furthermore, the $j$-th tank is assumed to have volume $V_j := \frac{N_j}{N} V$, where $N:= \sum_{j=1}^n N_j$ is the total number of boxes/systems. Note that the gases in the tanks all have the same density. As the tanks are isolated from each other, the total GPT entropy is given by $S_\text{before} = \sum_j S_j(w_j)$ where $S_j$ is the GPT entropy of the gas in the $j$-th tank: For a classical decomposition $w_j = \sum_k p_k^{(j)} w_k^{(j)}$ we thus have $S_j(w_j) = - N_j k_B \sum_k p_k^{(j)} \ln p_k^{(j)}$. We use the normalized entropy $s$ which is already divided by the number of particles, i.e. $s(v) = - k_B \sum_k q_k \ln q_k$ for a state $v$ with classical decomposition $v = \sum_k q_k v_k$. We write $S_\text{before} = \sum_j N_j s(w_j)$. The tanks are merged to a giant tank. The walls separating the tanks are removed such that the gases mix. Now, we can put the walls back in, a process that is now reversible as the gases are already mixed.  \\
The new ensemble state is given by $w' := \sum_j \frac{N_j}{N} w_j$ because with probability $\frac{N_j}{N}$, a random box belonged to the $w_j$-gas before. The total GPT entropy after this mixing process is given by $S_\text{after} = S(w') = N s(w')$.\\
We see that the tanks which originally contained the gases $w_j$ now contain the mixed gas $w'$ at same conditions $T$,$N_j$, $V_j$. Thus the only difference in entropy is caused by the GPT-systems. Therefore, we need to check
\begin{equation}
	S_\text{before} = \sum_j N_j s(w_j) \le S_\text{after} = N s(w')
\end{equation}
which is equivalent to:
\begin{equation}
	 \sum_j \frac{N_j}{N} s(w_j) \le s\left(\sum_j \frac{N_j}{N} w_j\right)
\end{equation}

\begin{center}
	\includegraphics[width=\textwidth]{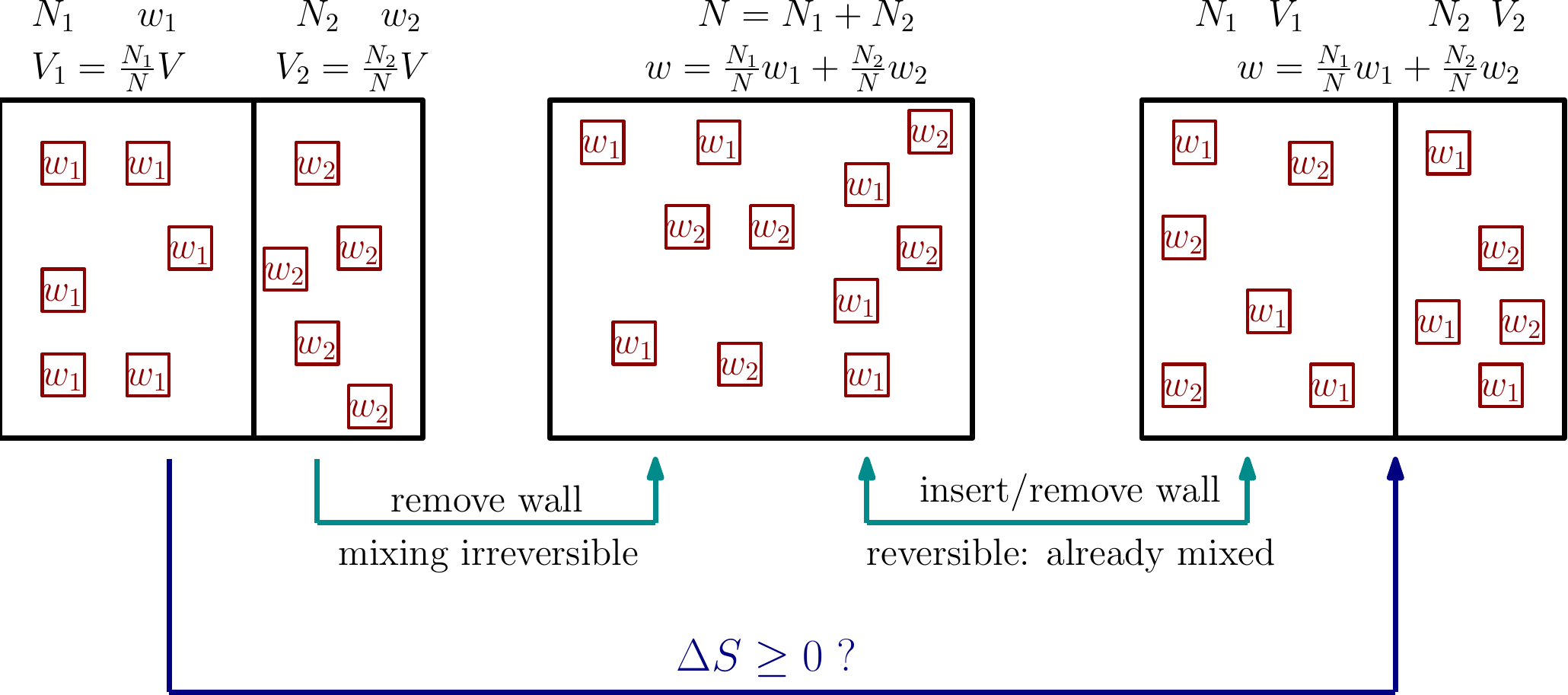}
	\captionof{figure}{\textit{\small At first, different pure gases are contained in tanks of same density. Then the gases are mixed by removing the walls. Afterwards the walls can be put back in giving the tanks from before but now with mixed gases in them. Does the entropy in the irreversible mixing procedure increase?}}
	\label{Fig:MixingProcess}
\end{center}

By continuity, we thus have to check for all normalized states $w_j \in \Omega_A$ and probability distributions $p_j$, that $s(\sum_j p_j w_j) \ge \sum_j p_j s(w_j)$. Therefore we have argued using a thermodynamic argument, that the entropy must be concave as function.\\
\\
As we now return to the mathematical properties of our entropy, we also return to the convention, that $S(\sum_j p_j w_j) = -\sum_j p_j \ln p_j$ for a classical decomposition $\sum_j p_j w_j$, i.e. we set $k_B =1$ and divide by the number of boxes. Now we show that the entropy is concave: 
\begin{theorem}
	Assume Postulates 1 and 2.\\
	The entropy is concave: Let $w_1,...,w_n \in \Omega_A$ and $p_1,...,p_n$ a probability distribution. Then:
	\begin{equation}
		S\left(\sum_j p_j w_j\right) \ge \sum_j p_j S(w_j)
	\end{equation}
\end{theorem}
\begin{proof}
	\begin{align*}
		0 &\le \sum_j p_j S(w_j||\sum_k p_k w_k) = -\sum_j p_j S(w_j) -\sum_j p_j\braket{w_j, \ln \left(\sum_k p_k w_k\right)}\\
		&= -\sum_j p_j S(w_j) -\braket{\sum_j p_j w_j, \ln \left(\sum_k p_k w_k\right)} = -\sum_j p_j S(w_j)+S\left(\sum_j p_j w_j\right)
	\end{align*}
\end{proof}

\section{Information-theoretic and operational considerations about the entropy}
\label{Section:InformationTheory}
\subsection{Motivation, definitions and conventions}
So far, all our considerations about the entropy have been from a thermodynamic point of view. However, as the GPT framework has a very close connection to quantum information theory and its operational way of thinking, it is not surprising that there has already been some work on the entropy using an operational/information-theoretic approach. In \cite{WehnerEntropy} and \cite{BarnumEntropy} (see also \cite{JapanEntropy} and \cite{LogicEntropy}) , two different entropies were defined for very general GPTs. One of them considers measurements, the other the construction of states. In classical and quantum theory, both agree with regular Shannon/von Neumann entropy. But in general GPTs, this is not necessarily in the case. The purpose of this section is to analyse these two entropies in the context of the strong structure provided by our postulates. One of the main results, Theorem \ref{theorem:MeasIsTherm}, namely that the measurement entropy coincides with our spectral definition of the entropy, was found in a collaboration between our group and Howard Barnum\footnote{Originally, we only considered the generalization of the Shannon entropy. But the other R$\acute{\text{e}}$nyi entropies use the exact same proof.}. The proof is a generalization of the proof for quantum theory found in \cite{WehnerEntropy}, Lemma B.1.\\
We will adapt the conventions from \cite{WehnerEntropy} and introduce some basic definitions about fine-graining of measurements, measurement entropy and decomposition entropy used there. Afterwards, we derive our results.\\
As this chapter considers the entropy from an information-theoretic perspective, all entropies in this chapter are defined with $\log := \log_2$ instead of $\ln$ and we omit $k_B$.\\
\\
Let $e_1,...,e_n$ and $f_1,...,f_m$ be two normalized measurements such that there exists a map $M: \{1,...,n\} \rightarrow \{1,...,m\}$ with 
\begin{equation}
	\sum_{\{ j| M(j) = k \}} e_j = f_k \ \ \ \forall k \in \{1,...,m\}
\end{equation}
If $M$ is bijective, then the measurement $\mathbf f$ is simply a \textbf{re-labelling} of $\mathbf e$. If there exists a $k$ with $M(j) \ne k \ \ \forall j$, then because of the normalization of the $\mathbf e$-measurement, $f_k = 0$ i.e. it is a trivial outcome that never happens. If $M$ is not injective, then $\mathbf f$ is a \textbf{coarse-graining} of $\mathbf e$ (or vice versa, $\mathbf e$ a \textbf{refinement} of $\mathbf f$) in the sense that $\mathbf f$ is obtained from $\mathbf e$ by collecting several outcomes of $\mathbf e$ and giving them a common, new outcome-label (and maybe adding the $0$-effect a few times). In that sense, we do not care about which of the $e_j$ triggered the new effect. An example is shown in Figure \ref{Fig:FineGraining}. \\
\begin{center}
	\includegraphics[width= 0.8 \textwidth]{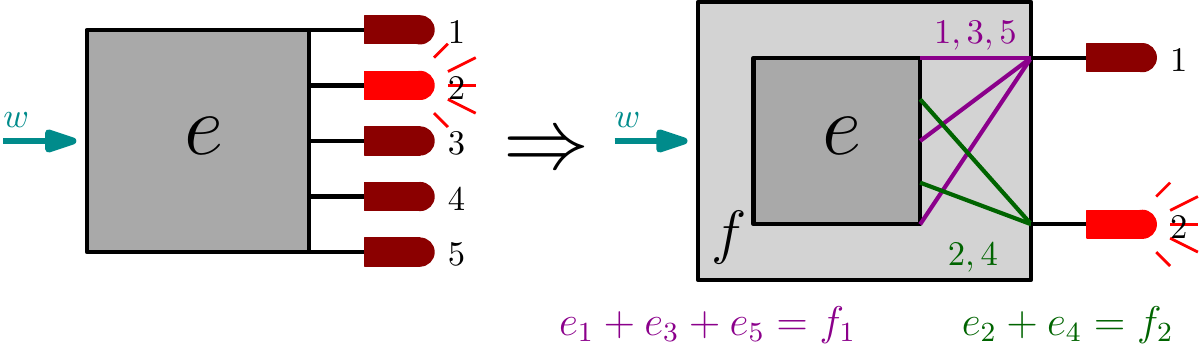}
	\captionof{figure}{\small \textit{This figure shows an example for measurements $\mathbf e$,$\mathbf f$ such that $\mathbf e$ is a refinement of $\mathbf f$. Here, the outcomes $1,3,5$ of $\mathbf e$ each trigger outcome $1$ of $\mathbf f$, while the outcomes $2,4$ of $\mathbf e$ trigger outcome $2$ of $\mathbf f$. Thus a $\mathbf f$-measurement can be constructed from a $\mathbf e$-measurement by checking if we have one of the outcomes $1,3,5$ or one of the outcomes $2,4$.}}
	\label{Fig:FineGraining}
\end{center}
However, there exist \textbf{trivial} refinements/coarse-grainings: for those, $e_j \propto f_{M(j)} \ \ \forall j$. We write $e_j = p_{j} f_{M(j)}$. Then such a measurement can be obtained by performing $\mathbf f$, and if outcome $k$ is triggered, we activate a classical random number generator which generates the final outcome $j$ among the $j$ with $M(j) = k$ with probability 
\begin{equation}
	\frac{p_{j}}{\sum_{\{a | M(a) = k\}} p_a}
\end{equation}
Thus a trivial refinement does not yield any additional information about the GPT-system. We only get additional information about the classical random number generator used at the end. So trivial refinements have no additional advantage in analyzing GPT-systems. And as we want to quantify the information of GPT-sources or the information missing about a GPT-system, we are not interested in any classical random number generator used at the end. An example is shown in Figure \ref{Fig:TrivialFineGraining}.\\
We call a measurement \textbf{fine-grained} if it does not have any non-trivial refinements. We call $\mathcal{E}^*$ the set of fine-grained measurements.\\
\begin{center}
	\includegraphics[width= 0.8 \textwidth]{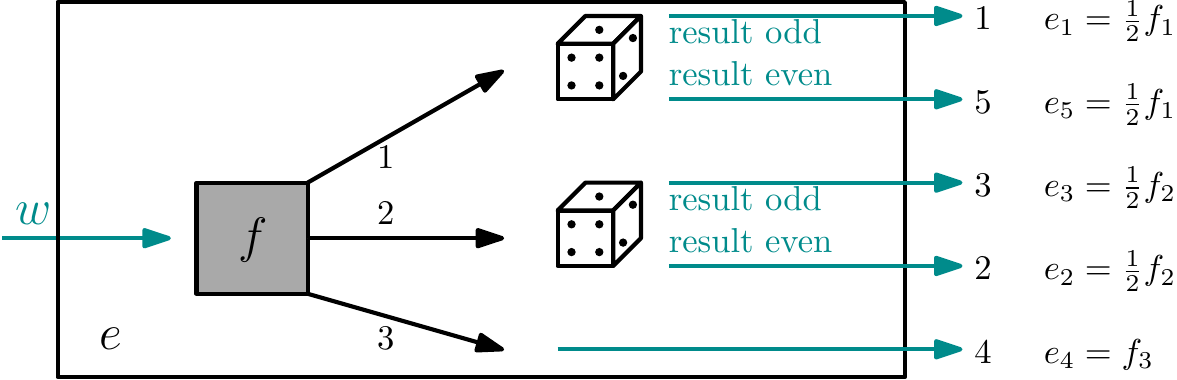}
	\captionof{figure}{\small \textit{This figure shows an example for measurements $\mathbf e$,$\mathbf f$ such that $\mathbf e$ is a trivial refinement of $\mathbf f$. Here, $e_1 = e_5 = \frac 1 2 f_1$, $e_2=e_3=\frac 1 2 f_2$ and $e_4 = f_3$. A $\mathbf e$-measurement can be constructed by using a $\mathbf f$-measurement and by using a classical random number generator afterwards, here a die. Thus the fine-graining from $\mathbf f$ to $\mathbf e$ does not yield any additional information about the GPT-system.} }
	\label{Fig:TrivialFineGraining}
\end{center}
Now we consider the R$\acute{\text{e}}$nyi entropies\cite{Renyi}, which are often used in information theory:\\
For a probability distribution $p=(p_1,p_2,...)$ the \textbf{R$\acute{\text{e}}$nyi entropies} are defined as:
	\begin{equation}
		H_\alpha(\mathbf p) = \frac 1 {1-\alpha} \log\left(\sum_j p_j^\alpha\right)
	\end{equation}
	where $\alpha \in ]0,\infty[$, $\alpha \ne 1$.
	Furthermore,
	\begin{equation}
		H_0(\mathbf p) := \lim_{\alpha \rightarrow 0} H_\alpha(\mathbf p) = \log | \rm{supp}(\mathbf p)| 
	\end{equation}
	with $\rm{supp}(\mathbf p)=\{p_j \ | \ p_j > 0\}$ is called the \textbf{max-entropy} and 
	\begin{equation}
		H_\infty(\mathbf p) := \lim_{\alpha \rightarrow \infty} H_\alpha(\mathbf p) = -\log \max_j p_j
	\end{equation}
	is called the \textbf{min-entropy}.
	Also, 
	\begin{equation}
		H_1(\mathbf p) := \lim_{\alpha \rightarrow 1} H_\alpha(\mathbf p) = -\sum_j p_j \log p_j = H(\mathbf p)
	\end{equation}
	is just the regular Shannon entropy $H$.\\
	For $\alpha \in [0,\infty]$ and GPTs satisfying Postulates 1 and 2, we generalize the classical R$\acute{\text{e}}$nyi entropies:
	\begin{equation}
		H_\alpha(w) = H_\alpha(\mathbf p)
	\end{equation}
	where $w=\sum_j p_j w_j$ is any classical decomposition.\\
	Following \cite{WehnerEntropy}, we introduce the R$\acute{\text{e}}$nyi measurement entropies and R$\acute{\text{e}}$nyi decomposition entropies which can also be used in GPTs which do not satisfy our postulates:\\
	For $\alpha \in [0,\infty]$, we define the \textbf{order-}$\mathbf{\alpha}$ \textbf{R$\acute{\text{e}}$nyi measurement entropy} to be:
	\begin{equation}
		\widehat{H}_\alpha(w) = \inf_{\mathbf e \in \mathcal{E}^*} H_\alpha(e_1(w),e_2(w),...)
	\end{equation}
	where $H_\alpha$ on the right hand side denotes the classical R$\acute{\text{e}}$nyi entropy.
	One uses fine-grained measurements because they yield the most information. Taking the infimum has two advantages: First of all, we eliminate the useless classical information caused by trivial refinements. Secondly, a (fine-grained) measurement with minimal entropy can be used to characterize a system; for example in quantum theory, particles prepared in a state $\ket{\psi}$ which all give the same energy in energy measurements would be said to be in an energy eigenstate. If instead we performed a position measurement, we would have a higher entropy, but the result of such a measurement is not a good characterization of the state (before the measurement).\\ 
	The \textbf{order-}$\mathbf{\alpha}$ \textbf{R$\acute{\text{e}}$nyi decomposition entropy} is defined as:
	\begin{equation}
		\widecheck{H}_\alpha(w) := \inf_{w=\sum_j q_j v_j \text{ convex decomp., } v_j\in \Omega_A \text{ pure}} H_\alpha(\mathbf q)
	\end{equation}
	This definition can be justified as follows: Assume we want to prepare a state $w$ by using states of maximal knowledge (i.e. pure states) $v_j$ and a random number generator, which gives output $j$ with probability $p_j$. Thus for a device, which outputs $v_j$ with probability $p_j$ and is described by $w$, we are interested in the lowest information content/entropy of the random number generator necessary to build such a device.\\
	Like in \cite{WehnerEntropy}, the \textbf{measurement entropy} is defined as 
	\begin{equation} \widehat H (w) := \inf_{\mathbf e \in \mathcal{E}^*} H(\mathbf e (w)) = \inf_{\mathbf e \in \mathcal{E}^*}\left[ -\sum_j e_j(w) \log e_j(w)\right]\end{equation}
	i.e. the measurement entropy is the order-$1$ R$\acute{\text{e}}$nyi measurement entropy. Similarly, the \textbf{decomposition entropy} is defined as
		\begin{equation}
		\widecheck{H}(w) := \inf_{w=\sum_j q_j v_j \text{ convex decomp., } v_j\in \Omega_A \text{ pure}} H(\mathbf q) 
	\end{equation}
	and coincides with the order-$1$ R$\acute{\text{e}}$nyi decomposition entropy.
\\
The order-0 R$\acute{\text{e}}$nyi entropy and the order-2 R$\acute{\text{e}}$nyi entropy will have a special role in our discussion, so we will motivate why they are interesting from an information-theoretic point of view (see also Figure \ref{Fig:Renyi}):\\ 
\\
At first we focus on the max-entropy, i.e. the order-0 R$\acute{\text{e}}$nyi entropy, i.e. $H_0(w) = \log | \rm{supp}(\mathbf p) |$ for $w =\sum_j p_j w_j$ a classical decomposition:\\
A preparation device randomly generates one of the states $w_j$ with probability $p_j$. Wlog, we assume the number of such $w_j$ with $p_j \ne 0$ is a power of 2. We want to ask yes-no-questions to determine which state is prepared. Then there exists a strategy which needs exactly $\log |\rm{supp}(\mathbf p)|$ questions, no matter what the state is and how unlikely it is\footnote{If the number of $w_j$ with $p_j >0$ is not a power of 2, we can still apply the same strategy by adding some probability-zero states until we reach a power of two. Then we find the true state in $\lceil \log |\rm{supp}(\mathbf p)| \rceil$ steps, and we will need more than $\lfloor \log |\rm{supp}(\mathbf p)| \rfloor$ steps. In that sense, $\log |\rm{supp}(\mathbf p)|$ can be seen as a continuous interpolation if the number of relevant $w_j$ is not a power of 2.}. No better strategy can be found which guarantees to need less steps, no matter how unlikely the state is: We consider only the possible states, i.e. those with $p_j \ne 0$. There are $|\rm{supp}(\mathbf p)|$ of those. We split these states into two groups of same size, and ask whether the state is in the first group. By that we eliminate one half of the states. Afterwards by the same procedure, we eliminate one half of the remaining states, and so on until only one remains. Thus we use $\log_2 |\rm{supp}(\mathbf p)|$ steps. Assume we ask another question, which might split the states into other fractions $q, 1-q$. Then we might be unlucky and the state is in the larger set, thus we have eliminated less states than with the strategy described before. As we want a guarantee for the maximal number of steps needed, no matter how unlucky we are or unlikely the (possible) state is, $\log |\rm{supp}(\mathbf p)|$ is the minimum number of steps we can guarantee.\\ 
\\
Next, we explain the relevance of the order-2 R$\acute{\text{e}}$nyi entropy, i.e. $H_2(w) = - \log(\sum_j p_j^2)$ for $w=\sum_j p_j w_j$ any classical decomposition. Consider two independent sources which are described by $w$ in the sense that they prepare $w_j$ with probability $p_j$. Then $\sum_j p_j^2$ is the probability that both sources have prepared the same state. Thus the order-2 R$\acute{\text{e}}$nyi entropy is also called collision entropy, ``collision'' meaning that the output of both independent sources is the same. 
\begin{center}
	\includegraphics[width= 0.8 \textwidth]{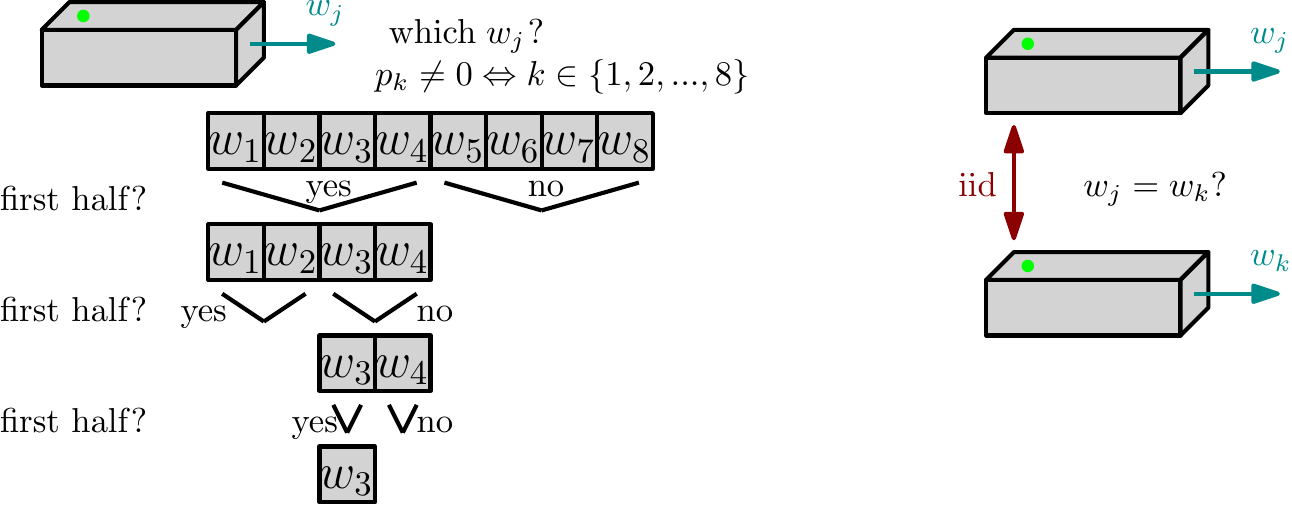}
	\captionof{figure}{\small \textit{The left part of this figure visualizes the argument used for the max-entropy for 8 states. No matter how unlikely the true state is, we only need 3 questions to find it. The right part visualizes the ``collision'' of the collision entropy.}}
	\label{Fig:Renyi}
\end{center}

\subsection{Results}
\begin{lem}
	Consider a GPT which satisfies the postulates of classical decomposability and strong symmetry from \cite{Postulates}. Let $(e_1,...,e_n)$ be a fine-grained measurement.\\
	 Then $e_j = c_j\braket{w_j,\cdot}$ with $c_j \in [0,1]$ and $w_j$ normalized and pure. 
	 \label{Lemma:FineGrainedEffect}
\end{lem}

\begin{proof}
	Let $j \in \{1,...,n\}$ be arbitrary. If $e_j = 0$, just choose $c_j=0$ and any pure state to be $w_j$.\\ Otherwise, because of the self-duality, there is a $w' \in A_+$ such that $\braket{w',\cdot}  = e_j$. As $w' \ne 0$ especially $u_A(w') \ne 0$. We can write $w' = c w$ with $w\in \Omega_A$, $c \in \mathbb R_{>0}$.\\
	 Now assume that $w$ is not pure. Then it has a non-trivial classical decomposition $w = \sum_{k=0}^N p_k v_k$, wlog $p_k \in ]0,1]$, $v_k$ pure, perfectly distinguishable and none of them equal to $w$. Especially $N \ge 1$, otherwise we had $w= v_0$ pure in contradiction to our assumption. Then $e_j = \sum_{k=0}^N c \cdot p_k \braket{v_k,\cdot}$. Wlog, we only consider the case $j = n$ (the other cases are included by relabelling such that the effect in consideration is called $e_n$), i.e. we have $e_n = \sum_{k=0}^N c \cdot p_k \braket{v_k,\cdot}$.\\
	 We define the measurement $e'_1 := e_1,...,e'_{n-1} := e_{n-1}$ and $e'_{n+i} := c p_i \braket{v_i,\cdot}$ for all $i \in \{0,...,N\}$. $\braket{v_i,\cdot} \ge 0$ on all states and thus $0\le c p_i \braket{v_i,\cdot} \le e_n \le u_A$ on all states which means the $e'_{n+i}$ are indeed effects. Especially $\sum_{i=0}^N e'_{n+i} = e_n$. We find $\sum_{k=1}^{n+N} e'_k = \sum_{k=1}^{n-1} e_k + \sum_{i=0}^N c p_i \braket{v_i,\cdot} = \sum_{k=1}^{n} e_k = u_A$, i.e. the $e'_1,...,e'_{n+N}$ form a measurement.\\
	 We define $M : \{1,...,n+N\} \to \{1,...,n\}$ by $M(i) := i$ for $1\le i \le n-1$, $M(i) := n$ for $i \ge n$. Then   
	\begin{align}
		&\sum_{\{j | M(j) = i\}} e'_j = e_i &\text{ for } i < n\\
		&\sum_{\{j | M(j) = i\}} e_j' = \sum_{j=n}^{n+N} e_j' = e_n  &\text{ for } i = n
	\end{align}
	Thus $(e'_1,...,e'_{n+N})$ is a fine-graining of $(e_1,...,e_n)$. It is non-trivial, because $e'_n = c p_0\braket{v_0,\cdot}$ is not proportional to  $e_n = \sum_{k=0}^N c \cdot p_k \braket{v_k,\cdot}$ (just check with $v_0$ and $v_1$ to see that this is true). This is in contradiction to the requirement that $(e_1,...,e_n)$ is a fine-grained measurement. Thus $w$ is pure. As $c = c \braket{w,w} =\braket{w',w}  = e_j(w) \le 1$, we have $c\in ]0,1]$.
\end{proof}

\vspace{5mm}

\begin{lem}
	Consider a GPT with classical decomposability and strong symmetry. $w\in \Omega_A$ arbitrary, $w= \sum_{j=1}^d p_j w_j$ a classical decomposition into a maximal frame\footnote{This can always be obtained by extending the frame of a classical decomposition to a maximal frame and adding the new elements with coefficients 0 to the classical decomposition. If one uses a smaller frame $w_1,...,w_n$ one can also use a (fine-grained) measurement $e_1,...,e_d$, which distinguishes a maximal frame-extension $w_1,...,w_d$, to distinguish the $w_1,..,w_n$. The effects $e_{n+1},...,e_d$ will always give 0. Especially, the Shannon entropy of the measurement probabilities is still $-\sum_{j=1}^d e_j(w)\log e_j(w) = -\sum_{j=1}^n p_j\log p_j$, which will be important for Theorem \ref{theorem:MeasIsTherm}. }. Then the measurement which perfectly distinguishes the $w_j$ (i.e. $e_k(w_j) = \delta_{jk}$) can be chosen to be fine-grained.
	\label{Lemma:FineGrainedFrame}
\end{lem}
\begin{proof}
	We consider $e_j := \braket{w_j,\cdot}$. As maximal frames add up to the order unit, they form a measurement. Now assume there is a fine-graining $e_k'$, $\sum_{\{j | M(j) = k\}} e_j' = e_k$. Using self-duality: $\sum_{\{j | M(j) = k\}} c_j' \braket{w_j',\cdot} = \braket{w_k,\cdot}$ where $c_j' \braket{w_j',\cdot} = e_j'$ with $w_j'$ normalized and $c_j' \ge 0$. So $\sum_{\{j | M(j) = k\}} c_j' w_j' = w_k$. Especially, $\sum_{\{j | M(j) = k\}} c_j'=\sum_{\{j | M(j) = k\}} c_j' u_A(w_j') = u_A(w_k) = 1$. Thus the $c_j'$ with $M(j)=k$ form a probability distribution and $\sum_{\{j | M(j) = k\}} c_j' w_j' = w_k$ is a convex decomposition of a pure state. This requires either $c_j' = 0$ or $w_j' = w_k$. In the first case $e_j' = 0 \propto e_k$, in the second case $e_j' = c_j' e_k \propto e_k$. Thus the fine-graining is trivial.
\end{proof}

\vspace{5mm}

\begin{lem}
	Let $\mathbf e=(e_1,...,e_N)\in \mathcal{E}^*$ be a fine-grained measurement in a theory fulfilling Postulates 1 and 2. Let $w\in \Omega_A$ be a state with classical decomposition $w=\sum_j p_j w_j$. Let $\mathbf{q} := (e_j(w))_j$ be the vector of outcome probabilities. Furthermore, $ \mathbf{p} := (p_j)_j$.\\
Then $d\le N$ where $d$ is the maximal frame size, i.e. the dimension (sometimes denoted $N_A$).\\
Furthermore, if $\mathbf{p'} := ( \mathbf{p},0,...,0)$ is an extension of $\mathbf{p}$ to an $N$-dimensional vector by adding zeroes (which is always possible because of $d \le N$), then $\mathbf q \prec \mathbf p'$, i.e. there is a bistochastic $N\times N$-matrix $M$ such that $\mathbf q = M \mathbf{p'}$. \label{Lemma:M}
\end{lem}

\begin{proof}
		Let $(e_1,...,e_N)$ be a fine-grained measurement. By our Lemma \ref{Lemma:FineGrainedEffect}, all effects are of the form $e_j = c_j \braket{w_j',\cdot}$ with $c_j \in [0,1]$ and $w_j'$ normalized and pure. Furthermore $\sum_{j=1}^N e_j = u_A$ as this is a measurement. Let $d$ denote the maximal frame size.\\
We define $q_l := e_l(w) = c_l \braket{w_l' , w}$. Now let us proof $\sum_{j=1}^N c_j = d$:\\
	As maximal frames add up to the order unit: $\sum_{j=1}^d v_j = u_A$ for all maximal frames $v_1$,...,$v_d$. So 
	\begin{align}
		\sum_{j=1}^N c_j &= \sum_{j=1}^N c_j u_A(w'_j) = \sum_{j=1}^N c_j \braket{w'_j, u_A} = \sum_{j=1}^N  e_j(u_A) = u_A(u_A) = \braket{u_A,u_A} \\
		&= \sum_{j,k = 1}^d \braket{v_j,v_k} = \sum_{j,k = 1}^d \delta_{jk} = d
	\end{align}
	Especially, as $c_j \le 1$, this implies $d\le N$.\\
	Now consider a classical decomposition of $w$ and extend it to a maximal frame decomposition $w=\sum_{j=1}^d p_j w_j$ by adding zeroes. Note that adding zeroes will not change $\mathbf{p'}$. Define $q_{l|j} := e_l(w_j)$.
	\begin{align}
		\sum_{j=1}^d q_{l|j} p_j &= \sum_{j=1}^d e_l(p_j w_j) = e_l(w) = q_l\\
		\sum_{j=1}^d q_{l|j} &= \sum_{j=1}^d e_l(w_j) = c_l \sum_{j=1}^d \braket{w'_l,w_j} = c_l u_A(w_l') = c_l\\
		\sum_{l=1}^N q_{l|j} &= \sum_{l=1}^N e_l(w_j) = u_A(w_j) = 1
	\end{align}
	Once more, we used that maximal frames add up to the order unit and that measurement effects sum up to the order unit.\\
	We extend $\mathbf p=(p_1,...,p_d)$ to a $N$-component vector $\mathbf {p'} := (p_1,...,p_d,0,...,0)$. Furthermore we define $M_{l,j} := q_{l|j}$ for $j \le d$ and $M_{l,j} := \frac{1-c_l}{N-d}$ for $j > d$ (the latter case can only happen if $d < N$). $M$ is a bistochastic $N \times N$-matrix:\\
	\begin{align}
		\sum_{l=1}^N M_{l,j} &= \sum_{l=1}^N q_{l|j} = 1 \text { for } j\le d\\
		\sum_{l=1}^N M_{l,j} &= \sum_{l=1}^N \frac{1-c_l}{N-d} = \frac{N-d}{N-d} = 1 \text{ for } j>d\\
		\sum_{j=1}^N M_{l,j} &= \sum_{j=1}^d q_{l|j} + (N-d) \frac{1-c_l}{N-d} = c_l +1-c_l = 1
	\end{align}
	Furthermore $M_{l,j}\ge 0$ as $q_{l|j} = e_l(w_j)\ge 0$ and $\frac{1-c_l}{N-d} \ge 0$ by $c_l \le 1$, $d < N$ (for $N=d$, there is no case $j>d$).
	$M$ has the important property that $\mathbf q = M \cdot \mathbf{p'}$: 
	\begin{equation}
		\sum_{j=1}^N M_{l,j} p_j' =  \sum_{j=1}^d M_{l,j} p_j =  \sum_{j=1}^d q_{l|j} p_j = q_l 
	\end{equation}
	Thus $M$ is bistochastic and by Birkhoff's theorem, $M$ is a statistical mixture of permutations: $M = \sum_{\sigma \in S_N} P_\sigma \sigma$.
\end{proof}

\vspace{5mm}

\begin{theorem}
	Consider a GPT which satisfies classical decomposability and strong symmetry. Then the R$\acute{\text{e}}$nyi entropies and the R$\acute{\text{e}}$nyi measurement entropies coincide, i.e.
	\begin{equation}
		H_j(w) = \widehat{H}_j(w) \ \forall w \in \Omega_A,\ j \in [0,\infty]
	\end{equation}
	In particular the measurement entropy is the same as the thermodynamic/spectral entropy\footnote{The proofs of Theorem \ref{theorem:MeasIsTherm} for the measurement entropy together with Lemma\ref{Lemma:M} where found in a collaboration with Howard Barnum and are an adaptation of the quantum proof found in \cite{WehnerEntropy}, Lemma B.1}, which for a state $w$ with classical decomposition $w = \sum_j p_j w_j$ is defined as $S(w) = -\sum_j p_j \log p_j$, i.e. we have  $\widehat H (w) = S(w) \ \ \forall w\in \Omega_A$. 
	\label{theorem:MeasIsTherm}
\end{theorem}

\begin{proof}
Once more for any fine-grained measurement $e_1,...,e_N$, we set $q_l := e_l(w)$, $\mathbf{p'} =(\mathbf{p},0,...,0)$ and let $M=\sum_{\sigma \in S_n} P_\sigma \sigma$ be the bistochastic $N\times N$ matrix with $\mathbf q = M \cdot \mathbf{p'}$, compare Lemma \ref{Lemma:M}.
 As the Shannon entropy is Schur-concave, we find 
	\begin{equation}
		H(\mathbf q) \ge \sum_{\sigma \in S_N} P_\sigma H(\sigma(\mathbf p')) = \sum_{\sigma \in S_N} P_\sigma H(\mathbf p') = H(\mathbf p') = H(\mathbf p) = - \sum_{j=1}^d p_j \log p_j = S(w)
	\end{equation} 
	Note that $H(\mathbf p) = - \sum_{j=1}^d p_j \log p_j = S(w)$ is the result of a measurement that distinguishes the states $w_j$: $e_k(w_j) = \delta_{jk}$. By Lemma \ref{Lemma:FineGrainedFrame} such a measurement can indeed be chosen to be fine-grained.\\
	For the other R$\acute{\text{e}}$nyi entropies $H_\alpha$, the same argument can be used because their classical counterpart is also Schur-concave.
\end{proof}

\vspace{5mm}

In the context of Postulates 1 and 2, Postulate 3 is equivalent to the \textbf{covering property} (see \cite{Postulates}):
\begin{defi}
	A GPT satisfies the \textbf{covering property iff:}\\ For every face $F$ and an atom $\hat{a}$ (i.e. the face given by a single pure state), the smallest face $F \vee \hat{a}$ containing $F$ and $\hat{a}$ is either $F$ itself or it covers $F$, i.e. $F\vee \hat{a}$ is a larger face than $F$ and there is no other face in between them.\\ In our context, covering property thus means:
$F$ a face of $A_+$, $w$ pure then the face $G$ generated by $F$ and $w$ has rank $|G| \le |F|+1$.\\	
\end{defi}

\begin{lem}
	Consider a GPT satisfying Postulates 1 and 2, and the covering property.\\
	Let $w_1,...w_n$ be pure states. Then the face $F$ generated by them satisfies $|F|\le n$. In particular, for every state in $F$, there is a classical decomposition using at most $n$ pure states.
	\label{Lemma:Covering}
\end{lem}

\begin{proof}
	We prove by using an induction:\\
	In the case of just a single pure state, the decomposition is unique.\\
	Now assume the statement is true for all sets of $n$ or less pure states. We also want to show that the statement is true for all sets of $n+1$ pure states:\\
	Let $w_1,...,w_{n+1}$ be pure states. By the induction hypothesis, the face $F$ generated by $w_1,...,w_n$ has a rank of at most $n$.\\
	By the covering property, the face $F\vee \big(\mathbb R_{\ge 0}\cdot \{ w_{n+1} \}\big)$ contains $w_1$, $w_2$,...,$w_{n+1}$ and has a rank of at most $n+1$, i.e. it is generated by a frame of at most $n+1$ states. As generated faces are minimal, also the face generated by $w_1$, $w_2$,...,$w_{n+1}$ may have no higher rank.
\end{proof}

\vspace{5mm}

Now we ask whether there is a convex decomposition into pure states which needs less trials than a frame decomposition.

\begin{theorem}
	Consider a GPT satisfying Postulates 1 and 2.\\
	Then Postulate 3 is equivalent to $\widecheck{H}_0 = \widehat{H}_0 = H_0$, i.e. in the context of Postulates 1 and 2, Postulate 3 is true exactly if the order-0 R$\acute{\text{e}}$nyi entropies coincide.
\end{theorem}

\begin{proof}
At first, we consider the ``$\mathbf{\Leftarrow}$''-direction:\\
Assume $H_0 = \widecheck{H}_0$, but that the covering property is not fulfilled. Then there exist a face $F$ and a pure state $w$, such that the face $G$ generated by both has rank $|F|+2$ or higher. Let $w_1,...,w_{|F|}$ be a frame generating the face $F$. Then $F$ is also generated by $v:= \frac{1}{|F|}\sum_{j=1}^{|F|} w_j$, i.e. the normalized projective unit. This statement is clear, because every face containing $v$ also contains all the $w_j$ (and vice versa), and $F$ is the smallest face for that.\\
So consider now the state $\frac 1 2 w + \frac 1 2 v$. If this state had a classical decomposition using only $|F|+1$ (or less) perfectly distinguishable pure states $v_1,...,v_{|F|+1}$, then the face generated by this frame already contains $v$ (and thus $w_1,...,w_{|F|}$) and $w$, but it only has rank $|F|+1$ in contradiction to our assumption. Thus the classical decomposition of $\frac 1 2 w + \frac 1 2 v$ uses at least $|F|+2$ perfectly distinguishable pure states with non-zero coefficients, i.e. $H_0(\frac 1 2 w + \frac 1 2 v) \ge \log(|F|+2)$. But $\frac 1 2 w + \frac 1 2 v =\frac{1}{2|F|} \sum_{j=1}^{|F|}w_j + \frac{1}{2} w$ is a convex decomposition into $|F|+1$ pure states, i.e.
\begin{equation}
	\widecheck{H}_0 \left(\frac 1 2 w + \frac 1 2 v\right) \le \log (|F|+1) < \log (|F|+2) \le H_0 \left(\frac 1 2 w + \frac 1 2 v \right)
\end{equation}
Once more, this is an contradiction to our assumption.\\
Thus if $H_0 = \widecheck{H}_0$, then also the covering property must hold.\\
\\
Now we proof the ``$\Rightarrow$''-direction:\\
Now assume that the covering property holds.\\
 Assume there is a state $w \in \Omega_A$ with $H_0(w) \ne \widecheck{H}_0(w)$. This requires $H_0(w) > \widecheck{H}_0(w)$, because the classical decompositions are also included in the infimum-definition of $\widecheck{H}_0$. As the minimization procedure only runs over values $1,2,...,|\Omega_A|$, the infimum actually is a minimum. Thus there exist pure states $w_1,...,w_{2^{\widecheck{H}_0(w)}}$ such that there is a convex decomposition: 
 \begin{equation} w= \sum_{j=1}^{2^{\widecheck{H}_0(w)}} p_j w_j \end{equation}
By Lemma \ref{Lemma:Covering}, the face generated by these pure states has at most rank $2^{\widecheck{H}_0(w)}$. By convexity, it includes $w$. Thus there is a classical decomposition of $w$ which uses a frame of size no larger than $2^{\widecheck{H}_0(w)}$. Thus $H_0(w) > \widecheck{H}_0(w)$ is impossible.
\end{proof}

\vspace{5mm}

However, the order-2 R$\acute{\text{e}}$nyi entropy $H_2(w) = -\log(\sum_j p_j^2)$ reduces without use of the third postulate, i.e. we do not need the third postulate for all R$\acute{\text{e}}$nyi entropies:
\begin{theorem}
	Consider a GPT which satisfies Postulates 1 and 2. Then for all states:
	\begin{equation}
		\widecheck{H}_2(w) = H_2(w)
	\end{equation}
\end{theorem}

\begin{proof}
	Let $w = \sum_j p_j w_j$ be a classical decomposition, and $w=\sum_j q_j v_j$ any convex decomposition into pure states. Then $\sum_j p_j^2 =\braket{w,w} = \sum_j q_j ^2 + \sum_{j\ne k} q_j q_k \braket{v_j, v_k}$. As $\braket{v_j, v_k} \ge 0$, we find $\sum_j q_j^2 \le \sum_j p_j^2$, thus $H_2(p) = -\log(\sum_j p_j^2) \le -\log(\sum_j q_j^2)$. Thus classical decompositions indeed minimize the collision entropy.
\end{proof}

\section{Weak spectrality does not imply spectrality}
\label{Section:Spectrality}
\subsection{The idea}
Our notion of classical decomposability is inspired by the diagonalization of density operators. If a state space satisfies the postulate of classical decomposability, we also say that it satisfies \textbf{weak spectrality}. However, in contrast to the eigenvalues of a density operator, it is not clear that the coefficients of a classical decomposition are unique except for permutation and zeroes. To show this result, we had to use the second postulate, strong symmetry. Thus we say a state space with weak spectrality satisfies \textbf{(unique/strong) spectrality}\footnote{The notions ``weak'' and ``unique spectrality'' have been defined by Howard Barnum while working with our group. This chapter is an answer for his question whether weak spectrality implies unique spectrality.}, if the coefficients are unique except for permutation and zeroes.\\
We want to construct an example which shows that weak spectrality does not imply unique spectrality. This means that we want to find an example for a state space, in which:
\begin{enumerate}
\item every state has at least one classical decomposition \vspace{-3mm}
\item there is at least one state which has at least two different classical decompositions whose coefficients are not just a permutation of each other
\end{enumerate}
This result shows that classical decomposability alone is not enough in order for our entropy to be well-defined.

To prove this, we consider an egg-like state space $\Omega_A$:
\begin{center}
\includegraphics[width=0.45\textwidth]{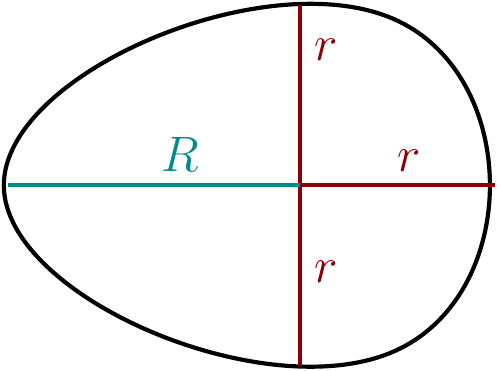} \vspace{-3mm}
\captionof{figure}{\small \textit{The state space is chosen to look like a 2D egg.}}
\end{center}
It is quite obvious that the state shown in the figure above has two different classical decompositions whose coefficients are not just a permutation of each other:
\begin{align}
	\begin{pmatrix} 0 \\ 0 \end{pmatrix} &= \frac 1 2 \begin{pmatrix} 0 \\ r \end{pmatrix}+ \frac 1 2 \begin{pmatrix} 0 \\ -r \end{pmatrix}= \frac{r}{r+R} \begin{pmatrix} -R \\ 0 \end{pmatrix} + \frac{R}{r+R} \begin{pmatrix} r \\ 0 \end{pmatrix}
\end{align}
So it remains to show that every state in this state space has a classical decomposition, i.e. that weak spectrality is satisfied.\\
The idea is like this: We put a tangent hyperplane (i.e. a straight tangent line) at an arbitrary boundary point of the egg. Afterwards, we consider a parallel hyperplane on the other side of the egg and move it towards the egg until it hits the egg.
\begin{center}
\includegraphics[width=0.45\textwidth]{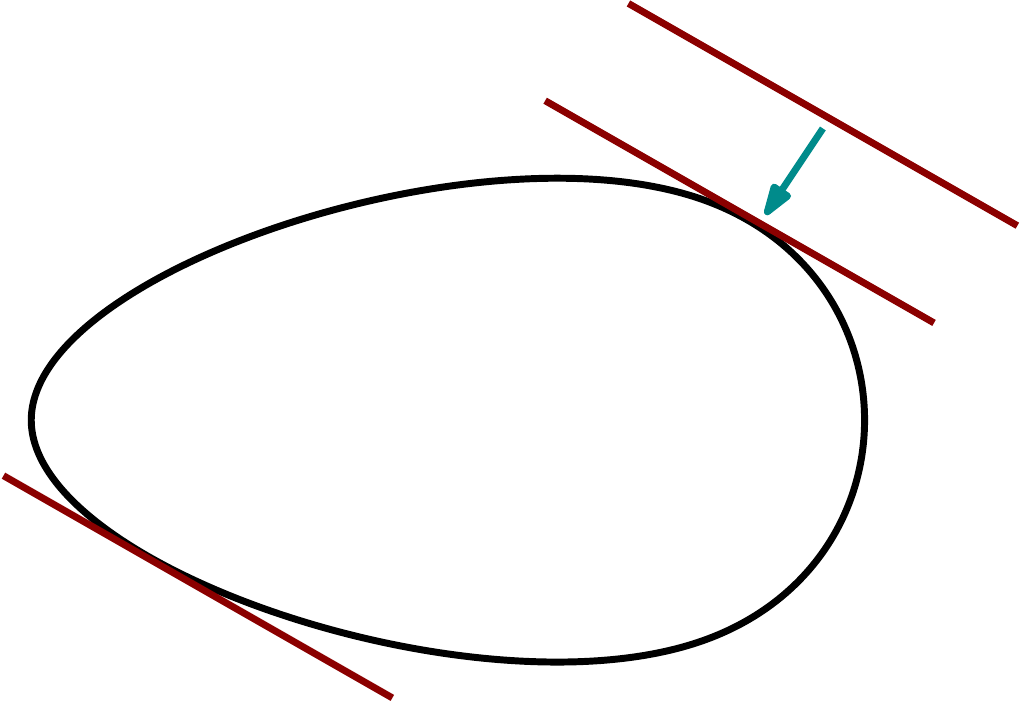} \vspace{-3mm}
\captionof{figure}{\small \textit{The first tangent line is chosen freely. The tangent line on the other side is constructed by moving a distant parallel line towards the egg.}}
\end{center}
We call the tangent lines $E_1$ and $E_2$, and the points where they and the egg intersect are called $p_1$ and $p_2$. The line that connects these points is called $l$. Now we move $p_1$ around the boundary of the egg, while keeping $E_1$ tangential. At the same time, we also move $p_2$ around the boundary such that $E_2$ stays tangential and parallel to $E_1$. In Appendix \ref{chap:Rotation}, there is a visualisation with Mathematica.
\begin{center}
\includegraphics[width=\textwidth]{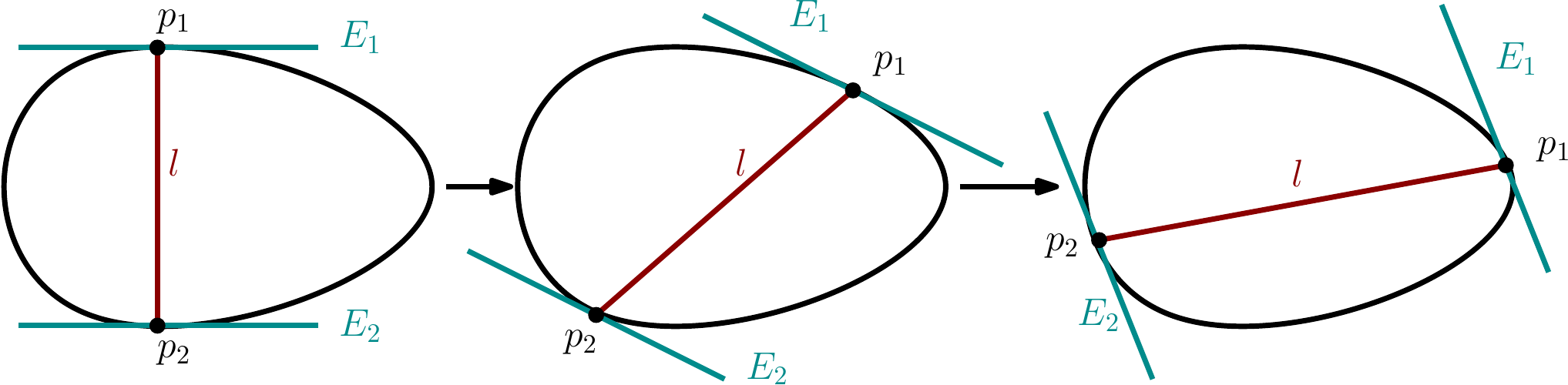}
\captionof{figure}{\small \textit{This figure shows how the tangential lines are rotated along the egg.}}
\end{center}
 The conjecture is that while doing this for full $360$ degrees, every point of the egg lies on the connection line $l$ at least once. If a point $p$ is found on $l$, it is a convex combination of $p_1$ and $p_2$. $p_1$, $p_2$ are pure. We define effects $e_1,e_2$ by $e_j(E_k) = \{\delta_{jk}\}$ and extend affine-linearly, as explained in Chapter \ref{Section:GPT}, Example \ref{Example:FigureToEffect}. Thus $p_1$ and $p_2$ are perfectly distinguishable pure states and $p$ has a classical decomposition. If the conjecture is true, then the egg satisfies weak spectrality. 

\subsection{Proof idea}
A ``topological'' proof idea is shown in Figure \ref{fig:top}. Like before, we put two parallel tangent hyperplanes (i.e. straight tangent lines) at the egg. Wlog, they are chosen such that the straight line through the egg connecting them coincides with the symmetry axis of the egg-shaped state space. \\
\begin{center}
\includegraphics[width=\textwidth]{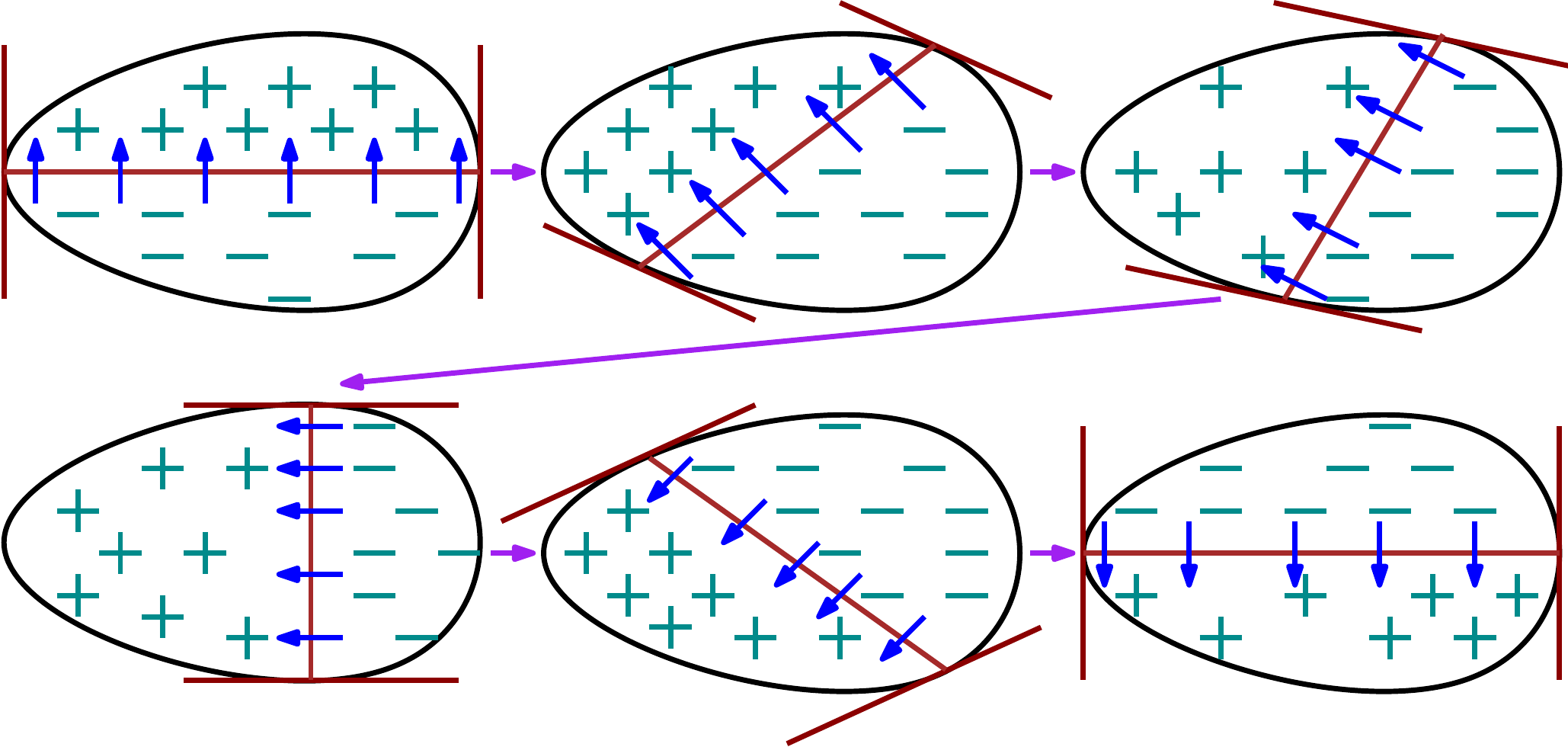}
\captionof{figure}{\small \textit{This figure visualizes an idea for a proof, that the egg space satisfies weak spectality.} }
 \label{fig:top}
\end{center}
Again, we call the tangent lines $E_1$ and $E_2$, and the points where they and the egg intersect are called $p_1$ and $p_2$. The line that connects these points is called $l$. The line $l$ splits the egg into two parts. We mark one part with $+$, the other with $-$. Now we move $p_1$ around the boundary of the egg, while keeping $E_1$ tangential. At the same time, we also move $p_2$ around the boundary such that $E_2$ stays tangential and parallel to $E_1$. We stop, when $p_1$ coincides with the other end of the symmetry axis. Then the situation looks almost like in the beginning, but the $+$- and $-$-parts are exchanged. This means that every point of the egg has changed its sign either from $+$ to $-$ or vice versa. As the line $l$ is moved continously through the egg, this requires that each point of the egg lay on $l$ at least once. This means that each point of the egg has at least one convex decomposition into perfectly distinguishable pure states ($p_1$ and $p_2$). \\
\\
While one can expect that this proof idea will work for many strictly convex state spaces, we will now consider an exact proof for our egg shaped state space based on this idea. We will need a parametrization of our egg shaped state space.\\

\subsection{Parametrization}
\begin{center}
\includegraphics[width=0.5\textwidth]{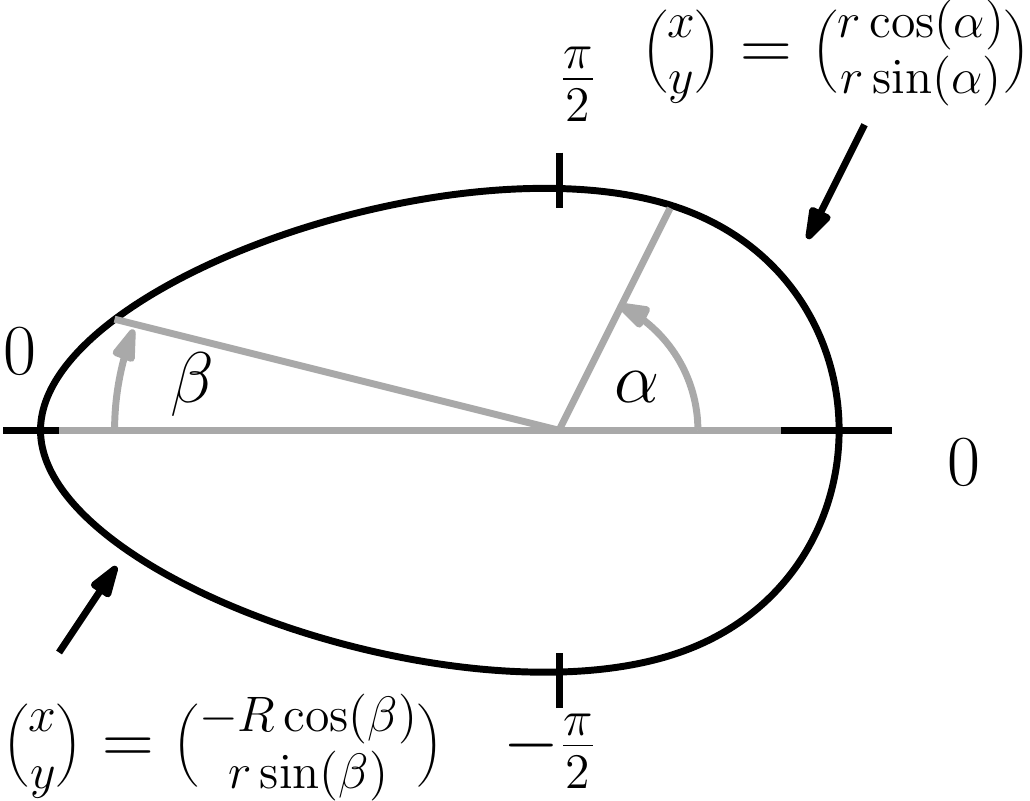}
\captionof{figure}{\small \textit{Detailed picture of how we parametrize the egg.}}
\end{center}
We want to provide an analytical proof, but also check our results with Mathematica. According to \cite{Trigonometrie}, Mathematica chooses to have $\text{arccot}$ values between $-\frac \pi 2$ and $\frac{\pi} 2$, the same for $\arctan$ \cite{Trigonometrie}.\\
So now we give the exact parametrization of the 2D egg:
The first part is a half-circle, given by 
\begin{equation}
	(x,y) = (r\cos(\alpha),r\sin(\alpha))
\end{equation}
for $\alpha \in [-\frac \pi 2, \frac{\pi} 2 ]$. The second part is an ellipse given by 
\begin{equation}
	(x,y) = (-R\cos(\beta), r\sin(\beta))
\end{equation}
for $\beta \in [-\frac \pi 2, \frac{\pi} 2 ]$. Choose an arbitrary angle $\alpha$. Then
\begin{align}
	\frac{\mathrm d y}{\mathrm d \alpha} = r \cos(\alpha)
	&&\frac{\mathrm d x}{\mathrm d \alpha} = -r \sin(\alpha)
\end{align}
Thus the slope of the tangential line is given\footnote{We apply the chain rule here: $\frac{\mathrm d y}{\mathrm dx} \approx \frac{y(\alpha)-y(\alpha_0)}{x(\alpha)-x(\alpha_0)} = \frac{y(\alpha)-y(\alpha_0)}{\alpha-\alpha_0} / \frac{x(\alpha)-x(\alpha_0)}{\alpha-\alpha_0} \approx  \frac{\mathrm d y}{\mathrm d\alpha} / \frac{\mathrm d x}{\mathrm d\alpha}$} by $\frac{\mathrm d y}{\mathrm d x} = -\cot(\alpha)$.
Now we want to find the parallel line on the other side. At first
\begin{align}
	\frac{\mathrm d y}{\mathrm d \beta} = r \cos(\beta)
	&&\frac{\mathrm d x}{\mathrm d \beta} = R \sin(\beta)
\end{align}
Thus the slope on the other side is given by $\frac{\mathrm d y}{\mathrm d x} = \frac{r}{R} \cot(\beta)$. We want the slopes to agree, i.e. $\frac{r}{R} \cot(\beta) = -\cot(\alpha)$. Thus 
\begin{equation}
	\beta(\alpha) = \text{arccot}\left(-\frac{R}{r} \cot{\alpha}\right)
\end{equation}
The line connecting $p_1$ and $p_2$ is 
\begin{equation}
	l_\alpha(p) = p\begin{pmatrix}r\cos(\alpha)\\ r\sin(\alpha) \end{pmatrix} +(1-p) \begin{pmatrix}-R\cos(\beta(\alpha)) \\ r\sin(\beta(\alpha)) \end{pmatrix}
\end{equation}
where $p \in [0,1]$.\\
\\
To prove weak spectrality, we have to show that every point of the egg lies on one of the $l_\alpha$. Based on our parametrization, we can implement this with Mathematica, see Appendix \ref{chap:WSMathematica}. In Appendix \ref{chap:AlternativeParam}, the same is done for an alternative parametrization. For both parametrizations, we first plot the boundary of the egg space. Afterwards, we plot all $l_\alpha(p)$, i.e. all points that have at least one classical decomposition. As one can see, Mathematica suggests that all points have a classical decomposition. Thus weak spectrality holds, but not unique spectrality as already argued in the beginning. Take a moment to appreciate that by choosing $\alpha, \beta \in [-\pi/2, \pi/2]$, we have chosen the same conventions as Mathematica.

\subsection{Exact proof}
Now, we show an analytical proof that does not need Mathematica:\\
So far, we have found 
\begin{align}
	l_\alpha(p) &= p\begin{pmatrix}r\cos(\alpha)\\ r\sin(\alpha) \end{pmatrix} +(1-p) \begin{pmatrix}-R\cos(\beta(\alpha)) \\ r\sin(\beta(\alpha))\end{pmatrix} \\&= \begin{pmatrix}-R\cos(\beta(\alpha)) \\ r\sin(\beta(\alpha))\end{pmatrix}+p\begin{pmatrix}r\cos(\alpha)+R\cos(\beta(\alpha))\\ r\sin(\alpha)-r\sin(\beta(\alpha)) \end{pmatrix} =: \vec{a_{\alpha}}+p\cdot \vec{t_{\alpha}}
\end{align}
where 
\begin{equation}
	\beta(\alpha) = \text{arccot}\left(-\frac{R}{r} \cot{\alpha}\right)
\end{equation}
$\vec{t_\alpha}$ gives the direction of the straight line $l_\alpha$. A vector perpendicular to $l_\alpha$ is thus given by:
\begin{equation}
	\vec{n_\alpha} := \begin{pmatrix} -r\sin(\alpha)+r\sin(\beta(\alpha)) \\ r\cos(\alpha)+R\cos(\beta(\alpha))  \end{pmatrix}
\end{equation}
While intuitively clear, we want to prove that $\beta(\alpha)$ is well-defined and continous. \\
Figures \ref{fig:cot} and \ref{fig:arccot} show, that the only critical arguments are $\alpha=0$ (here cot diverges) and $\alpha=\pm \frac \pi 2$ (here arccot$(0)$).
\begin{align}
	\lim_{\alpha \to 0} \beta(\alpha) &= \text{arccot}(\mp \infty) = 0\\
	\lim_{\alpha \nearrow + \frac \pi 2} \beta(\alpha) &= \lim_{x\searrow 0} \text{arccot}(- x) = -\frac \pi 2 \\
	\lim_{\alpha \searrow - \frac \pi 2} \beta(\alpha) &= \lim_{x\nearrow 0} \text{arccot}(- x) = +\frac \pi 2 
\end{align}

\begin{minipage}{0.45\textwidth}
	\includegraphics[width=0.95\textwidth]{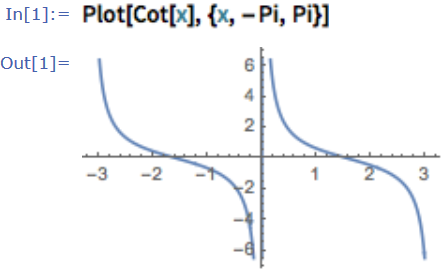}
	\captionof{figure}{\small \textit{The cot-function has zeroes for $\pm \frac{\pi}2$ and diverges in 0. (from\cite{Trigonometrie})}}
\label{fig:cot}
\end{minipage}
\hfill
\begin{minipage}{0.45\textwidth}
	\includegraphics[width=0.95\textwidth]{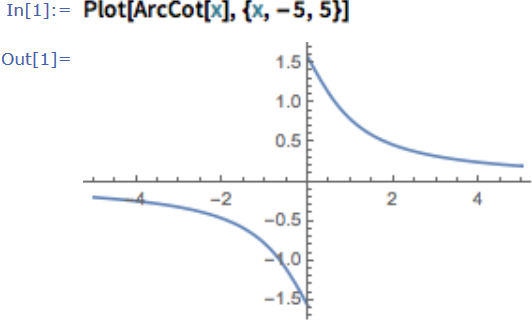}
	\captionof{figure}{\small \textit{The arccot function as chosen by Mathematica(\cite{Trigonometrie}) and us is discontinous in 0. For infinite argument, arccot approaches 0.}}
\label{fig:arccot}
\end{minipage}

Thus, as expected, for $\alpha=0$ we find $\beta = 0$, while for $\alpha=\pm \frac \pi 2$ we find $\beta = \mp \frac \pi 2$. Note that we used that $\alpha \in [-\frac \pi 2, \frac \pi 2]$. Thus $\beta(\alpha)$ is well-defined (by continuous extension) and continuous.\\
Now let $\vec w \in \Omega_A$ be an arbitrary state. Then:
\begin{align}
	\vec w \in l_\alpha &\Leftrightarrow \exists p \in [0,1]: \vec w = \vec{a_\alpha}+p\cdot \vec{t_\alpha} \Leftrightarrow  \exists p \in [0,1]: \vec w -\vec{a_\alpha} = p\cdot \vec{t_\alpha} \\ 
	&\Leftrightarrow \big[\vec w - \vec{a_\alpha} \big] \cdot \vec{n_\alpha} = 0 
\end{align}
In the last equivalence, for $\Rightarrow$ we multiplied with $\vec{n_\alpha}$ and used $\vec{n_\alpha} \cdot \vec{t_\alpha} = 0$. For $\Leftarrow$, we used that $\vec{n_\alpha}, \vec{t_\alpha}$ are an orthogonal basis\footnote{As $l_\alpha$ was constructed from two opposing points in the egg and $\vec{t_\alpha}$ is the vector connecting these points, $\vec{t_\alpha} \ne 0 \quad \forall \alpha$. Except for sign and exchange, $\vec{n_\alpha}$ has the same components and thus also never vanishes. }, thus $\big[\vec w - \vec{a_\alpha} \big] \cdot \vec{n_\alpha} = 0$ implies that 
$\vec w - \vec{a_\alpha} \propto \vec{t_\alpha}$. As $\vec w$ is a state\footnote{$l_\alpha$ connects two opposing boundary points given for $p=0$ and $p=1$. Thus for $p \notin [0,1]$, one leaves the egg. Also note, that the condition $p\in [0,1]$ in these equivalences is not really necessary. It is sufficient to check if $\vec w$ is on the straight line given by $l_\alpha$, $p \in [0,1]$ follows because $\vec w$ is a state.}, the proportionality constant has to be $\in [0,1]$.\\
Now consider an arbitrary state $\vec w \in \Omega_A$. We define the function $g_{\vec{w}} : \left[-\frac \pi 2, \frac \pi 2\right] \to \mathbb R$,
\begin{equation}
	 \ g_{\vec w}(\alpha) = \big[\vec w - \vec{a_\alpha} \big] \cdot \vec{n_\alpha} = \left[ \vec w -   \begin{pmatrix}-R\cos(\beta(\alpha)) \\ r\sin(\beta(\alpha))\end{pmatrix}\right] \cdot \begin{pmatrix} -r\sin(\alpha)+r\sin(\beta(\alpha)) \\ r\cos(\alpha)+R\cos(\beta(\alpha))  \end{pmatrix}
\end{equation}
$g_{\vec w}$ is a continuous function. Now we compare the situations for $\alpha = \pm \frac \pi 2$.\\
For $\alpha = \pm \frac \pi 2$, we find $\beta = \mp \frac \pi 2$ and thus $\vec{n_\alpha} = \begin{pmatrix} \mp 2r \\ 0 \end{pmatrix}$. Therefore:

\begin{align}
	g_{\vec w}\left(+\frac \pi 2\right) &= \left[ \vec w -   \begin{pmatrix}-R\cos(-\frac \pi 2) \\ r\sin(-\frac \pi 2)\end{pmatrix}\right] \cdot \begin{pmatrix} -2r \\ 0 \end{pmatrix} = - \left[ \vec w -   \begin{pmatrix}-R\cos(\frac \pi 2) \\ r\sin(\frac \pi 2)\end{pmatrix}\right] \cdot \begin{pmatrix} 2r \\ 0 \end{pmatrix} \nonumber \\
	&= -g_{\vec w}\left(-\frac \pi 2\right)
\end{align}
\begin{center}
	\includegraphics[width=0.8\textwidth]{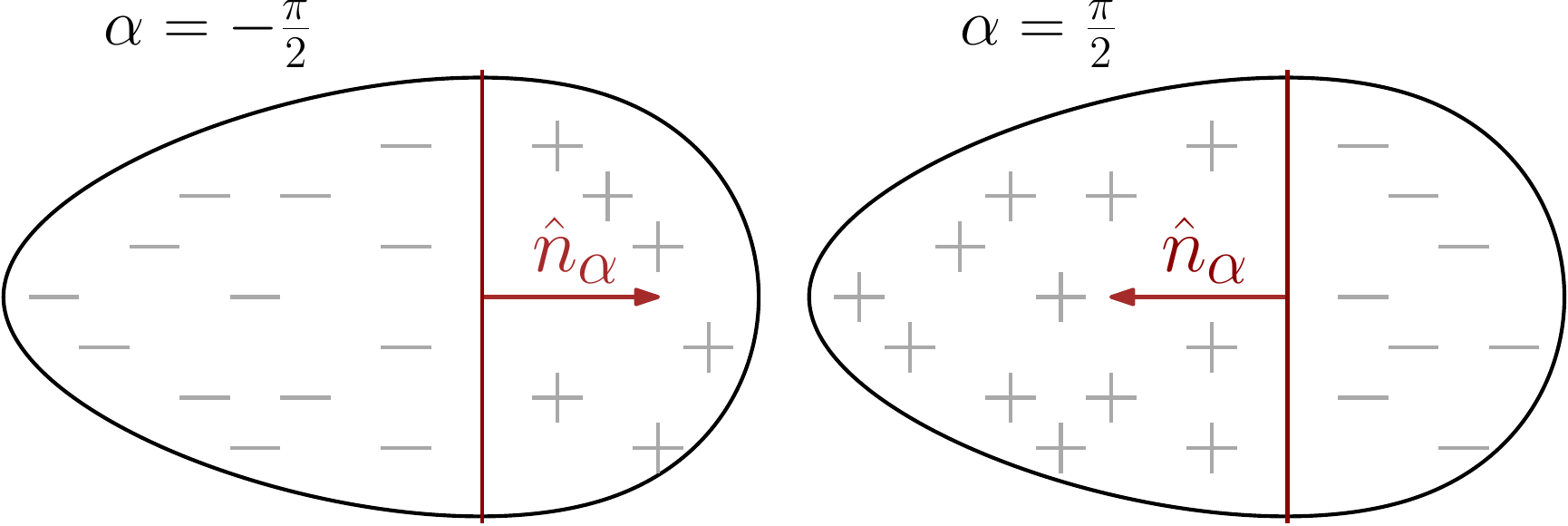}
	\captionof{figure}{\small \textit{This figure visualises the connection between the exact analytical proof and the proof idea presented before.}}
\end{center}

Thus, if not already $g_{\vec w}\left(+\frac \pi 2\right)=0$, $g_{\vec w}$ changes its sign. In that case, by the intermediate value theorem , there is an $\alpha \in [-\frac \pi 2, \frac \pi 2]$ with $g_{\vec w}(\alpha) = 0$.Thus in all cases, there is a $\alpha$ with $\big[\vec w - \vec{a_\alpha} \big] \cdot \vec{n_\alpha} = 0$, i.e. $\vec w \in l_\alpha$. This means that $\vec w$ has a classical decomposition. As $\vec w \in \Omega_A$ was arbitrary, weak spectrality/classical decomposability holds.

\subsection{Conclusion}
We have analysed a state space which has a 2D egg-shape. Each state in this state space has a classical decomposition (i.e. weak spectrality holds). We used Mathematica and an analytical proof (based on an idea that should also work for other strictly convex state spaces) to show that weak spectrality holds. But there is at least one state, which has classical decompositions that differ in their coefficients (i.e. unique spectrality does not hold). This has the important thermodynamic consequence, that an entropy similar to the von Neumann entropy would not be well-defined for a state, because its value depends on the choice of classical decomposition. Especially, the state would not fully characterize the ensemble, i.e. ``the state does not describe the state of the system''. Furthermore if the ensemble is not realized by perfectly distinguishable pure states, in general we cannot not define the entropy at all. In contrast, the postulate of (unique) spectrality directly leads to the generalization of the von Neumann entropy. In fact, a postulate or result of the form of unique spectrality is needed to know what the eigenvalues and orthonormal eigenvectors can be replaced with to formulate the generalized von Neumann entropy. Thus it is an interesting question to investigate what other principles or postulates lead to unique spectrality. 
\section{Pfister's state discrimination principle}
\label{Section:StateDiscrimination}
This section has no direct relevance for entropy or thermodynamics. However, the strong structure provided by our postulates also fulfils other physically motivated principles/postulates. One example is Pfister's \textit{state discrimination principle} \cite{Pfister}: To illustrate its significance, we start with a seemingly trivial example: We can perfectly distinguish hats and T-shirts. Among the T-shirts, we can distinguish blue T-shirts from black T-shirts. We can also distinguish hats, blue T-shirt and black T-shirt from each other. What seems quite trivial, is not necessarily true in all GPTs. We assume we wish to analyze an object which is either a GPT-hat or a blue or black GPT-T-shirt. If we ask \textit{Is it a GPT-hat?} it may happen that the object is destroyed or collapses into another object. Then in case of a GPT-T-shirt, asking \textit{Is it blue or black?} might not be possible or might give the wrong answer. Thus Pfister postulated that this will not happen. The basic idea is in the spirit of Specker's principle: 
\textit{``Do you know what, according to me, is the fundamental theorem of quantum mechanics? (. . . ) That is, if you have several questions and you can answer any two of them, then you can also answer all three of them''} (see e.g. \cite{Specker} for more information and references).

\begin{theorem}
	Consider a GPT which satisfies Postulate 1 (Classical Decomposability) and Postulate 2 (Strong Symmetry) from \cite{Postulates}.\\
	Then $A$ fulfils Pfister's\cite{Pfister} Postulate 3, the \textit{state discrimination principle}:\\
	Let $B_1, B_2 \subset \Omega_A$ be perfectly distinguishable sets of states. Assume that in addition, there are subsets $B_3, B_4 \subset B_2$ such that $B_3$ is perfectly distinguishable from $B_4$. Then $B_1,B_3,B_4$ are perfectly distinguishable. 
\end{theorem}

\begin{proof}
	By definition of perfect distinguishability, there are effects $e_1,e_2$ such that $e_k(B_j) = \{\delta_{jk}\}$ for $j,k =1,2$ and effects $e_3,e_4$ such that $e_k(B_j) = \{\delta_{jk}\}$ for $j,k =3,4$. Especially, $B_j$ lies in the faces $e_j^{-1}(1)$ and $e_k^{-1}(0)$ for $\{k, j\} = \{1,2\}$ or $\{k, j\} = \{3,4\}$ (interpreted as subsets of $\Omega_A$), compare Lemma \ref{Lemma:EffectFace}. Let $E_j$ be the minimal face, which contains $B_j$. Then especially $E_j \subset e_j^{-1}(1),e_k^{-1}(0)$ for $\{k,j\}= \{1,2\}$ or $\{k,j\} = \{3,4\}$.\\
There are frames $w_1^{(j)},...,w_{|E_j|}^{(j)}$ that generate $E_j$.  \\
There is a 1-1 correspondence between the faces of $\Omega_A$ and $A_+$, especially $\mathbb R_{\ge 0} \cdot e_j^{-1}(1)$, $\mathbb R_{\ge 0} \cdot e_k^{-1}(0)$, $\mathbb R_{\ge 0} E_j$ are the corresponding faces of $A_+$. \\
By Proposition \ref{Prop:FaceCorrespondence},  $\mathbb R_{\ge 0} \cdot E_j$ is also generated by the frame $w_1^{(j)},...,w_{|E_j|}^{(j)}$.\\ By Proposition \ref{Prop:ProjectiveUnit} the projective unit 
\begin{equation*}\sum_{a=1}^{|E_j|} \braket{w_a^{(j)},\cdot } = u_{[\mathbb R_{\ge 0}\cdot E_j]} =: u_j\end{equation*} 
 is an effect with $u_{[\mathbb R_{\ge 0}\cdot E_j]}(w) = 1$ for all $w \in \big(\mathbb R_{\ge 0}\cdot E_j\big) \cap \Omega_A = E_j$ and $0 \le u_{[\mathbb R_{\ge 0}\cdot E_j]} \le u_A$.\\ 
\\
By $e_1$ and $e_2$, $B_1$ and $B_3 \subset B_2$ are perfectly distinguishable. As $E_1$ and  $E_3 \subset E_2$ are the minimal faces which contain $B_1$ or $B_3 \subset B_2$, we find $E_1 \subset e_1^{-1}(1),e_2^{-1}(0)$ and $E_3 \subset e_2^{-1}(1),e_1^{-1}(0)$. By perfect distinguishability thus $\braket{w_a^{(1)}, w_b^{(3)}}= 0$. By the same reasoning $\braket{w_a^{(1)}, w_b^{(4)}}= 0$. Furthermore as the $w_a^{(j)}$ form frames for fixed $j$, $\braket{w_a^{(j)}, w_b^{(j)}} = \delta_{a,b}$.\\
Furthermore, $B_3$ and $B_4$ are perfectly distinguishable by $e_3$ and $e_4$. As $E_3$ and $E_4$ are the minimal faces which contain $B_3, B_4$ we find $E_3 \subset e_3^{-1}(1),e_4^{-1}(0)$ and $E_4 \subset e_4^{-1}(1),e_3^{-1}(0)$. Thus $E_3$ and $E_4$ are perfectly distinguishable by $e_3$ and $e_4$ and we find $\braket{w_a^{(3)},w_b^{(4)}} = 0$.\\
So in total, we found: $\braket{w_a^{(j)}, w_b^{(k)}} = \delta_{ab}\delta_{jk}$ for $j,k \in \{1,3,4\}$. 
Therefore $\{w_a^{(j)} |  j\in \{1,3,4\}\}$ is a frame and $\sum_{j=1,3,4} u_j = \sum_{j =1,3,4} \sum_a \braket{w_a^{(j)},\cdot} \le u_A$. \\ 
\\
The faces $E_1, E_3, E_4$ are orthogonal to each other because they are pairwise perfectly distinguishable. Thus $u_j(E_k) = 0$ for $j \ne k$. (Alternative reason: $\sum_{j=1,3,4} u_j \le u_A$ and thus $u_j(E_j) = 1$ implies $u_k(E_j) = \delta_{jk}$.)\\
\\
In total we found $u_j(B_k) = \{\delta_{jk}\}$ for $j,k \in \{1,3,4\}$. The $u_j$ are effects whose sum is no larger that $u_A$. Thus the $u_j$ perfectly distinguish $B_1,B_3,B_4$. 
\end{proof}

\section{Conclusion and outlook}
\label{Section:Outlook}
In this thesis, we have explored the thermodynamic consequences of two important postulates. These two postulates are the condensed form of important structural properties of quantum theory and classical theory. The first postulate is that every state belongs to a classical subspace, non-classical behaviour only possible because there might be different classical subspaces. The second postulate is the computational equivalence of all n-level systems.\\
Following a suggestion by J. Barrett, we adapted a thought experiment by von Neumann and an expansion thereof to derive a thermodynamic entropy for GPTs fulfilling the two postulates. This entropy is a direct analogue of the von Neumann entropy. We proved that this entropy is well-defined and we showed how it behaves for decomposition into perfectly distinguishable states. Furthermore, we constructed projective measurements and proved the second law for such measurements as well as mixing processes. In that context, we also generalized observables to our GPTs. We showed that the information-theoretic measurement entropy coincides with our entropy and generalized this result to other R\'{e}nyi entropies. Furthermore, we investigated the R\'{e}nyi decomposition entropies in the context of third order interference. Then we used an egg-shaped state space to show that the first postulate alone does not imply the second and thus does not always lead to a well-defined entropy. Furthermore, we showed that our GPTs satisfy the state discrimination principle formulated by Pfister.\\
\\
There are many ways how to build upon the results of this thesis:\\
 An interesting question is, whether the third postulate (no 3rd order interference) already is an consequence of the first two postulates or whether it is independent. If it is a consequence, then the GPTs considered by us all have the important property that they do no exhibit non-trivial 3rd order interference. If the third postulate does not follow from the first two, then there might be many more GPTs for which our results work. In particular, Postulates 1,2,4 provide entropy and energy, and thus allow to define the free energy and perform equilibrium thermodynamics. So if Postulate 3 really is necessary to obtain quantum theory, then we could define equilibrium thermodynamics for some non-classical and non-quantum systems. Also, we have seen that the third postulate is equivalent to the result, that the decomposition max-entropy is identical with the spectral definition of the max-entropy. Investigating this equivalence might be a way to find out whether the third postulate is independent of the first two postulates. In the same way, it would be interesting to analyze whether the decomposition entropy agrees with the spectral entropy, and especially whether we need the third postulate for this identity. Furthermore, it is an interesting question whether the first two postulates can be replaced by weaker postulates. For example, using unique spectrality leads to a well-defined generalization of the von Neumann entropy. As projective measurements play an important role, it would be interesting to check what consequences projective state spaces would have. Also one should try to consider infinite-dimensional GPTs.


\thispagestyle{empty}




\singlespacing


\appendix
\section{Appendix: Measurement effects can be extended to linear functions}
\label{Chapter:ConvexLinearToLinear}
In this appendix, we consider convex-linear effects only defined on $\Omega_A$ and show that that they can be extended to linear functions on $A$. The proof can also be found in \cite{Hardy} and \cite{Barrett}. It is included such that the introduction to GPTs is complete.\\ 
\\
For any state $w$ in $A_+$, there is a normalized state $v$ and some $p \ge 0$ with $w = p v$. Except for the $0$-state, we find $p = u_A(w)$, $v=\frac{w}{u_A(w)}$, i.e. the decomposition is unique (except for the $0$-state). Thus it is well-defined to set $e(w) = p e(v)$. This also is well-defined for the $0$-state, $e(0) = 0$. This definition agrees with our interpretation of subnormalized states: $e(w) = u_A(w) e(\frac{w}{u_A(w)})$ is the probability that $w$ is successfully prepared times the probability that $e$ is triggered, assuming that the preparation is successful. Especially, the impossible state never triggers $e$. Thus we find $e(p w) = p e(w)$ for all $p \in \mathbb R_{\ge 0}$, $w \in A_+$. Furthermore let $\sum_j p_j w_j$ with $p_j \in \mathbb R_{\ge 0}$, $w_j \in A_+$. Wlog, $p_j \ne 0$ and $w_j \ne 0$ as the zero state will cause no problems for the linearity. Then:
\begin{align}
	e\left(\sum_j p_j w_j\right) &= \sum_a p_a u_A(w_a) e\left(\sum_j \frac{u_A(w_j) p_j}{\sum_k p_k u_A(w_k)} \frac{w_j}{u_A(w_j)}\right)\\
		&= \sum_a p_a u_A(w_a) \sum_j \frac{u_A(w_j) p_j}{\sum_k p_k u_A(w_k)} e\left(\frac{w_j}{u_A(w_j)}\right)\\
		&= \sum_j p_j e(w_j)
\end{align}
For the first ``$=$'', we used $e(p w) = p e(w)$ for all $p \in \mathbb R_{\ge 0}$, $w \in A_+$ and that cones are closed under sums and multiplication with non-negative numbers. For the second ``$=$'', we used that $e$ is assumed to be convex-linear on $\Omega_A$. For the third ``$=$'' we used $e(p w) = p e(w)$ once more.\\
Now, let $w= \sum_j p_j w_j$ be an arbitrary linear combination with $w, w_j \in A_+$. Let $M_+$ be the set with $p_j \ge 0$, $M_-$ the set with $p_j< 0$. Then $w + \sum_{j \in M_-} |p_j| w_j = \sum_{j \in M_+} |p_j| w_j$ are positive linear combinations and thus $e(w) + \sum_{j \in M_-} |p_j| e(w_j) = \sum_{j \in M_+} |p_j| e(w_j)$ gives $e(w) = \sum_j p_j e(w_j)$.\\
As the cone is generating, there is a basis $v_1,...,v_n \in A_+$. We extend to $A$ by $e(\sum_j q_j v_j) = \sum_j q_j e(v_j)$ for all real $q_j$. This definition does not depend on the choice of basis, as $e$ already is linear on $A_+$. 
\section{Appendix: Correspondence between the faces of $\Omega_A$ and $A_+$}
\label{Section:FaceCorrespondence}

In this appendix we prove that there is a bijective correspondence between the faces of $\Omega_A$ and the faces of $A_+$.

\begin{prop}
	For an abstract state space, a face $F$ of $\Omega_A$ induces a face $\mathbb R_{\ge 0} \cdot F$ of $A_+$. Vice versa, a face $F\ne \{0\}$ of $A_+$ induces a face $\Omega_A \cap F$ of $\Omega_A$. Furthermore, if $w_1,...w_m\in \Omega_A$ generate the face $F$ of $\Omega_A$ ($A_+$), they also generate the corresponding face of $A_+$ ($\Omega_A$). Furthermore for a face $F\subset \Omega_A$, $\Omega_A \cap (\mathbb R_{\ge 0}\cdot F) = F$, and vice versa for a face $G\subset A_+$, $G \ne \{0\}$, we find $G= \mathbb R_{\ge 0} \cdot (\Omega_A \cap G)$.
\end{prop}

\begin{proof}
	Let $F$ be a face of $\Omega_A$. We show that $F^+ := \mathbb R_{\ge 0} \cdot F$ is a face of $A_+$:\\
	Let $v \in F^+$. $0\in F^+$ trivially because $F$ is a face and thus F is not empty. So assume $v\ne 0$. Then $\frac{v}{u_A(v)} \in F$: By definition, there is a $v'\in F\subset \Omega_A$ and a $\lambda \in \mathbb R_{\ge 0}$ with $v= \lambda v'$. Thus $u_A(v) = \lambda u_A(v') = \lambda$, therefore $v'= \frac{v}{u_A(v)}$. Thus by definition $\mathbb R_{\ge 0} \cdot v \subset F^+$. Consider now $w \in A_+$, $w=pw_1+(1-p)w_2$ with $w_1,w_2 \in F^+$ and $p\in (0,1)$. Wlog we can assume $w_1\ne 0 \ne w_2$ because we already know that all non-negative multiples are contained in $F^+$. Because of the normalization, this also implies $w\ne 0$. Thus: 
	\begin{equation}
	  \frac{w}{u_A(w)} = \frac{p u_A(w_1)}{u_A(w)} \frac{w_1}{u_A(w_1)} + \frac{(1-p)u_A(w_2)}{u_A(w)}\frac{w_2}{u_A(w_2)}
	\end{equation}
	Because of proper normalization and positivity, the coefficients $\frac{p u_A(w_1)}{u_A(w)}$ and $\frac{p u_A(w_2)}{u_A(w)}$ form a probability distribution. Furthermore, $\frac{w_1}{u_A(w_1)},\frac{w_2}{u_A(w_2)} \in F$ just like before. As $F$ is convex, $\frac{w}{u_A(w)} \in F$. Thus $w\in F^+$, i.e. $F^+ = \mathbb R_{\ge 0} \cdot F$ is convex.\\ 
	For $w\in F^+$, assume $w=pw_1+(1-p)w_2$ with $w_1,w_2 \in A_+$ and $p\in (0,1)$. If $u_A(w) = 0$, then $w=0$ because of strict positivity of $u_A$. As $u_A(w_{1,2}) \ge 0$, thus $u_A(w_{1,2}) = 0$ and therefore $w_{1,2} = 0$, again because of strict positivity. If $u_A(w) \ne 0$, then $w':= \frac w {u_A(w)} \in F$. If $u_A(w_2) = 0$ (or analogously, $u_A(w_1) = 0$), then $w = p w_1$. Then $w_1 = \frac 1 p w \in \mathbb R_{\ge 0}F$ because $w\in \mathbb R_{\ge 0} F$. So assume now $u_A(w_2) \ne 0$ and $u_A(w_1) \ne 0$. Then 
	\begin{equation}
	  \frac{w}{u_A(w)} = \frac{p u_A(w_1)}{u_A(w)} \frac{w_1}{u_A(w_1)}+\frac{(1-p) u_A(w_2)}{u_A(w)} \frac{w_2}{u_A(w_2)} \in F
	\end{equation}
	As all the states are properly normalized and $p,1-p,u_A(w), u_A(w_1),u_A(w_2) > 0$, we find $1>\frac{p u_A(w_1)}{u_A(w)} > 0$, $\frac{p u_A(w_1)}{u_A(w)}+\frac{(1-p) u_A(w_1)}{u_A(w)} = 1$. As $F$ is a face of $\Omega_A$, $\frac{w_1}{u_A(w_1)}\in F$ and $\frac{w_2}{u_A(w_2)} \in F$. Therefore $w_1,w_2 \in \mathbb R_{\ge 0} \cdot F$. Thus $F^+$ is a face.\\
	\\
	Now let $F \ne \{0\}$ be a face of $A_+$. We now show that $G:= F \cap \Omega_A$ is a face of $\Omega_A$:\\
	$G$ is not empty, because there is a $v\in F$ with $v\ne 0$, i.e. $u_A(v)>0$. Thus $\frac{v}{u_A(v)} \in G$ by Lemma \ref{lemma:rays}. \\
	Let $w_1,w_2 \in G$, $p\in(0,1)$. Then $w_1,w_2 \in F$  and as $F$ is convex, $pw_1+(1-p)w_2 \in F$. By proper normalization and as properly normalized states of $A_+$ are elements of $\Omega_A$, $pw_1+(1-p)w_2 \in G$.\\
	Let $w \in G$, $p\in (0,1)$, $w_1,w_2 \in \Omega_A$ with $w= pw_1+(1-p)w_2$. By definition, also $w \in F$. As $F$ is a face, $w_1,w_2 \in F$. As $w_{1,2}\in \Omega_A$, in total $w_1,w_2 \in \Omega_A \cap F = G$. Thus G is a face.\\
	\\
	Now assume that the face $F\subset \Omega_A$ is generated by $w_1,...,w_m$. Then the minimal face
	\begin{equation*}
	  G := \bigcap_{H \subset A_+ \text{ face}, \ w_1,...,w_m \in H} H 
	\end{equation*}  
	of $A_+$ containing $w_1,...,w_m$ is a subset of the face $\mathbb R_{\ge 0}\cdot F$, i.e. $G \subset \mathbb R_{\ge 0}\cdot F$. If $G \ne \mathbb R_{\ge 0}\cdot F$, then $\exists w\in \mathbb R_{\ge 0}\cdot F$ but $w \notin G$. As all faces contain the $0$ and $u_A$ is strictly positive, $u_A(w) \ne 0$. The face $G \cap \Omega_A$ contains $w_1,...,w_m$. But $\frac{w}{u_A(w)} \in F$ is not found in $G\cap \Omega_A$, which is a contradiction to $F$ being minimal as a generated face.\\
	\\
	Now assume that the face $F\subset A_+$ is generated by $w_1,...,w_m \in \Omega_A$. Then $w_1,...,w_m$ are also found in the face $\Omega_A \cap F \subset \Omega_A$. Now consider the minimal face $G\subset \Omega_A$ containing the $w_1,...,w_m$. By definition, $G\subset F\cap \Omega_A$. If $G\ne F \cap \Omega_A$, then $\exists w \in F\cap \Omega_A$, but $w \notin G$. The face $\mathbb R_{\ge 0} \cdot G$ also contains the $w_1,...,w_m$, but not $w$. This is a contradiction to $F$ being minimal.\\
\\
	Now consider a face $F\subset \Omega_A$. We have $F \subset (\mathbb R_{\ge 0}\cdot F)$ and $F\subset \Omega_A$, thus $F\subset \Omega_A \cap (\mathbb R_{\ge 0}\cdot F)$. For $w \in \Omega_A \cap (\mathbb R_{\ge 0}\cdot F)$, we have $w\in \Omega_A$ and $\exists w' \in F$, $p \in \mathbb R_{\ge 0}$ with $w = p w'$. By $w\in \Omega_A$, we find $p=1$ because of normalization, thus $w\in F$  and in total $\Omega_A \cap (\mathbb R_{\ge 0}\cdot F) = F$.\\
\\
	Now consider a face $G\subset A_+$ with $G \ne \{0\}$. We have $(\Omega_A \cap G) \subset G$ and thus by Lemma \ref{lemma:rays} $\mathbb R_{\ge 0} \cdot (\Omega_A \cap G) \subset G$. Vice versa for $w \in G$: If $w \ne 0$, then $\frac{w}{u_A(w)} \in \Omega_A$ by $A_+ =\mathbb R_{\ge 0} \Omega_A$ and proper normalization. Thus $\frac{w}{u_A(w)} \in \Omega_A \cap G$ by Lemma \ref{lemma:rays}. Thus $w \in R_{\ge 0} \cdot (\Omega_A \cap G)$. As $G \ne \{0\}$, $\exists w' \in G:$ $u_A(w')> 0$. $\frac{w'}{u_A(w')} \in G \cap \Omega_A$, Thus $0 \in R_{\ge 0} \cdot (\Omega_A \cap G)$. In total, $G = R_{\ge 0} \cdot (\Omega_A \cap G)$.  
\end{proof} 

\section{Appendix: Entropy decreasing SWAP operation} 
\label{Section:SWAP}
In this appendix we explain the formal details of a quantum operation which replaces an arbitrary incoming quantum system by a pure state.\\ 
We consider two n-level Hilbert spaces $\mathcal H_a= \mathcal H_b$ with orthonormal base $\ket{1},...,\ket{n}$.\\
We define the linear map $U$ by $U(\ket{j}_a\otimes \ket{k}_b) := \ket{k}_a \otimes \ket{j}_b$ and extend linearly. This map is often called \textit{SWAP-operation}. It is unitary: For $\ket{\psi} = \sum_{j,k} p_{jk} \ket{j}_a\otimes \ket{k}_b$, $\ket{\phi} = \sum_{j,k} q_{jk} \ket{j}_a\otimes \ket{k}_b$ we find:
\begin{align*}
	\braket{U \cdot \psi |U \cdot \phi} &= \sum_{jk,rt} p_{jk}^* q_{rt} \big( \bra{k}_a\otimes \bra{j}_b\big) \big( \ket{t}_a\otimes \ket{r}_b\big) = \sum_{jk,rt} p_{jk}^* q_{rt} \delta_{kt} \delta_{jr} \\
	&= \sum_{jk,rt} p_{jk}^* q_{rt} \big( \bra{j}_a\otimes \bra{k}_b\big) \big( \ket{r}_a\otimes \ket{t}_b\big) = \braket{\psi | \phi}
\end{align*}
In particular, we find:
\begin{align*}
	U \big( \sum_{jk} \rho_{jk} \ket{j}\bra{k}_a \big) \otimes \ket{1}\bra{1}_b U^\dagger &= \sum_{jk} \rho_{jk} U \big( \ket{j}_a \otimes \ket{1}_b \big)  \big(\bra{k}_a\otimes\bra{1}_b\big) U^\dagger \\
	&= \sum_{jk} \rho_{jk} \big( \ket{1}_a \otimes \ket{j}_b \big) \big(\bra{1}_a\otimes\bra{k}_b\big) \\
	&= \ket{1}\bra{1}_a \otimes \big( \sum_{jk} \rho_{jk} \ket{j} \bra{k}_b \big)
\end{align*}
For the density operator $\rho_a$ on $\mathcal H_a$ we define the map:
\begin{equation*} T(\rho_a) := \text{Tr}_b\Big( U \rho_a \otimes \ket{1}\bra{1}_b U^\dagger \Big) \end{equation*}
In this form it is clear that this is a quantum transformation. The physical implementation is also clear: One starts with a n-level system in the state $\rho_a$. Then one creates another n-level system (of the same physical realization) initialized in the state $\ket{1}$. Then one exchanges the label of the two systems and throws away the original system, keeping only the system initialized to $\ket{1}$. This is clearly a experimentally possible transformation. But as the original system is simply forgotten and we just create any system of the same physical implementation in any state we like, it is clear that the transformation can decrease the entropy: $T(\frac{1}{n} \sum_j \ket{j}\bra{j}_a) = \ket{1}\bra{1}_a$, thus $S(\frac{1}{n} \sum_j \ket{j}\bra{j}_a) = -\sum_j \frac{1}{n} \ln \left( \frac{1}{n}\right) =\ln(n) > 0 = S(\ket{1}\bra{1}_a) = S(T(\frac{1}{n} \sum_j \ket{j}\bra{j}_a))$.\\
We note that this transformation does not change the normalization, thus it simply induces the order unit as measurement. Thus in most generality, there will be operations that can decrease entropy.
\section{Appendix: Mathematica shows that weak spectrality in the 2D egg-shape holds}
\subsection{Appendix: States with classical decomposition}
\label{chap:WSMathematica}
\hspace{-20mm}
\includegraphics[height=0.8\textheight, trim = 10mm 70mm 0mm 25mm, clip]{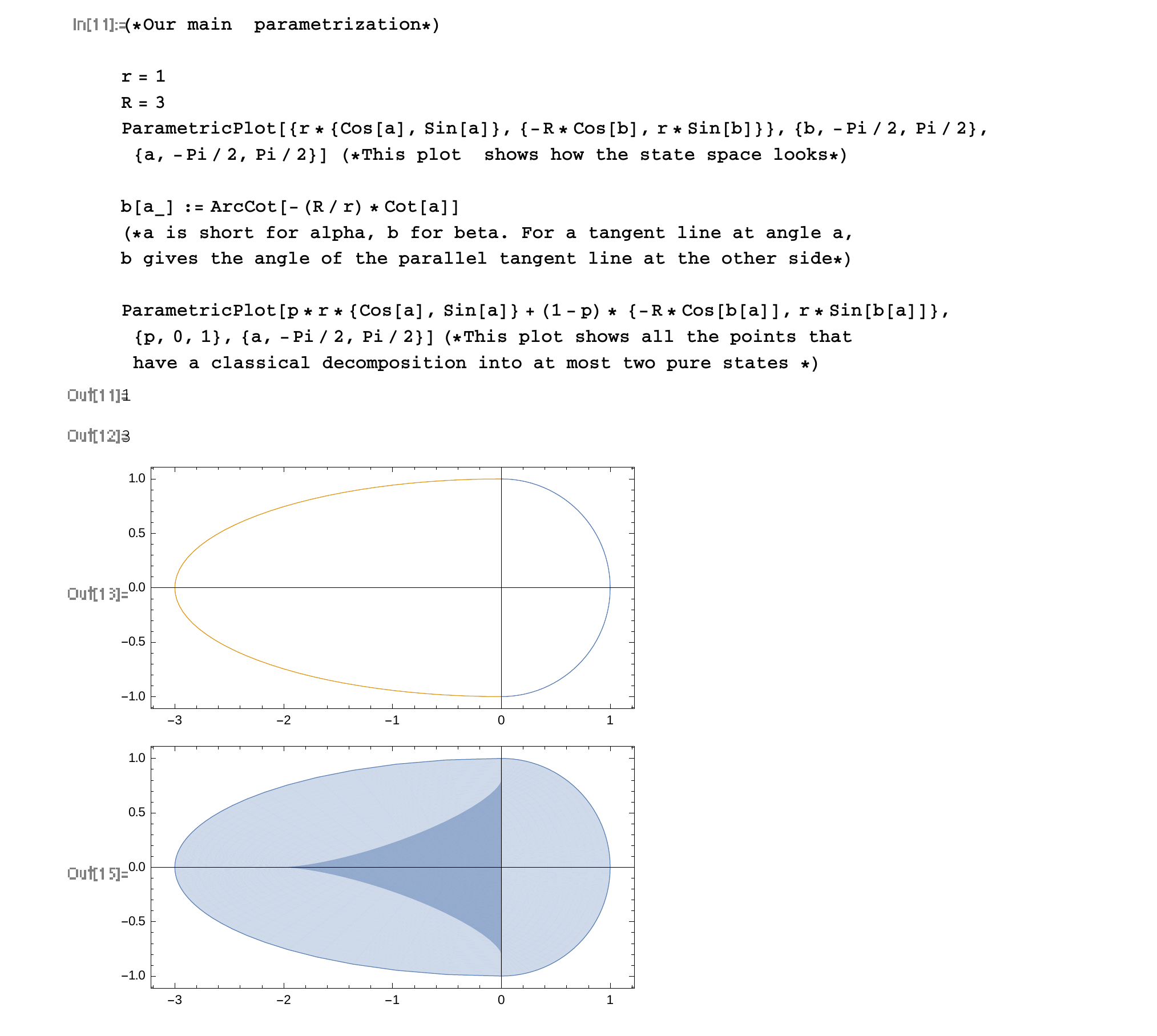}

\subsection{Appendix: Alternative parametrization}
\label{chap:AlternativeParam}
\begin{center}
\includegraphics{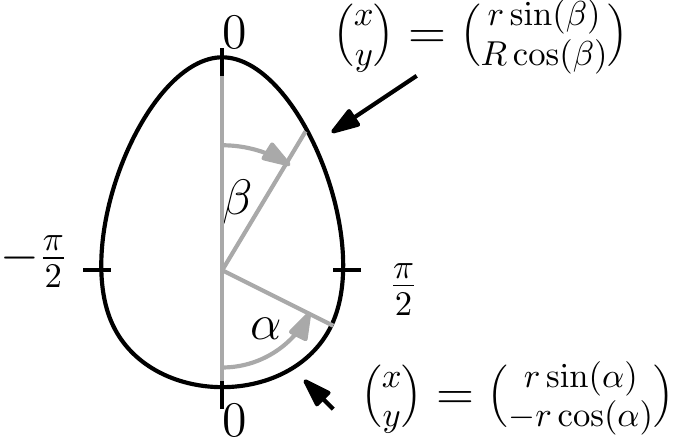}
\captionof{figure}{\small \textit{Detailed picture of how we parametrize the egg.}}
\end{center}
So now we give an alternative parametrization of the 2D egg:
The first part is a half-circle, given by 
\begin{equation}
	(x,y) = (r\sin(\alpha),-r\cos(\alpha))
\end{equation}
for $\alpha \in [-\frac \pi 2, \frac{\pi} 2 ]$. The second part is an ellipse given by 
\begin{equation}
	(x,y) = (r\sin(\beta), R\cos(\beta))
\end{equation}
for $\beta \in [-\frac \pi 2, \frac{\pi} 2 ]$. Choose an arbitrary angle $\alpha$. Then
\begin{align}
	\frac{\mathrm d y}{\mathrm d \alpha} = r \sin(\alpha)
	&&\frac{\mathrm d x}{\mathrm d \alpha} = r \cos(\alpha)
\end{align}
Thus the slope of the tangent line is given by $\frac{\mathrm d y}{\mathrm d x} = \tan(\alpha)$.
Now we want to find the parallel tangent line on the other side. At first
\begin{align}
	\frac{\mathrm d y}{\mathrm d \beta} = -R \sin(\beta)
	&&\frac{\mathrm d x}{\mathrm d \beta} = r \cos(\beta)
\end{align}
Thus the slope on the other side is given by $\frac{\mathrm d y}{\mathrm d x} = -\frac{R}{r} \tan(\beta)$. We want the slopes to agree, i.e. $-\frac{R}{r} \tan(\beta) = \tan(\alpha)$. Thus 
\begin{equation}
	\beta(\alpha) = \text{arctan}\left(-\frac{r}{R} \tan{\alpha}\right)
\end{equation}
The line connecting $p_1$ and $p_2$ is 
\begin{equation}
	l_\alpha(p) = p\begin{pmatrix}r\sin(\alpha)\\ -r\cos(\alpha) \end{pmatrix} +(1-p) \begin{pmatrix}r\sin(\beta(\alpha)) \\ R\cos(\beta(\alpha)) \end{pmatrix}
\end{equation}
where $p \in [0,1]$. The next page uses Mathematica to show that weak spectrality / classical decomposability holds using this parametrization.
\includepdf[pages={1}]{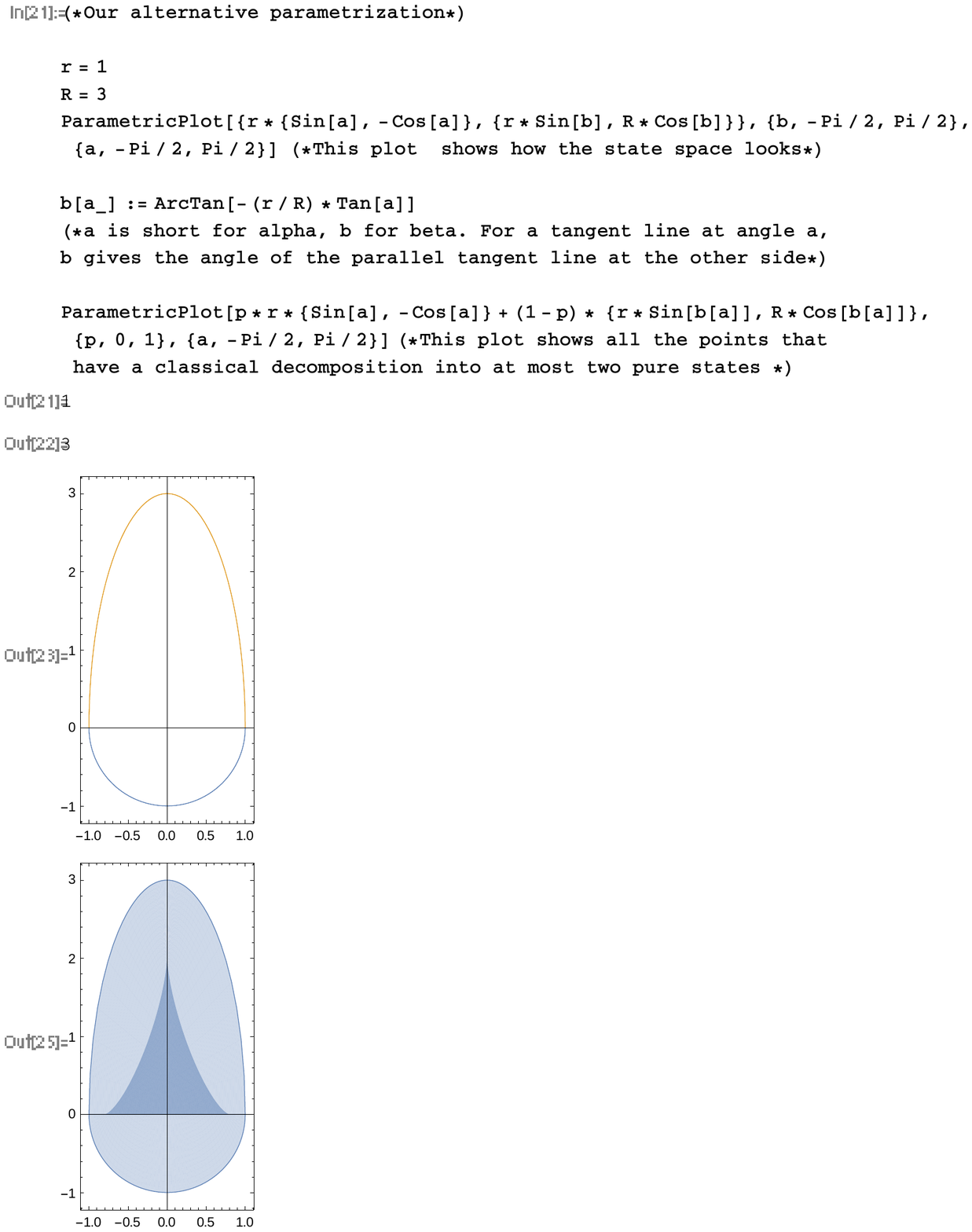} 

\subsection{Appendix: Visualization of the tangent lines with Mathematica}
\label{chap:Rotation}
\vspace{-20mm}
\hspace{-25mm}
\includegraphics[width=1.14\textwidth]{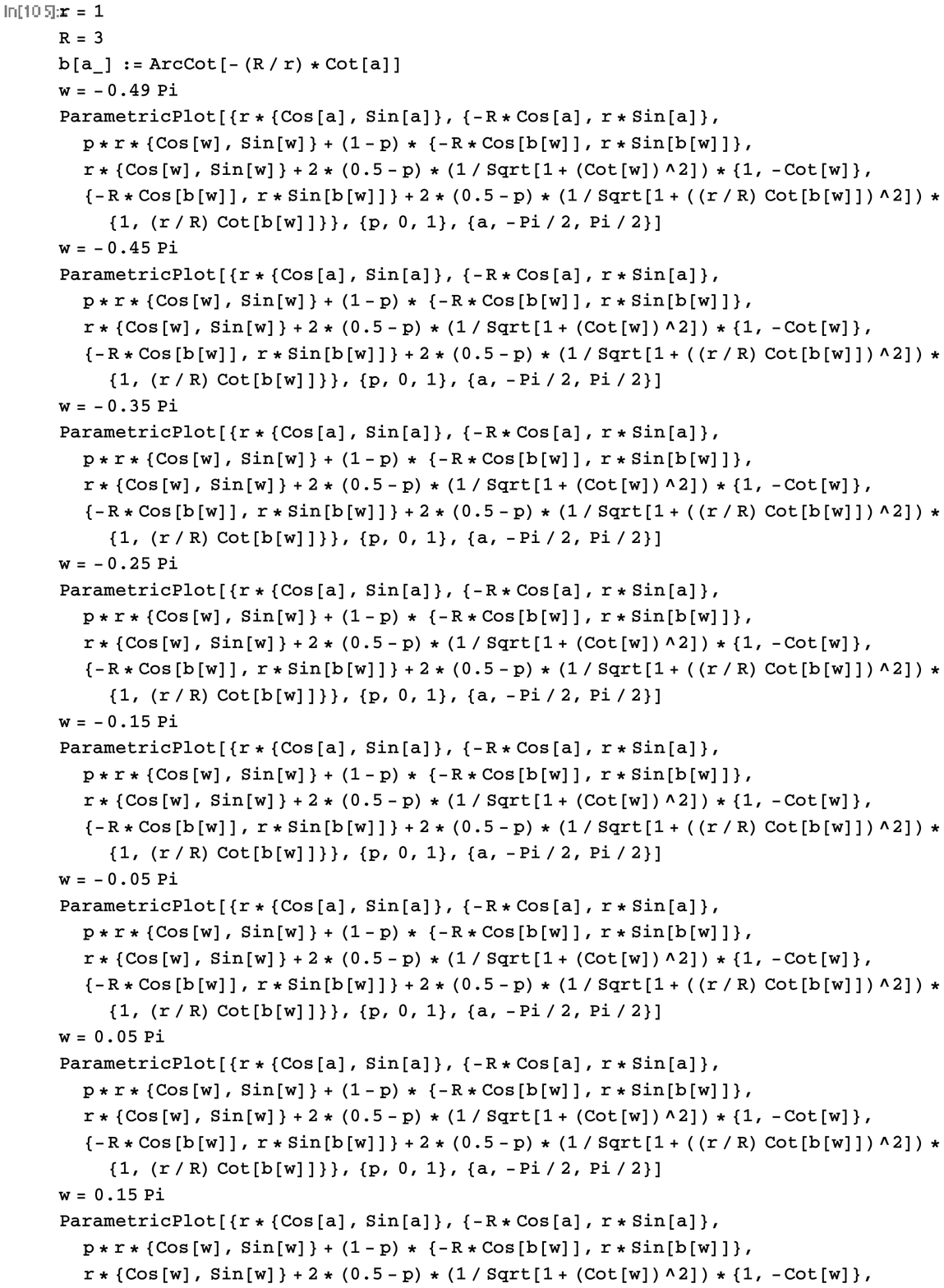}


\section*{Supplementary material}

```
In[21]:= (*Our alternative parametrization*)

r = 1
R = 3
ParametricPlot[{r * {Sin[a], -Cos[a]}, {r * Sin[b], R * Cos[b]}}, {b, -Pi / 2, Pi / 2},
  {a, -Pi / 2, Pi / 2}] (*This plot  shows how the state space looks*)

b[a_] := ArcTan[-(r / R) * Tan[a]]
(*a is short for alpha, b for beta. For a tangent line at angle a,
b gives the angle of the parallel tangent line at the other side*)

ParametricPlot[p * r * {Sin[a], -Cos[a]} + (1 - p) * {r * Sin[b[a]], R * Cos[b[a]]},
  {p, 0, 1}, {a, -Pi / 2, Pi / 2}] (*This plot shows all the points that
  have a classical decomposition into at most two pure states *)
```

Out[21]= 1

Out[22]= 3

Out[23]=
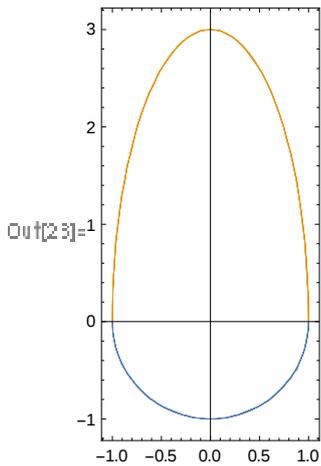

Out[25]=
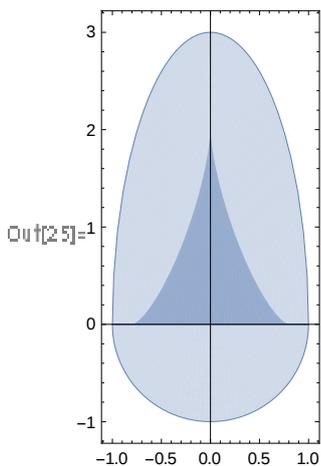

## D.3  Appendix: Visualization of the tangent lines with Mathematica

```mathematica
In[105]:= r = 1
R = 3
b[a_] := ArcCot[-(R/r)*Cot[a]]
w = -0.49 Pi
ParametricPlot[{r*{Cos[a], Sin[a]}, {-R*Cos[a], r*Sin[a]},
   p*r*{Cos[w], Sin[w]} + (1-p)*{-R*Cos[b[w]], r*Sin[b[w]]},
   r*{Cos[w], Sin[w]} + 2*(0.5-p)*(1/Sqrt[1+(Cot[w])^2])*{1, -Cot[w]},
   {-R*Cos[b[w]], r*Sin[b[w]]} + 2*(0.5-p)*(1/Sqrt[1+((r/R) Cot[b[w]])^2])*
     {1, (r/R) Cot[b[w]]}}, {p, 0, 1}, {a, -Pi/2, Pi/2}]
w = -0.45 Pi
ParametricPlot[{r*{Cos[a], Sin[a]}, {-R*Cos[a], r*Sin[a]},
   p*r*{Cos[w], Sin[w]} + (1-p)*{-R*Cos[b[w]], r*Sin[b[w]]},
   r*{Cos[w], Sin[w]} + 2*(0.5-p)*(1/Sqrt[1+(Cot[w])^2])*{1, -Cot[w]},
   {-R*Cos[b[w]], r*Sin[b[w]]} + 2*(0.5-p)*(1/Sqrt[1+((r/R) Cot[b[w]])^2])*
     {1, (r/R) Cot[b[w]]}}, {p, 0, 1}, {a, -Pi/2, Pi/2}]
w = -0.35 Pi
ParametricPlot[{r*{Cos[a], Sin[a]}, {-R*Cos[a], r*Sin[a]},
   p*r*{Cos[w], Sin[w]} + (1-p)*{-R*Cos[b[w]], r*Sin[b[w]]},
   r*{Cos[w], Sin[w]} + 2*(0.5-p)*(1/Sqrt[1+(Cot[w])^2])*{1, -Cot[w]},
   {-R*Cos[b[w]], r*Sin[b[w]]} + 2*(0.5-p)*(1/Sqrt[1+((r/R) Cot[b[w]])^2])*
     {1, (r/R) Cot[b[w]]}}, {p, 0, 1}, {a, -Pi/2, Pi/2}]
w = -0.25 Pi
ParametricPlot[{r*{Cos[a], Sin[a]}, {-R*Cos[a], r*Sin[a]},
   p*r*{Cos[w], Sin[w]} + (1-p)*{-R*Cos[b[w]], r*Sin[b[w]]},
   r*{Cos[w], Sin[w]} + 2*(0.5-p)*(1/Sqrt[1+(Cot[w])^2])*{1, -Cot[w]},
   {-R*Cos[b[w]], r*Sin[b[w]]} + 2*(0.5-p)*(1/Sqrt[1+((r/R) Cot[b[w]])^2])*
     {1, (r/R) Cot[b[w]]}}, {p, 0, 1}, {a, -Pi/2, Pi/2}]
w = -0.15 Pi
ParametricPlot[{r*{Cos[a], Sin[a]}, {-R*Cos[a], r*Sin[a]},
   p*r*{Cos[w], Sin[w]} + (1-p)*{-R*Cos[b[w]], r*Sin[b[w]]},
   r*{Cos[w], Sin[w]} + 2*(0.5-p)*(1/Sqrt[1+(Cot[w])^2])*{1, -Cot[w]},
   {-R*Cos[b[w]], r*Sin[b[w]]} + 2*(0.5-p)*(1/Sqrt[1+((r/R) Cot[b[w]])^2])*
     {1, (r/R) Cot[b[w]]}}, {p, 0, 1}, {a, -Pi/2, Pi/2}]
w = -0.05 Pi
ParametricPlot[{r*{Cos[a], Sin[a]}, {-R*Cos[a], r*Sin[a]},
   p*r*{Cos[w], Sin[w]} + (1-p)*{-R*Cos[b[w]], r*Sin[b[w]]},
   r*{Cos[w], Sin[w]} + 2*(0.5-p)*(1/Sqrt[1+(Cot[w])^2])*{1, -Cot[w]},
   {-R*Cos[b[w]], r*Sin[b[w]]} + 2*(0.5-p)*(1/Sqrt[1+((r/R) Cot[b[w]])^2])*
     {1, (r/R) Cot[b[w]]}}, {p, 0, 1}, {a, -Pi/2, Pi/2}]
w = 0.05 Pi
ParametricPlot[{r*{Cos[a], Sin[a]}, {-R*Cos[a], r*Sin[a]},
   p*r*{Cos[w], Sin[w]} + (1-p)*{-R*Cos[b[w]], r*Sin[b[w]]},
   r*{Cos[w], Sin[w]} + 2*(0.5-p)*(1/Sqrt[1+(Cot[w])^2])*{1, -Cot[w]},
   {-R*Cos[b[w]], r*Sin[b[w]]} + 2*(0.5-p)*(1/Sqrt[1+((r/R) Cot[b[w]])^2])*
     {1, (r/R) Cot[b[w]]}}, {p, 0, 1}, {a, -Pi/2, Pi/2}]
w = 0.15 Pi
ParametricPlot[{r*{Cos[a], Sin[a]}, {-R*Cos[a], r*Sin[a]},
   p*r*{Cos[w], Sin[w]} + (1-p)*{-R*Cos[b[w]], r*Sin[b[w]]},
   r*{Cos[w], Sin[w]} + 2*(0.5-p)*(1/Sqrt[1+(Cot[w])^2])*{1, -Cot[w]},
```

```
      {-R*Cos[b[w]], r*Sin[b[w]]} + 2*(0.5-p)*(1/Sqrt[1+((r/R) Cot[b[w]])^2])*
        {1, (r/R) Cot[b[w]]}}, {p, 0, 1}, {a, -Pi/2, Pi/2}]
   w = 0.25 Pi
   ParametricPlot[{r*{Cos[a], Sin[a]}, {-R*Cos[a], r*Sin[a]},
      p*r*{Cos[w], Sin[w]} + (1-p)*{-R*Cos[b[w]], r*Sin[b[w]]},
      r*{Cos[w], Sin[w]} + 2*(0.5-p)*(1/Sqrt[1+(Cot[w])^2])*{1, -Cot[w]},
      {-R*Cos[b[w]], r*Sin[b[w]]} + 2*(0.5-p)*(1/Sqrt[1+((r/R) Cot[b[w]])^2])*
        {1, (r/R) Cot[b[w]]}}, {p, 0, 1}, {a, -Pi/2, Pi/2}]
   w = 0.35 Pi
   ParametricPlot[{r*{Cos[a], Sin[a]}, {-R*Cos[a], r*Sin[a]},
      p*r*{Cos[w], Sin[w]} + (1-p)*{-R*Cos[b[w]], r*Sin[b[w]]},
      r*{Cos[w], Sin[w]} + 2*(0.5-p)*(1/Sqrt[1+(Cot[w])^2])*{1, -Cot[w]},
      {-R*Cos[b[w]], r*Sin[b[w]]} + 2*(0.5-p)*(1/Sqrt[1+((r/R) Cot[b[w]])^2])*
        {1, (r/R) Cot[b[w]]}}, {p, 0, 1}, {a, -Pi/2, Pi/2}]
   w = 0.45 Pi
   ParametricPlot[{r*{Cos[a], Sin[a]}, {-R*Cos[a], r*Sin[a]},
      p*r*{Cos[w], Sin[w]} + (1-p)*{-R*Cos[b[w]], r*Sin[b[w]]},
      r*{Cos[w], Sin[w]} + 2*(0.5-p)*(1/Sqrt[1+(Cot[w])^2])*{1, -Cot[w]},
      {-R*Cos[b[w]], r*Sin[b[w]]} + 2*(0.5-p)*(1/Sqrt[1+((r/R) Cot[b[w]])^2])*
        {1, (r/R) Cot[b[w]]}}, {p, 0, 1}, {a, -Pi/2, Pi/2}]
   w = 0.49 Pi
   ParametricPlot[{r*{Cos[a], Sin[a]}, {-R*Cos[a], r*Sin[a]},
      p*r*{Cos[w], Sin[w]} + (1-p)*{-R*Cos[b[w]], r*Sin[b[w]]},
      r*{Cos[w], Sin[w]} + 2*(0.5-p)*(1/Sqrt[1+(Cot[w])^2])*{1, -Cot[w]},
      {-R*Cos[b[w]], r*Sin[b[w]]} + 2*(0.5-p)*(1/Sqrt[1+((r/R) Cot[b[w]])^2])*
        {1, (r/R) Cot[b[w]]}}, {p, 0, 1}, {a, -Pi/2, Pi/2}]
```

Out[105]= 1

Out[106]= 3

Out[108]= 1.53938

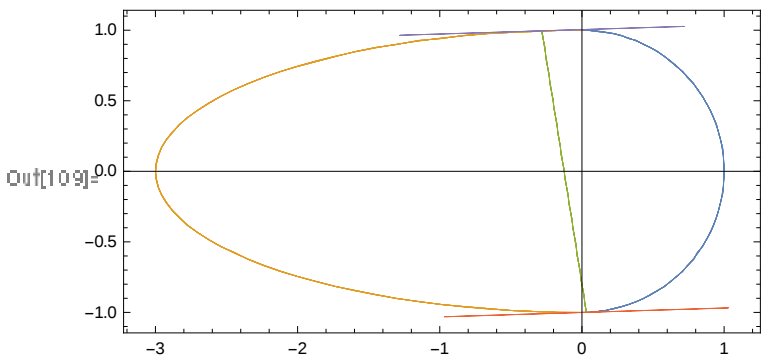

Out[110]= 1.41372

Out[111]=

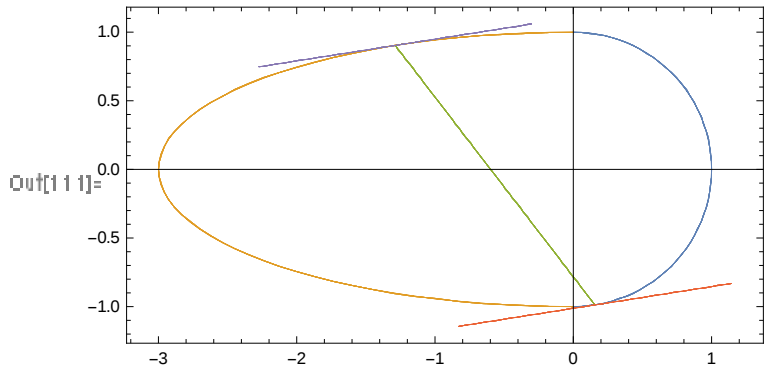

Out[112]= 1.09956

Out[113]=

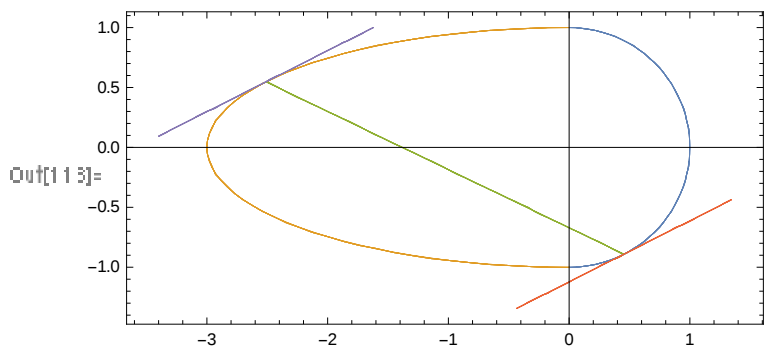

Out[114]= 0.785398

Out[115]=

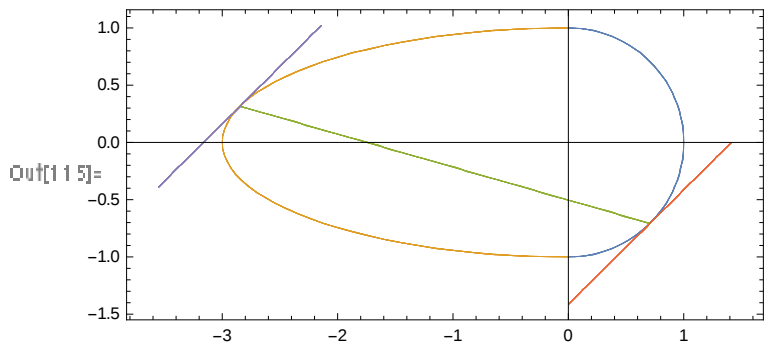

Out[116]= 0.471239

Out[117]=

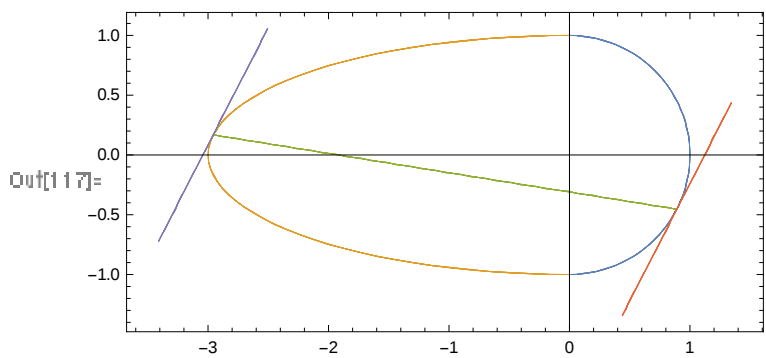



\newpage
\thispagestyle{empty}
\quad
\newpage
\thispagestyle{empty}

\section*{Acknowledgments}

I thank Dr. Markus P. Müller for the possibility to write a Master thesis in the fields of quantum information theory and quantum foundations. His guidance helped me very much and his ideas often closed the missing link to new proofs and new theorems. But more importantly, my fascination for the information-theoretic insights into quantum theory are mainly due to his expertise and I feel lucky for working with him.\\
\\
I also thank Jonathan Barrett for suggesting to adapt von Neumann's thought experiment for GPTs, without his idea this thesis would not exist. Furthermore, I am grateful for his hospitality during our stay in Oxford.\\
\\
Furthermore, I thank Howard Barnum for very useful discussions and fruitful collaboration while working with our group. In particular, Howard had several ideas for further projects connected to our current work.\\
\\
I thank my family and my friends for supporting me in hard times. It is good to know that there exist people who accept me no matter whether I fail or succeed.\\
\\
Also, I thank the people at the institute for theoretical physics for the good atmosphere.\\
\\
I thank Prof. Dr. Manfred Salmhofer for agreeing to read and grade my thesis.\\
\\
At last but not least, I thank all the researchers from the quantum foundations and quantum information community. Many surprising relations between information-theoretic principles and physics have been discovered and I am convinced that the very core of physics and reality is an information-theoretic principle.\\

\cleardoublepage

\thispagestyle{empty}

\setlength{\parindent}{0em}

Erklärung:\par
\vspace{3\baselineskip}
Ich versichere, dass ich diese Arbeit selbstständig verfasst habe und keine
anderen als die angegebenen Quellen und Hilfsmittel benutzt habe.\par
\vspace{5\baselineskip}
Heidelberg, den \quad .\quad .2015\hspace{3cm}\dotfill

\end{document}